\documentclass[11pt]{llncs-overlay}
\let\proof\relax
\let\endproof\relax
\usepackage{amsthm} %
\usepackage{fullpage}
\usepackage{lmodern}

\crefname{conjecture}{Conjecture}{Conjectures}

\title{Towards Unclonable Cryptography in the Plain Model}
\author{Céline Chevalier\inst{1,2} \and Paul Hermouet\inst{1,2,3} \and Quoc-Huy
	Vu\inst{4}}
\institute{DIENS, École normale supérieure, PSL University, CNRS, INRIA, Paris, France
	\and CRED, Université Panthéon-Assas Paris II, Paris, France
	\and LIP6, Sorbonne Université, Paris, France
	\and Léonard de Vinci Pôle Universitaire, Research Center, Paris La Défense, France\\
	\email{\{celine.chevalier, paul.hermouet, quoc.huy.vu\}@ens.fr}}
\let\oldmaketitle\maketitle
\renewcommand{\maketitle}{\oldmaketitle\setcounter{footnote}{0}}

\makeatletter%
\begin{document}

\maketitle

\begin{abstract}
  By leveraging the no-cloning principle of quantum mechanics, unclonable
  cryptography enables us to achieve novel cryptographic protocols that are
  otherwise impossible classically.
  Two most notable examples of unclonable cryptography are quantum
  copy-protection and unclonable encryption.
  Most known constructions rely on the quantum random oracle model (as opposed
  to the plain model), in which all parties have access in superposition to a
  powerful random oracle.
  Despite receiving a lot of attention in recent years, two important open
  questions still remain: copy-protection for point functions in the plain
  model, which is usually considered as feasibility demonstration, and
  unclonable encryption with unclonable indistinguishability security in the
  plain model.
  A core ingredient of these protocols is the so-called monogamy-of-entanglement
  property.
  Such games allow quantifying the correlations between the outcomes of multiple
  non-communicating parties sharing entanglement in a particular context.
  Specifically, we define the games between a challenger and three players in
  which the first player is asked to split and share a quantum state between the
  two others, who are then simultaneously asked a question and need to output
  the correct answer.

  In this work, by relying on previous works of Coladangelo, Liu, Liu, and
  Zhandry (Crypto'21) and Culf and Vidick (Quantum'22), we establish a new
  monogamy-of-entanglement property for subspace coset states, which allows us
  to progress towards the aforementioned goals.
  However, it is not sufficient on its own, and we present two conjectures that
  would allow first to show that copy-protection of point functions exists in
  the plain model, with different challenge distributions (including arguably
  the most natural ones), and then that unclonable encryption with unclonable
  indistinguishability security exists in the plain model.

  We believe that our new monogamy-of-entanglement to be of independent interest,
  and it could be useful in other applications as well.
  To highlight this last point, we leverage our new monogamy-of-entanglement property to show the existence of a tokenized signature scheme with a new security definition, called unclonable unforgeability.
\end{abstract}

\newcommand{\pnote}[1]{\textcolor{magenta}{\small {\textbf{(Paul:} #1\textbf{)}}}}
\newcommand{\ptodo}[1]{\textcolor{magenta}{\small {\textbf{(Paul: TODO} #1\textbf{)}}}}
\newcommand{\pttodo}[1]{\textcolor{magenta}{\small {\textbf{(Paul: ! TODO !} #1\textbf{)}}}}
\newcommand{\ptttodo}[1]{\textcolor{magenta}{\small {\textbf{(Paul: !!! TODO !!!} #1\textbf{)}}}}
\newcommand{\hnote}[1]{\textcolor{magenta}{\small {\textbf{(Huy:} #1\textbf{)}}}}
\newcommand{\red}[1]{\textcolor{red}{#1}}

\newcommand{\algofont}{\mathsf}
\newcommand{\advfont}{\mathcal}
\newcommand{\progfont}{\mathsf}
\newcommand{\setfont}{\mathcal}
\newcommand{\schemefont}{\mathsf}
\newcommand{\pr}{\textnormal{Pr}}
\newcommand{\uniform}{\mathcal{U}}

\newtheorem{_assumption}{Assumption}
\newenvironment{assumption}{
  \color{red}
  \begin{_assumption}
}
{
  \end{_assumption}
}

\newcommand{\intval}[1]{\llbracket #1 \rrbracket}
\newcommand{\trace}{\mathsf{Tr}}
\newcommand{\expectop}{\mathop{\expectationname}}
\newcommand{\p}{\textsf{P}}

\newcommand{\shO}{\mathsf{shO}}

\newcommand{\commit}{\algofont{Commit}}
\newcommand{\setup}{\algofont{Setup}}
\newcommand{\keygen}{\algofont{KeyGen}}
\newcommand{\cpa}{\mathsf{CPA}}

\renewcommand{\verify}{\algofont{Verify}}

\newcommand{\chall}{\mathsf{Challenger}}
\newcommand{\alice}{\mathsf{Alice}}
\newcommand{\bob}{\mathsf{Bob}}
\newcommand{\charlie}{\mathsf{Charlie}}
\newcommand{\alex}{\mathsf{Alex}}
\newcommand{\clover}{\mathsf{Clover}}
\newcommand{\oracle}{\mathcal{O}}
\newcommand{\adverA}[1]{\advfont{A}_#1}
\newcommand{\adverB}[1]{\advfont{B}_#1}

\newcommand{\kset}{\setfont{K}}
\newcommand{\mset}{\setfont{M}}
\newcommand{\rset}{\setfont{R}}
\newcommand{\cset}{\setfont{C}}
\newcommand{\xset}{\mathcal{X}}
\newcommand{\yset}{\mathcal{Y}}
\newcommand{\zset}{\mathcal{Z}}
\newcommand{\kspace}{\kset}
\renewcommand{\mspace}{\mset}
\newcommand{\rspace}{\rset}
\newcommand{\cspace}{\cset}
\newcommand{\xspace}{\xset}
\newcommand{\yspace}{\yset}
\newcommand{\zspace}{\zset}

\newcommand{\pf}{\mathsf{PF}}
\newcommand{\pke}{\mathsf{PKE}}

\newcommand{\nce}{\mathsf{NCE}}
\newcommand{\fake}{\mathsf{Fake}}
\newcommand{\reveal}{\mathsf{Reveal}}

\newcommand{\scheme}[1]{\langle#1\rangle}
\newcommand{\aux}{\mathsf{aux}}
\newcommand{\distrib}{\advfont{D}}
\newcommand{\ccprog}{\algofont{CC}}
\newcommand{\ccobf}{\algofont{CC\text{-}Obf}}
\newcommand{\simul}{\algofont{Sim}}
\newcommand{\ketbra}[2]{{\vert#1\rangle\!\langle#2\vert}}
\newcommand{\qpt}{\textsf{QPT}}
\newcommand{\bqp}{\textsf{BQP}}

\newcommand{\var}[1]{\ensuremath{{\operatorname{#1}}}}
\newcommand{\hgate}{\mathsf{H}}
\newcommand{\Hgate}{\mathsf{H}}
\newcommand{\CZgate}{\mathsf{Ctrl\var{-}Z}}
\newcommand{\CPgate}{\mathsf{Ctrl}\var{-}\mathsf{P}}
\newcommand{\Measurement}{\mathsf{M}}
\newcommand{\Xgate}{\mathsf{X}}
\newcommand{\Ygate}{\mathsf{Y}}
\newcommand{\Zgate}{\mathsf{Z}}
\newcommand{\CUgate}{\mathsf{Ctrl\var{-}U}}
\newcommand{\Rgate}{\mathsf{R}}
\newcommand{\Tgate}{\mathsf{T}}
\newcommand{\Pgate}{\mathsf{P}}
\newcommand{\Ugate}{\mathsf{U}}
\newcommand{\CNOTgate}{\mathsf{CNOT}}
\newcommand{\CCNOTgate}{\mathsf{CCNOT}}

\newcommand{\registerfont}{\mathsf}
\newcommand{\cpro}{\schemefont{CP}}
\newcommand{\cpprotect}{\algofont{Protect}}
\newcommand{\qkeygen}{\algofont{QKeyGen}}
\newcommand{\sd}{\algofont{SD}}

\newcommand{\del}{\algofont{Del}}
\renewcommand{\verify}{\algofont{Verify}}
\newcommand{\crt}{\algofont{crt}}

\newcommand{\tokengen}{\algofont{TokenGen}}

\makeatletter
\def\squarecorner#1{
    \pgf@x=\the\wd\pgfnodeparttextbox%
    \pgfmathsetlength\pgf@xc{\pgfkeysvalueof{/pgf/inner xsep}}%
    \advance\pgf@x by 2\pgf@xc%
    \pgfmathsetlength\pgf@xb{\pgfkeysvalueof{/pgf/minimum width}}%
    \ifdim\pgf@x<\pgf@xb%
        \pgf@x=\pgf@xb%
    \fi%
    \pgf@y=\ht\pgfnodeparttextbox%
    \advance\pgf@y by\dp\pgfnodeparttextbox%
    \pgfmathsetlength\pgf@yc{\pgfkeysvalueof{/pgf/inner ysep}}%
    \advance\pgf@y by 2\pgf@yc%
    \pgfmathsetlength\pgf@yb{\pgfkeysvalueof{/pgf/minimum height}}%
    \ifdim\pgf@y<\pgf@yb%
        \pgf@y=\pgf@yb%
    \fi%
    \ifdim\pgf@x<\pgf@y%
        \pgf@x=\pgf@y%
    \else
        \pgf@y=\pgf@x%
    \fi
    \pgf@x=#1.5\pgf@x%
    \advance\pgf@x by.5\wd\pgfnodeparttextbox%
    \pgfmathsetlength\pgf@xa{\pgfkeysvalueof{/pgf/outer xsep}}%
    \advance\pgf@x by#1\pgf@xa%
    \pgf@y=#1.5\pgf@y%
    \advance\pgf@y by-.5\dp\pgfnodeparttextbox%
    \advance\pgf@y by.5\ht\pgfnodeparttextbox%
    \pgfmathsetlength\pgf@ya{\pgfkeysvalueof{/pgf/outer ysep}}%
    \advance\pgf@y by#1\pgf@ya%
}
\makeatother

\pgfdeclareshape{square}{
    \savedanchor\northeast{\squarecorner{}}
    \savedanchor\southwest{\squarecorner{-}}

    \foreach \x in {east,west} \foreach \y in {north,mid,base,south} {
        \inheritanchor[from=rectangle]{\y\space\x}
    }
    \foreach \x in {east,west,north,mid,base,south,center,text} {
        \inheritanchor[from=rectangle]{\x}
    }
    \inheritanchorborder[from=rectangle]
    \inheritbackgroundpath[from=rectangle]
}

\newcommand{\nottouchingarrow}[2]{($ (#1)!1.0/10!(#2) $) -- ($ (#1)!9.0/10!(#2) $)}
\newcommand{\nottouchingarrowcenternode}[4]{($ (#1)!1.0/10!(#2) $) -- node[#3] {#4} ($ (#1)!9.0/10!(#2) $)} %
\newcommand{\ptriv}{p^{triv}}
\newcommand{\indcor}{\mathsf{IND\text{--}CoR}}
\newcommand{\inddor}{\mathsf{IND\text{--}DoR}}
\newcommand{\genTrigger}{\progfont{GenTrigger}}
\newcommand{\funcfont}{\mathsf}
\newcommand{\mbr}{\progfont{P}}
\newcommand{\obfmbr}{\widehat{\mbr}}
\newcommand{\Qmr}{\progfont{Q_{m, r}}}
\newcommand{\Qmrp}{\progfont{Q_{m, r'}}}
\newcommand{\can}{\progfont{Can}}
\newcommand{\rmPi}{\mathrm{\Pi}}
\newcommand{\rmPhi}{\mathrm{\Phi}}
\newcommand{\thetaset}{\mathrm{\Theta}}
\newcommand{\id}{\mathbb{I}}
\newcommand{\pbar}{\bar{\mathrm{P}}}
\newcommand{\qbar}{\bar{\mathrm{Q}}}
\newcommand{\ie}{i.e.}
\newcommand{\hyb}{\mathsf{H}}
\newcommand{\xistrib}{\mathcal{X}}
\newcommand{\family}{\mathcal{F}}
\newcommand{\pirate}{\mathcal{P}}
\newcommand{\freeloader}{\mathcal{F}}
\newcommand{\encode}{\algofont{Encode}}
\newcommand{\qkey}{\rho}
\newcommand{\crcp}{\algofont{CRCP}}
\newcommand{\crprf}{\algofont{CRPRF}}
\newcommand{\ti}{\algofont{TI}}
\newcommand{\params}{\mathsf{param}}
\newcommand{\subexp}{\mathsf{subexp}}

\DeclarePairedDelimiter\size{\lvert}{\rvert}%
\DeclarePairedDelimiterX{\inp}[2]{\langle}{\rangle}{#1, #2}

\newcommand{\gameprefix}{\textbf{Game}}
\newcommand{\gamesprefix}{\textbf{Games}}
\newcommand{\hybridprefix}{\textbf{Hybrid}}
\renewcommand{\pcgamename}{\ensuremath{G}}
\newcommand{\pchybridname}{\ensuremath{H}}
\newcommand{\removeindent}{\parshape %
  2 %
  15pt \dimexpr\linewidth-15pt %
  \parindent \dimexpr\linewidth-\parindent %
}

\makeatletter
\renewenvironment{gamedescription}[1][]{%
\begingroup%
\setkeys{pcgameproof}{#1}
\@pc@ensureremember%
\setcounter{pcgamecounter}{\@pcgameproofgamenr}%
\setcounter{pcstartgamecounter}{\@pcgameproofgamenr}\stepcounter{pcstartgamecounter}%
\begin{enumerate}[align=left, leftmargin=0pt, labelindent=0pt,listparindent=\parindent, labelwidth=0pt, itemindent=!,itemsep=6pt]%
}{\end{enumerate}\@pc@releaseremember\endgroup}
\makeatother

\renewcommand{\describegame}{%
\addtocounter{pcgamecounter}{1}%
\item[%
\gameprefix~\ensuremath{\pcgamename_{\thepcgamecounter}\gameprocedurearg}:]%
}

\newcommand{\describecgame}[1]{%
  \item[%
  \gameprefix~\ensuremath{\pcgamename_{\thepcgamecounter,\text{{#1}}}\gameprocedurearg}:]%
}

\newcommand{\describehybrid}{%
\addtocounter{pcgamecounter}{1}%
\item[%
\hybridprefix~\ensuremath{\pchybridname_{\thepcgamecounter}\gameprocedurearg}:]%
}

\newcommand{\describechybrid}[1]{%
\addtocounter{pcgamecounter}{1}%
\item[%
\hybridprefix~\ensuremath{\pchybridname_{(\thepcgamecounter,\text{#1})}\gameprocedurearg}:]%
}

\newcommand{\describegames}[1]{%
\addtocounter{pcgamecounter}{1}%
\item[%
\gamesprefix~\ensuremath{\pcgamename_{\thepcgamecounter\var{-}#1}\gameprocedurearg}:]%
\setcounter{pcgamecounter}{#1}
}

\newcommand{\game}{\pcgamename}
\newcommand{\hybrid}{\pchybridname}
\newcommand{\explaingame}{\smallbreak\removeindent} %

\newcommand{\currentgamenbr}{\thepcgamecounter}
\newcommand{\currentgame}{\pcgamename_{\thepcgamecounter}}
\newcommand{\previousgamenbr}{\addtocounter{pcgamecounter}{-1}\arabic{pcgamecounter}\addtocounter{pcgamecounter}{1}}
\newcommand{\previousgame}{\pcgamename_{\previousgamenbr}}
\makeatletter
\newcommand{\gamelabel}[1]{\def\@currentlabel{\pcgamename_{\currentgamenbr}}\label{#1}}
\makeatother
\newcommand{\gamebox}[1]{{\setlength{\fboxsep}{1pt}\fbox{\ifmmode$\displaystyle#1$\else#1\fi}}}
\newcommand{\gamehighlight}[2][gamechangecolor]{%
{\setlength{\fboxsep}{2pt}\colorbox{#1}{\ifmmode$\displaystyle#2$\else#2\fi}}%
}
\newcommand{\gamedoublebox}[1]{{\setlength{\fboxsep}{1pt}\fbox{\setlength{\fboxsep}{1pt}\fbox{\ifmmode$\displaystyle#1$\else#1\fi}}}}
\newcommand{\gameblue}[1]{{\color{red} #1}}

\newenvironment{innerproof}
 {\renewcommand{\qedsymbol}{}\proof}
 {\endproof}

\section{Introduction}

\subsection{Unclonable Cryptography}

Quantum information enables us to achieve new cryptographic primitives that are
impossible classically, leading to a prominent research area named unclonable
cryptography.
At the heart of this area is the no-cloning principle of quantum
mechanics~\cite{wootters1982single}, which has given rise to many unclonable
cryptographic primitives.
This includes quantum money~\cite{wiesner1983conjugate}, quantum
copy-protection~\cite{aaronson2009quantum}, unclonable
encryption~\cite{broadbent2019uncloneable}, single-decryptor
encryption~\cite{C:CLLZ21}, and many more.
In this work, our focus is on quantum copy-protection and unclonable encryption.

\paragraph{Copy-protection for point functions.}
Quantum copy-protection, introduced by Aaronson in~\cite{aaronson2009quantum}, is a
functionality preserving compiler that transforms programs into quantum states.
Moreover, we require that the resulting copy-protected state should not allow
the adversary to copy the functionality of the state.
In particular, this unclonability property states that, given a copy-protected quantum program, no adversary can produce two (possibly entangled) states that both can be used to compute this program.
Testing whether these two states can compute the program is done by sampling two challenges input for the program from a certain \emph{challenge distribution}.
Then, informally, each state is used as a quantum program and run on the corresponding challenge to produce some outcome; and the test passes if this outcome is the one that would be output by the program on this input.

While copy-protection is known to be impossible for general unlearnable
functions and the class of de-quantumizable algorithms~\cite{EC:AnaLaP21},
several feasibility results have been demonstrated for cryptographic functions (e.g.,
pseudorandom functions, decryption and signing
algorithm~\cite{C:CLLZ21,TCC:LLQZ22}).
Of particular interest to us is the class of point functions, which is of the
form \(f_{y}(\cdot)\): it takes as input \(x\) and outputs 1 if and only if
\(x = y\).

Prior
works~\cite{EPRINT:ColMajPor20,TCC:AnaKal21,C:AKLLZ22,C:AnaKalLiu23,chevalier2023semi}
achieved a copy-protection scheme for point functions with different type of
states (e.g., BB84 states~\cite{bennett2020quantum} or coset states~\cite{C:CLLZ21}) and
different challenge distributions.
However, in contrast to known constructions for copy-protection for
cryptographic functions which are in the plain model, these constructions for
copy-protection for point functions are almost all in the quantum random oracle model.
The only known copy-protection for point functions scheme in the plain model (without
random oracle or another setup assumption) was
recently constructed in~\cite{chevalier2023semi}, but this scheme was shown to be
secure with respect to a ``less natural'' challenge distribution.
We note that different feasibility for the same copy-protection scheme, based on
different challenge distributions, can be qualitatively incomparable.
That is, security established under one challenge distribution might not
necessarily guarantee security under a different challenge distribution.

Given the inability to prove security with respect to certain natural challenge
distributions for copy-protection for point functions, an important question
that has been left open from prior works is the following:
\begin{center}
  \emph{\textbf{Question 1.}
	Do copy-protection schemes for point functions, with negligible security\\and
	natural challenge distributions, in the plain model exist?}
\end{center}

\paragraph{Unclonable encryption.}
Unclonable encryption, introduced by Broadbent and
Lord~\cite{broadbent2019uncloneable} based on a previous work of Gottesman~\cite{gottesman2002uncloneable}, is another beautiful primitive of
unclonable cryptography.
Roughly speaking, unclonable encryption is an encryption scheme with quantum
ciphertexts having the following security guarantee: given a quantum ciphertext,
no adversary can produce two (possibly entangled) states that both encode some
information about the original plaintext.
Interestingly, besides its own applications, unclonable encryption also implies
private-key quantum money, and copy-protection for a restricted class of
functions~\cite{broadbent2019uncloneable,TCC:AnaKal21}.

Despite being a natural primitive, constructing unclonable encryption has
remained elusive.
Prior works~\cite{broadbent2019uncloneable,TCC:AnaKal21} established the
feasibility of unclonable encryption satisfying a weaker property called
unclonability, which can be seen as a \emph{search}-type security.
This weak security notion is far less useful, as it does not imply the standard
semantic security of an encryption scheme, and also does not lead to the
application implication listed above.
The stronger notion, the so-called \emph{unclonable indistinguishability}, is only
known to be achievable in the quantum random oracle model~\cite{C:AKLLZ22}.
Given the notorious difficulty of building unclonable encryption in the standard
model, the following question has been left open from prior works:

\begin{center}
  \emph{\textbf{Question 2.}
	Do encryption schemes satisfying unclonable indistinguishability in the
	plain model exist?}
\end{center}

\subsection{Monogamy Games}

In order to understand better the difficulty of achieving such goals, we first recall the security 
definitions of these primitives, called anti-piracy security. This notion is 
defined through a piracy game, in which Alice is given a certain quantum state.
Alice must then split this state and share it between two other 
adversaries, Bob and Charlie. Then, Bob and Charlie receive a challenge and must guess the correct 
answer.

This security can be proven through the use of monogamy games, which are games whose winning 
probability is restricted by the monogamy-of-entanglement; in order to win the game with the highest 
probability, the players have to leverage the power of entanglement in the best possible way, but 
monogamy-of-entanglement prevent them to win with probability $1$. As a simple example, consider the 
following game, studied in particular in \cite{Tomamichel_2013}. This game is between a challenger and 
two players, Bob and Charlie. Bob and Charlie are first asked to prepare a tripartite quantum state 
$\rho_{ABC}$; then to send the register $A$ to the challenger; and finally to share the remaining 
registers between themselves. From this step, Bob and Charlie cannot communicate anymore. Then, the 
challenger measures each qubit of this register in a random basis - either computational or diagonal - 
and reveal the bases to Bob and Charlie. Bob and Charlie are now both asked to guess the outcome of 
the challenger's measurement. The maximum winning probability of this game is $\frac{1}{2} + 
\frac{1}{2\sqrt{2}}$.

In the following, we consider games with a slightly different structure: the games are between a 
challenger and three players, Alice, Bob, and Charlie - where Bob and Charlie cannot communicate. The 
challenger first sends a quantum state to Alice who has to split it and share it between Bob and 
Charlie. Bob and Charlie are then asked a question and both need to return the correct answer. 
Interestingly, in these games, the questions asked Bob and Charlie would have been easily answered by 
Alice before she splits the state. We are indeed interested in how well she can split the state to 
preserve as much as possible the information necessary to answer correctly in each share.

In \cite{C:CLLZ21}, the authors defined the \emph{coset states}: quantum states of the form
\(\ket{A_{s, s'}} \coloneqq \sum_{x \in A} (-1)^{\inp{x}{s'}} \ket{x + s}\) (up to renormalization)
for a subspace \(A \subseteq \FF_{2}^{n}\) and two vectors \(s, s' \in \FF_{2}^{n}\). Loosely
speaking, a coset state \(\ket{A_{s, s'}}\) embeds information on both the regular coset \(A + s\) and
its dual coset \(A^\perp + s'\), in the sense that measuring a coset state in the computational basis
yields a random vector in the regular coset; and measuring it in the diagonal basis yields a vector in
the dual coset. The coset states feature a so-called \emph{strong monogamy-of-entanglement property}
(proven in \cite{Culf2022monogamyof}). This property states that no adversaries Alice, Bob and Charlie
can win the following monogamy game with non-negligible probability. Given a random coset state
\(\ket{A_{s, s'}}\), Alice has to split the state and share it between Bob and Charlie. Bob and Charlie
then receive the description of the subspace \(A\) as the question, and are asked to return a vector
 in the regular coset \(A + s\) for Bob, and a vector in the dual coset \(A^\perp + s'\) for Charlie.

\subsection{Our contributions}

Unfortunately, these monogamy games are not adapted to some distributions, specifically the identical 
and product distributions, where the elements drawn can be equal. To solve this issue, we present in 
this game a monogamy game that we call \emph{monogamy game with identical basis}.

Informally, in this game, Bob and Charlie are not asked to return a vector belonging to different
cosets (the regular coset for Bob and the dual coset for Charlie), but are instead instructed to
return a vector belonging to the same coset (both in the regular coset or both in the dual coset).

Of course, without any additional constraints, Alice could simply measure the coset state in, say, the
computational basis, and forward the outcome to both Bob and Charlie. The latter could in turn simply
return this outcome and always answer correctly. To prevent such a trivial strategy, the challenger
instructs Bob and Charlie on the basis in which the vectors they return must belong. More precisely,
the challenger sends as a question the subspace description $A$ as for the \cite{C:CLLZ21} game, but
also a bit $b$. The expected vectors must then both belong to the regular coset if $b = 0$, or in the
dual coset if $b = 1$. Crucially, this basis $b$ is sampled and revealed to the adversaries
\emph{after} Alice splits the state. Otherwise, she could simply measure the state in the
computational or the diagonal basis depending on the value of $b$, and forward the outcome to Bob and
Charlie.

We prove that the winning probability of this game is at most negligibly higher than $1/2$, which 
corresponds to the trivial strategy in which Alice always measures the coset state in the 
computational basis and forwards the outcome to Bob and Charlie, who in turn return it. An 
illustration of this game is depicted in \cref{fig:new-moe-coset}.
This new monogamy-of-entanglement (MoE) property of coset states might be of independent 
interest.\footnotemark 

Unfortunately, for reasons detailed in Section~\ref{sec:conjecture}, this new monogamy game is not 
sufficient to answer the two questions above affirmatively. We also need that the existence of a 
compute-and-compare obfuscator~\cite{C:CLLZ21} is still true in a non-local setting. Admitting this 
conjecture, we present a construction of copy-protection of point functions with negligible security 
in the plain model. We show that this construction is secure, for three families of distributions: 
product distributions, identical distributions and non-colliding distribution. Secondly, we exhibit 
two constructions of unclonable encryption with unclonable indistinguishability security in the plain 
model: one for single-bit encryption and the other for multi-bit encryption. Our constructions based 
on the construction of single-decryptor, introduced by~\cite{C:CLLZ21}, with new security variants. 

\footnotetext{Recently, \cite{cryptoeprint:2023/410} also presented a new version of monogamy-of-entanglement game, using a similar idea.
We discuss the differences between their version and ours in \cref{sec:moe-parallel}.
}

\paragraph{Unclonable unforgeability for tokenized signatures.}
We also present a new security definition for tokenized signatures, and show the existence of a tokenized signature scheme featuring this security.
Loosely speaking, a tokenized signature scheme (\cite{ben2023quantum}) allows an authority to generate quantum signing tokens, which can be used to sign one, and only one, message on the authority's behalf.
Up to our knowledge, the existing literature on the subject only consider weak (\cite{ben2023quantum,C:CLLZ21}) and strong unforgeability (\cite{chevalier2023semi}), where the adversary is asked to return two valid signatures out of a single token.
We define a so-called unclonable unforgeability property for tokenized signature schemes, where a first adversary is asked to split a token such that each part can be used to produce a valid signature of a random message \emph{chosen after the splitting}.
We define formally this new property, and show how our new monogamy of entanglement with identical basis property can be used to prove that the \cite{C:CLLZ21} construction achieves this security.

\paragraph{Concurrent and independent work.}
The first version of this paper appeared concurrently and independently with two
other works considering similar tasks.
However, at a high level, the themes of these two papers and ours are quite
different.
Coladangelo and Gunn~\cite{EPRINT:ColGun23} show the feasibility of
copy-protection of puncturable functionalities and point functions through a new
notion of quantum state indistinguishability obfuscation, which is also
introduced in the same paper.
Ananth and Behera~\cite{EPRINT:AnaBeh23} also show constructions for
copy-protection of puncturable functionalities (including point functions) and
unclonable encryption, based on a new notion of unclonable puncturable
obfuscation.
Among the two, the latter is most similar to our work.
Their construction of unclonable puncturable obfuscation, which is the backbone
for their applications (of copy-protection of point functions and unclonable
encryption), is based on the recent construction of copy-protection of
pseudorandom functions and single-decryptor of Coladangelo et
al.~\cite{C:CLLZ21}.
They show that a slightly modified construction of \cite{C:CLLZ21} achieves
anti-piracy security with different challenge distributions and preponed
security.
Apart from the naming, these security notions are identical to what we consider
here in our paper.

After posting the first version of our paper online, we have had discussions with the authors
of \cite{EPRINT:AnaBeh23}.
We acknowledge that the idea of introducing conjectures was inspired by the work
of Ananth and Behera~\cite{EPRINT:AnaBeh23}.
We compare~\cite{EPRINT:AnaBeh23}'s conjectures with ours in
Section~\ref{sec:conjecture}.

\subsection{Technical Overview}

\paragraph{A new monogamy-of-entanglement game of coset states.}
In the heart of our results is a new monogamy-of-entanglement property of coset
states, drawing inspiration from previous works~\cite{C:CLLZ21,Tomamichel_2013}.
A coset state is a quantum state of the form
\(\ket{A_{s, s'}} \coloneqq \frac{1}{\sqrt{\size{A}}} \sum_{x \in A} (-1)^{\inp{x}{s'}} \ket{x + s}\)
for a subspace \(A \subseteq \FF_{2}^{n}\) and two vectors \(s, s' \in \FF_{2}^{n}\).
Loosely speaking, a coset state \(\ket{A_{s, s'}}\) embeds information on both
the coset \(A + s\) and its dual \(A^\perp + s'\), and has the following
monogamy-of-entanglement property~\cite{C:CLLZ21}: given a random coset state
\(\ket{A_{s, s'}}\), no adversary - Alice - can split the state and share it to two
other non-communicating adversaries - Bob and Charlie - such that, given the description of the subspace
\(A\), Bob returns a vector in the coset \(A + s\) and Charlie a vector in the dual \(A^\perp + s'\).

In this paper, we introduce a new variant of the monogamy-of-entanglement
property of coset states.
In this variant, Bob and Charlie both have to output a vector in the same coset
space, either \(A + s\) or \(A^\perp + s'\), but they learn the challenge coset
space only during the challenge phase after receiving the state from Alice.
Crucially Alice also does not know the challenge coset space before the
challenge phase.
We call this new game as \emph{monogamy-of-entanglement game with identical basis}.
An illustration of this new game is depicted in~\cref{fig:new-moe-coset}.
We will show that the winning probability of this game is at most negligibly far
way from \(1/2\), which corresponds to the trivial strategy in which Alice
always measures the coset state in the computational basis and forwards the
outcome to both Bob and Charlie, who in turn output it.
We will also show that the winning probability of this game can be made
negligible by parallel repetition (see \cref{subsec:moe-parallel}).
\begin{figure}
    \centering
    \begin{tikzpicture}[framed]

        \node[draw, square, align=center] (challenger) {$\mathsf{Challenger}$\\
        \scriptsize $(A, s, s') \gets \$$\\
        \scriptsize $b \gets \bin$};

        \node[right=25mm of challenger, draw, square, minimum width={width("AAAA")}] (alice) {$\adv$};
        \draw[->] (challenger)
            -- node[above, midway] {$\ket{A_{s, s'}}$}
            node[below, midway] {} (alice);

        \node[draw, square, above right=5mm and 4mm of alice, minimum width={width("AAAA")}] (bob) {$\bdv$};
        \draw[->] (alice) |- node[above]{$\sigma_1$} (bob);
        \node[above left=5mm and 15mm of bob] (xb) {$(A, b)$};
        \draw[->, dashed] (challenger) |- (xb) -| (bob);
        \node[right=7mm of bob] (yb) {$u_1$};
        \draw[->] (bob) -- (yb);

        \node[draw, square, below right=5mm and 3mm of alice, minimum width={width("AAAA")}] (charlie) {$\cdv$};
        \draw[->] (alice) |- node[below]{$\sigma_2$} (charlie);
        \node[below left=10mm and 15mm of charlie] (xc) {$(A, b)$};
        \draw[->, dashed] (challenger) |- (xc) -| (charlie);
        \node[right=5mm of charlie] (yc) {$u_2$};
        \draw[->] (charlie) -- (yc);

        \node[right=of yb, align=left, rectangle, draw] () {
            Winning Condition:\\
            $u_1, u_2 \in \begin{cases}
                A + s & \text{ if } b = 0\\
                A^\perp + s' & \text{ if } b = 1
              \end{cases}$};

    \end{tikzpicture}

    \caption{Monogamy-of-Entanglement Game with Identical Basis (Coset Version).
    Remark that, in the original monogamy-of-entanglement game for coset states \cite{C:CLLZ21}, the challenger does not sample $b$, hence there is no $b$ sent to $\bdv$ and $\cdv$, and the winning condition is $\left(u_1 \in A + s\right) \land \left(u_2 \in A^\perp + s'\right)$.}
    \label{fig:new-moe-coset}
\end{figure}
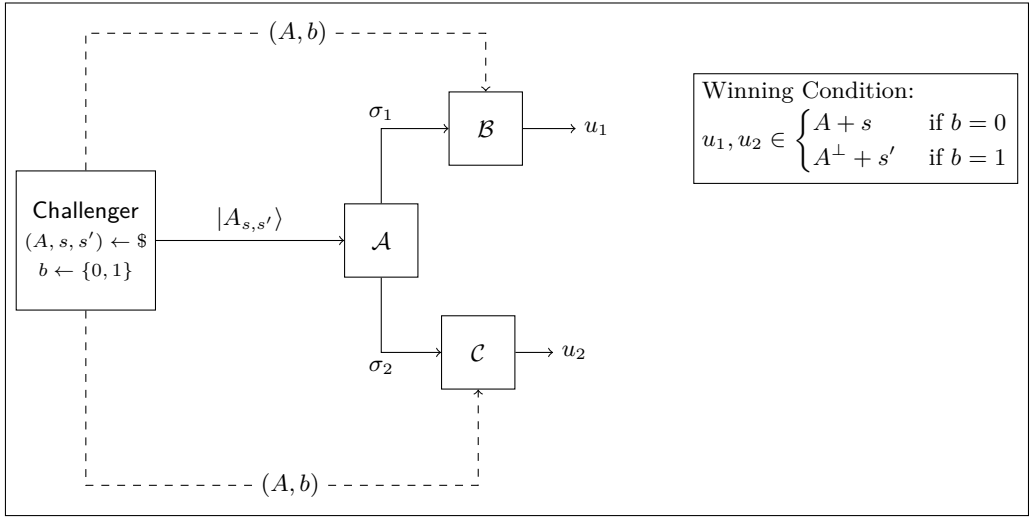
 
We now give a sketch of the proof for this new monogamy-of-entanglement game.
For simplicity, we describe the proof of the \emph{BB84 version} of our new
monogamy-of-entanglement game, since the coset version reduces to this game as
proven in \cite{Culf2022monogamyof}.
In the BB84 version, the challenger sends $n$ BB84 states
$\bigotimes_{i = 1}^n \ket{x_i}^{\theta_i}$ to Alice, and Bob and Charlie are
given the basis $\theta$ and a random bit $b$.
To win the game, Bob and Charlie both need to output a bitstring $x^*$ such that
$x^*$ is equal to $x$ on all the indices $i$ such that $\theta_i = b$.
This proof uses the template of \cite{Culf2022monogamyof} and can be described in
three steps.
We refer the reader to~\cref{sec:new-moe} for the formal proof.
\begin{enumerate}
  \item In the first step, we define the \emph{extended non-local game}
    \cite{JMRW16} associated to this monogamy-of-entanglement game.
    This game is between a challenger and two players.
    The players start by preparing a tripartite quantum state $\rho_{012}$; each of
    them keep one register, and they send the last one, say $\rho_2$, to the
    challenger.
    After this point, the players are not allowed to communicate.
    The challenger samples $n$ BB84 basis $\theta \in \bin^n$ at random, then measures
    each qubit $\rho_{C, i}$ of $\rho_C$ in the corresponding basis $\theta_i$; let $x$
    denote the outcome.
    Finally, the challenger sends $\theta$, as well as a random bit $b$, to the two
    players.
    Each player is asked to output a bitstring $x^*$ such that $x^*$ is equal to
    $x$ on all the indices $i$ such that $\theta_i = b$.

    We show that the largest winning probability of the monogamy game is the
    same as the one of this extended non-local game.
    In this step, we use a technique from \cite{Tomamichel_2013} to bound this
    winning probability.
  \item In the second step, we express any strategy for this extended non-local
    game with security parameter \(n \in \NN\) as a tripartite quantum state
    \(\rho_{012}\) as well as two families of projective measurements,
    \(\{B^{\theta, b}\}\) and \(\{C^{\theta, b}\}\), both indexed by \(\theta \in \thetaset_n\)
    and \(b \in \bin\).
    We define the projector
    \(\rmPi_{\theta, b} = \sum_{x \in \bin^n} \ketbra{x}{x}^\theta \otimes B^{\theta, b}_{x_{T_b}} \otimes C^{\theta, b}_{x_{T_b}}\)
    such that the winning probability of this strategy is
    \(p_{win} = \expectop_{\theta, b}\left[ \trace\left(\rmPi_{\theta, b}\ \rho_{012}\right)\right]\).
    Then, we show the following upper-bound:
    \begin{align*}
      p_{win} & \leq \frac{1}{2N}\sum_{\substack{1 \leq k \leq N                                                                                             \\\alpha \in \bin}}\ \max_{\theta, b} \big \lVert \rmPi_{\theta, b} \rmPi_{\pi_{k, \alpha}(\theta \Vert b)} \rVert\\
      ~       & = \frac{1}{2} + \frac{1}{2N}\sum_{1 \leq k \leq N}\ \max_{\theta, b} \big \lVert \rmPi_{\theta, b} \rmPi_{\pi_{k, 1}(\theta \Vert b)} \rVert
    \end{align*}
    where \(\{\pi_{k, \alpha}\}_{k \in \intval{1, N}, \alpha \in \bin}\) is a mutually
    orthogonal family of permutations to be defined later in the proof.
    We want the maximum in the equation above to be as small as possible.
    The goal of step 3 is to find such a family.
  \item In the third step, we show that, as long as \(b' \neq b\), the quantity
    \(\lVert \rmPi_{\theta, b}\ \rmPi_{\theta', b'} \rVert\) depends on the number of
    indices on which $\theta$ and $\theta_i$ differ.
    More precisely, \(\lVert \rmPi_{\theta, b}\ \rmPi_{\theta', b'} \rVert\) is
    upper-bounded by $2^{-d(\theta) / 4}$, where $d(\theta)$ is the number of such
    indices.
    Thus, we choose our family of permutations such that, for all $k$, the last
    bit of $\pi_{k, 1}(\theta, b)$ is $1 - b$ and $d(\theta)$ is constant.
    We build upon a result of \cite{Culf2022monogamyof} to construct a family of
    permutations with the aforementioned properties.
    More concretely, \cite{Culf2022monogamyof} define a mutually orthogonal
    family of permutations $\pi_k$, indexed by $1 \geq k \geq N$, with the latter
    property.
    We then define another family $\widetilde{\pi}_{k, b}$, indexed by $1 \geq k \geq N$
    and $b \in \bin$, where $\widetilde{\pi}_{k, b}(\theta, b) = \pi_k(\theta) \Vert 1 - b$.
    It is easy to see that this new family of permutations has both the former
    and the latter properties, and we prove that it is also a mutually
    orthogonal family.
\end{enumerate}

\paragraph{Anti-piracy security.}
We now describe how this new monogamy-of-entanglement property might
allow us to obtain new constructions of copy-protection and unclonable
encryption.
We note that we did not manage to prove security of our constructions from
standard assumptions, and they are indeed based on this new monogamy game and a
conjecture that we also introduce in this work.
We first recall the anti-piracy security definition, discuss several challenge
distributions for copy-protection of point functions, and then present
techniques to achieve security with respect to these challenge distributions.

A piracy game is formalized as a security experiment against a triple of cloning adversaries Alice, Bob, and Charlie.
Alice receives a copy-protected program
\(\rho_{f} \coloneqq \cpprotect(f)\), which can be used to evaluate a classical function
\(f\), prepares a bipartite state, and sends each half of the state to the two
other non-communicating adversaries Bob and Charlie.
In the challenge phase, Bob and Charlie receive
inputs \(c_{1}, c_{2}\), sampled from a challenge distribution and are asked to
output \(b_{1}, b_{2}\).
The adversaries win if \(b_{i} = f(c_{i})\) for
\(i \in \{1, 2\}\).

It turns out that the choice of challenge distribution plays a crucial role in
evaluating security of copy-protection schemes.
Indeed, previous constructions of copy-protection of point functions have
considered different challenge
distributions~\cite{EPRINT:ColMajPor20,TCC:BJLPS21,C:AKLLZ22,chevalier2023semi}.
Some are considered ``less natural'' than the others.
Ideally, we would like to prove security of the scheme in a way that is
independent of the chosen challenge distribution.
In this paper, we make progress towards achieving this goal.
In particular, in the following, let \(y \in \bin^{n}\) the copy-protected point,
\(x, x'\) random strings drawn from some distribution.
We consider the following challenge distributions for copy-protecting point
functions.
\begin{itemize}
  \item \textit{Identical:} Bob and Charlie get
        either \((y, y)\) or \((x, x)\) with probability \(\frac{1}{2}\) each, where \(x\)
        is drawn uniformly at random from \(\bin^{n} \setminus \{y\}\).
  \item \textit{Product:} Bob and Charlie get either
        \((y, y)\), \((x, y)\), \((y, x)\), or \((x, x')\) each with probability
        \(\frac{1}{4}\), where \(x, x'\) are drawn uniformly at random from
        \(\bin^{n} \setminus \{y\}\).
  \item \textit{Non-Colliding:} Bob and Charlie get either
        \((x, y)\), \((y, x)\), or \((x, x')\) each with probability
        \(\frac{1}{3}\), where \(x, x'\) are drawn uniformly at random from
        \(\bin^{n} \setminus \{y\}\).
\end{itemize}
Arguably, since the copy-protected point basically represents the entire
functionality of the point function, one would say the product distribution is
the most meaningful and natural one.
However, the only known construction known before our work that achieves
copy-protection of point functions in the plain model is the one given
in~\cite{chevalier2023semi}, which only achieves security w.r.t non-colliding
distribution.
Our construction is identical to that of \cite{chevalier2023semi} and our main
technical contribution lies in our proof technique showing that
\cite{chevalier2023semi} construction can achieve security with respect to
product and identical distributions.

We continue by recalling \cite{chevalier2023semi} construction and briefly explain
where it fails when proving security w.r.t the product challenge distribution,
then we describe techniques that might allow us to overcome the
problems.

\paragraph{\cite{chevalier2023semi}'s copy-protection of point functions.}
At a high level, \cite{chevalier2023semi} scheme uses a copy-protection scheme
of pseudorandom functions (PRFs) \(\prf(\key, \cdot)\) from~\cite{C:CLLZ21}.
Protecting a point function \(\pf_y\) is done in the following way: sample a PRF
key \(\key\); then copy-protect \(\key\) using the PRF protection algorithm to
get \(\qkey_\key\); and finally compute \(z \gets \prf(\key, y)\) and return the
outcome \(z\) as well as \(\qkey_\key\).
One can evaluate the copy-protected point function \(\pf_{y}\) on an input \(x\)
in the following way: compute \(\prf(\key, x)\) using the evaluation algorithm
of the PRF copy-protection scheme, then check whether the outcome equals \(z\)
or not and return \(1\) or \(0\) accordingly.
Although \cite{chevalier2023semi} construction can be cast in the form, we note
that their reduction (and ours) go through an intermediate notion of single-decryptor, which ultimately reduces to some form of monogamy-of-entanglement of
hidden coset states.
We refer the reader to the formal proof provided in~\cref{sec:sd}
and~\cref{sec:copy-prot} for more details.

\paragraph{Challenges when proving anti-piracy security w.r.t the product distribution.}
To prove security based on monogamy-of-entanglement of coset states, the authors
of~\cite{chevalier2023semi} (based on techniques from~\cite{C:CLLZ21}) use an
extraction property of compute-and-compare obfuscation to extract and outputs two vectors
which, with non-negligible probability belong respectively to \(A + s\) and
\(A^\perp + s'\), which works perfectly when the challenge distribution is
non-colliding.
However, when considering the identical distribution (or the product distribution for the case when the challenge inputs are \((y, y)\)), the adversaries are required to output two vectors
(not necessarily different) from \emph{the same} coset space: that is, they are
either both in \(A+s\) or \(A^{\perp} + s'\).
This in turn leads to no violation against the monogamy-of-entanglement game
describe above.
Worse, if the first adversary Alice knows which basis it would play with (either the
computational basis for coset space \(A + s\) or the Hadamard basis for coset
space \(A^{\perp} + s'\)), the adversaries can win the game trivially.

Our observation here is that the challenge inputs pair \((y, y)\) corresponds to
a description of the challenge basis for the monogamy-of-entanglement with
identical basis game: in particular, let \(y \coloneqq y_{0} \ldots y_{n}\), each
\(y_{i}\) describes the challenge basis for the \(i\)-th instance of the
monogamy game: if \(y_{i} = 0\), it is the computational basis (corresponding to
the coset space \(A_{i} + s_{i}\)), otherwise, it is the Hadamard basis
(corresponding to the coset space \(A^{\perp}_{i} + s'_{i}\)).
The final step in the proof is to show that, if there exists an adversary that
wins the anti-piracy game with challenge input \((y, y)\), we can construct two
non-communicating extractors that output \(n\) vectors
\((v_{i}, w_{i})_{i \in [1, n]}\) satisfying that \(v_{i}, w_{i}\) both belong to
the same challenge coset space for all \(i \in [1, n]\).
In the proof of \cite{C:CLLZ21}, where the challenge instances given to the two
non-communicating adversaries Bob and Charlie are sampled independently,
this step can be done by using extracting compute-and-compare obfuscation technique.
However, in our case, we face a new problem that now the extraction needs to be
done simultaneously where the two challenges are correlated.
To remedy the issue, we propose a new conjecture on simultaneous extracting from
compute-and-compare obfuscation.
We note that weaker version of this conjecture has been proven
in~\cite{C:CLLZ21}.

\paragraph{Simultaneously extracting from compute-and-compare obfuscation
  conjecture.}
Assuming \(\iO\) and (sub-exponentially) hardness of LWE, for
(sub-exponentially) unpredictable distribution $\distrib$, there exists a
compute-and-compare obfuscator~\cite{C:CLLZ21}.
We are interested in whether this result still holds in a non-local context.
More precisely, consider the two following tasks, which we call
\emph{simultaneous distinguishing} and \emph{simultaneous predicting}.
Simultaneous predicting asks two players, Bob and Charlie, given a function
associated to a compute-and-compare program, and a quantum state as auxiliary
information on the program, to output the associated lock value.
Crucially, the challenge given to Bob and Charlie might be correlated.
In simultaneous distinguishing, Bob and Charlie are given either an obfuscated
compute-and-compare program, or the outcome of a simulator on this program's
parameters.
As in simultaneous predicting, they are also given a quantum state each, but
here, they are asked to tell whether they received the obfuscated program, or
the simulated one.
Our conjecture essentially says that simultaneous predicting \emph{implies}
simultaneous distinguishing, for certain challenge distributions.

\paragraph{Unclonable encryption.}
In this paper, we also propose a construction for unclonable encryption with unclonable indistinguishability in the plain model.
The unclonable indistinguishability for this primitive is also defined through a piracy game, in which Alice receives a quantum encryption of a bit $b$, prepares a bipartite state, and sends each half of the state to two non-communicating adversaries Bob and Charlie.
In the challenge phase, Bob and Charlie both receive
the decryption key $\key$ and are asked to output \(b_{1}, b_{2}\).
Alice, Bob, and Charlie win if \(b_{i} = b\) for
\(i \in \{1, 2\}\).
Our construction of unclonable encryption also uses a copy-protection scheme of PRF.
A key is simply a random bitstring $\key_S$.
Encrypting a bit $b$ is done by in the following way: sample a PRF
key \(\key_P\); then copy-protect \(\key_P\) using the PRF protection algorithm to
get \(\rho_{\key_P}\); finally sample a fresh random bitstring $r$ and output $(r, y, \rho_{\key_P})$ where $y$ is either $\prf(\key_P, \key_S \oplus r)$ if $b = 0$, or a random bitstring if $b = 1$.
Similarly, as for copy-protection of point function, the security of our unclonable encryption construction also reduces to a monogamy-of-entanglement game.
As in the piracy game for this primitive, the same challenge is used for both Bob and Charlie, we meet the same problem as for our copy-protection construction, namely that the adversaries are required to output two vectors from the same coset space.

\section*{Acknowledgements}
This work was supported in part by the French ANR projects CryptiQ
(ANR-18-CE39-0015) and SecNISQ (ANR-21-CE47-0014).
\section{Preliminaries}

\subsection{Notations}
Throughout this paper, \(\secpar\) denotes the security parameter.
The notation \(\negl\) denotes any function \(f\) such that \(f(\lambda) = \lambda^{-\omega(1)}\), and \(\poly\) denotes any function \(f\) such that \(f(\lambda) = \mathcal{O}(\lambda^c)\) for some \(c > 0\).
The notation $\subexp(\secpar)$ denotes a sub-exponential function.

When sampling uniformly at random a value \(x\) from a set~\(\mathcal{S}\), we employ the notation \(x \sample \mathcal{S}\).
When sampling a value \(x\) from a probabilistic algorithm \(\adv\), or from a distribution $\distrib$, we employ
the notation \(a \gets \adv\), or $a \gets \distrib$.

By \(\ppt\) we mean a polynomial-time non-uniform family of probabilistic
circuits, and by \(\qpt\) we mean a polynomial-time family of quantum circuits.
When we write that an algorithm $\adv$ is ``efficient'', we mean that $\adv$ is $\qpt$.
We note $\adv(x; r)$ to denote that we run $\adv$ on input $x$ with random coins $r \in \bin^{\poly}$ as the random tape.
In the context of security games, we abuse the notations and sometimes write ($\qpt$) adversary instead of ($\qpt$) algorithm.
We also sometimes write that a $\qpt$ algorithm is run on a classical input $x$ instead of writing that it is run on $\ketbra{x}{x}$.

\subsection{Distributions}
\label{subsec:preli-distribs}
We define two families of distributions that we often consider in this paper.

\begin{definition}[Uniform Distribution]
  Let $S_\secpar$ be any set, and $\secpar \in \NN$.
  We write that a distribution $\distrib_\secpar$ over $S_\secpar \times S_\secpar$ is \emph{uniform} if it yields pairs of the form $(x_1, x_2)$ where $x_1$ and $x_2$ are independently and uniformly sampled from $S_\secpar$.

  Similarly, we write that a family of distributions $\distrib = \{\distrib_\secpar\}_{\secpar \in \NN}$ is uniform if all the $\distrib_\secpar$ are uniform.
\end{definition}

\begin{definition}[Identical Distribution]
  Let $S_\secpar$ be any set, and $\secpar \in \NN$.
  We write that a distribution $\distrib_\secpar$ over $S_\secpar \times S_\secpar$ is \emph{identical} if it yields pairs of the form $(x, x)$ where $x$ uniformly sampled from $S_\secpar$.

  Similarly, we write that a family of distributions $\distrib = \{\distrib_\secpar\}_{\secpar \in \NN}$ is identical if all the $\distrib_\secpar$ are identical.
\end{definition}

\subsection{Coset States}
Given a subspace $A \subset \FF_2^n$ of dimension $n / 2$ and a pair of vectors $(s, s') \in \FF_2^n$, the coset state $\ket{A_{s, s'}}$ is defined as
$$
  \ket{A_{s, s'}} := \frac{1}{\sqrt{2^{n/2}}} \sum_{a \in A} (-1)^{a \cdot s'} \ket{a + s}
$$
where $a \cdot s'$ denotes the inner product between $a$ and $s'$.

In particular, a coset state is such that $\Hgate^{\otimes n} \ket{A_{s, s'}} = \ket{A^\perp_{s', s}}$, where $A^\perp$ is the complement of $A$, \ie{} $A^\perp := \{u \in \FF_2^n\ \vert\ u \cdot v = 0\ \forall v \in A\}$.

\paragraph{Canonical representation.}
\label{par:can-coset}
As the canonical representation of a coset $A + s$, we use the lexicographically smallest vector of the coset; and for $u \in \FF_2^n$, we note $\can_A(u)$ the function that returns the canonical representation (also noted coset representative) of $A + u$.
We note that if $u \in A + s$, then $\can_A(u) = \can_A(s)$.
Also, the function $\can_A(\cdot)$ is efficiently computable given a description of $A$.

\subsection{Indistinguishable Obfuscation}
\begin{definition}[Indistinguishability Obfuscator~\cite{C:BGIRSVY01}]
  A uniform \ppt{} machine \(\iO\) is called an indistinguishability obfuscator
  for a classical circuit class \(\{\mathcal{C}_{\secpar}\}_{\secpar \in \NN}\) if the following
  conditions are satisfied:
  \begin{itemize}
    \item For all security parameters \(\secpar \in \NN \), for all
      \(C \in \mathcal{C}_{\secpar}\), for all input \(x\), we have that
    \[
      \prob{C'(x) = C(x) \mid C' \gets \iO(\secpar, C)} = 1.
    \]
    \item For any (not necessarily uniform) distinguisher \(\mathcal{D}\), for all
    security parameters \(\secpar \in \NN \), for all pairs of circuits
    \(C_{0}, C_{1} \in \mathcal{C}_{\secpar}\), we have that if
    \(C_{0}(x) = C_{1}(x)\) for all inputs \(x\), then
    \[
      \advantage{\iO}{}[(\secpar, \adv)] \coloneqq \lvert \prob{\mathcal{D}(\iO(\secpar, C_{0})) = 1} - \prob{\mathcal{D}(\iO(\secpar, C_{1})) =1} \rvert \leq \negl.
    \]
  \end{itemize}
  We further say that \(\iO\) is \(\delta\)-secure, for some concrete negligible
  function \(\delta(\secpar)\), if for all \qpt{} adversaries \(\adv\), the advantage
  \(\advantage{\iO}{}[(\secpar, \adv)]\) is smaller than \(\delta(\secpar)^{\Omega(1)}\).
  \end{definition}

\subsection{Compute-and-Compare Obfuscation}
\label{sec:preli-cc}

\begin{definition}[Compute-and-Compare Programs]
  Given a function \(f: \bin^{n} \to \bin^{m}\) along with a lock value \(y \in \bin^{m}\) and a message \(m \in \mspace\), we define the compute-and-compare program:
  \begin{align*}
    \ccprog[f, y, m](x) \coloneqq \begin{cases}
                                               m & \text{ if } f(x) = y, \\
                                               \bot & \text{ otherwise }.
                                  \end{cases}
  \end{align*}
  When the function, lock value, and message of a compute-and-compare program are not useful in the context, we will sometimes simply write $\ccprog$ in lieu of $\ccprog[f, y, m]$.
\end{definition}

\begin{definition}[Unpredictable Distribution]
  Let \(\distrib \coloneqq \{\distrib_{\secpar}\}_{\secpar \in \NN}\) be a family of
  distributions over pairs of the form \((\ccprog[f, y, m], \aux)\) where
  \(\ccprog[f, y, m]\) is a compute-and-compare program and \(\aux\) is some
  (possibly quantum) auxiliary information.
  We say that \(\distrib\) is an \emph{unpredictable distribution} if for all
  \(\qpt\) algorithm \(\adv\), we have that
  \[
    \probsublong{(\ccprog[f, y, m], \aux) \gets \distrib_{\secpar}}{\adv(\secparam, f, \aux) = y} \leq \negl.
  \]
  Note that, in this paper, we abuse the notation and write $f$ to denote indifferently the function $f$ or an efficient description of $f$.
\end{definition}

\begin{definition}[Sub-Exponentially Unpredictable Distribution]
  Let \(\distrib \coloneqq \{\distrib_{\secpar}\}_{\secpar \in \NN}\) be a family of
  distributions over pairs of the form \((\ccprog[f, y, m], \aux)\) where
  \(\ccprog[f, y, m]\) is a compute-and-compare program and \(\aux\) is some
  (possibly quantum) auxiliary information.
  We say that \(\distrib\) is a \emph{sub-exponentially unpredictable distribution} if for all
  \(\qpt\) algorithm \(\adv\), we have that
  \[
    \probsublong{(\ccprog[f, y, m], \aux) \gets \distrib_{\secpar}}{\adv(\secparam, f, \aux) = y} \leq \frac{1}{\subexp(\secpar)}.
  \]
  Note that, in this paper, we abuse the notation and write $f$ to denote indifferently the function $f$ or an efficient description of $f$.
\end{definition}

\begin{definition}[Compute-and-Compare Obfuscator]
\label{def:ccobf}
  A \(\ppt\) algorithm \(\ccobf\) is said to be a compute-and-compare
  obfuscator for a family of unpredictable distributions
  \(\distrib \coloneqq \{\distrib_\secpar\}\) if:
  \begin{itemize}
    \item \(\ccobf\) is functionality preserving:  for all \(x\),
    \[
      \prob{\ccobf(\secparam, \ccprog)(x) = \ccprog(x)} \geq 1 - \negl
    \]
    \item \(\ccobf\) has distributional indistinguishability: there exists
      a \(\qpt\) simulator \(\simul\) such that
    \[
      \left\{\ccobf(\secparam, \ccprog), \aux\right\} \approx_{c} \left\{\simul(\secparam, \ccprog.\mathsf{param}), \aux\right\},
    \]
      where \((\ccprog, \aux) \gets \distrib_{\secpar}\), and $\ccprog.\mathsf{param}$ denotes the input size, output size, and circuit size of $\ccprog$, that are not required to be obfuscated.
  \end{itemize}
\end{definition}

\begin{theorem}[\cite{C:CLLZ21}]
  Assuming post-quantum indistinguishable obfuscation, and the hardness of LWE, there exist compute-and-compare obfuscators for sub-exponentially unpredictable distributions.
\end{theorem}

\subsection{Pseudorandom functions}
This subsection is adapted from \cite{chevalier2023semi,C:CLLZ21}.
A pseudorandom function ~\cite{FOCS:GolGolMic84} consists of a keyed
function \(\prf\) and a set of keys \(\kspace\) such that for a randomly chosen key
\(\key \in \kspace\), the output of the function \(\prf(\key, x)\) for any input \(x\) in
the input space \(\xspace\) ``looks'' random to a \qpt{} adversary, even when
given a polynomially many evaluations of \(\prf(\key, \cdot)\).
Puncturable pseudorandom functions have an additional property that some keys can be generated
\emph{punctured} at some point, so that they allow to evaluate the pseudorandom function at all
points except for the punctured points.
Furthermore, even with the punctured key, the pseudorandom function evaluation at a punctured
point still looks random.

Punctured pseudorandom functions are originally introduced
in~\cite{AC:BonWat13,PKC:BoyGolIva14,CCS:KPTZ13}, who observed that it is
possible to construct such puncturable pseudorandom functions for the construction
from~\cite{FOCS:GolGolMic84}, which can be based on any one-way
function~\cite{HILL99}.

\begin{definition}[Puncturable Pseudorandom Function]
  A pseudorandom function \(\prf : \kspace \times \xspace \to \yspace\) is a
  \emph{puncturable pseudorandom function} if there is an addition key space
  \(\kspace_{p}\) and three \ppt{} algorithms
  \(\prf = \scheme{\keygen, \puncture, \eval}\) such that:
  \begin{itemize}
    \item \(\key \gets \keygen(1^{\secpar})\).
      The key generation algorithm \(\keygen\) takes the security parameter
      \(1^{\secpar}\) as input and outputs a random key \(\key \in \kspace\).
    \item \(\key\{x\} \gets \puncture(\key, x)\).
      The puncturing algorithm \(\puncture\) takes as input a pseudorandom function key
      \(\key \in \kspace\) and \(x \in \xspace\), and outputs a key
      \(\key\{x\} \in \kspace_{p}\).
    \item \(y \gets \eval(\key\{x\}, x')\).
      The evaluation algorithm takes as input a punctured key
      \(\key\{x\} \in \kspace_{p}\) and \(x' \in \xspace\), and outputs a classical
      string \(y \in \yspace\).
  \end{itemize}

  We require the following properties of \(\prf\).
  \begin{itemize}
    \item \textbf{Functionality preserved under puncturing.}
      For all \(\secpar \in \NN \), for all \(x \in \xspace\),
      \begin{equation*}
        \pr\left[\forall x' \in \xspace \setminus \{x\}: \eval(\key\{x\}, x') = \eval(\key, x')
            \ \middle\vert
            \begin{array}{r}
              \key \sample \keygen(1^{\secpar}) \\
              \key\{x\} \sample \puncture(\key, x)
          \end{array}
        \right] = 1.
      \end{equation*}
    \item \textbf{Pseudorandom at punctured points.}
      For every \qpt{} adversary \(\adv \coloneqq (\adv_{1}, \adv_{2})\), and
      every \(\secpar \in \NN \), the following holds:
      \begin{align*}
        & \left\lvert\pr\left[
          \begin{array}{c}
            1 \gets \adv_{2}(\key\{x^{*}\}, y, \tau)
          \end{array}
          \ \middle\vert
          \begin{array}{r}
            (x^{*}, \tau) \gets \adv_{1}(1^{\secpar}, \tau) \\
            \key \sample \keygen(1^{\secpar}) \\
            \key\{x^{*}\} \sample \puncture(\key, x^{*}) \\
            y \gets \eval(\key, x^{*})
          \end{array}
        \right]\right.
        \\
        & \left.\quad - \pr\left[
          \begin{array}{c}
            1 \gets \adv_{2}(\key\{x^{*}\}, y, \tau)
          \end{array}
          \ \middle\vert
          \begin{array}{r}
            (x^{*}, \tau) \gets \adv_{1}(1^{\secpar}, \tau) \\
            \key \sample \keygen(1^{\secpar}) \\
            \key\{x^{*}\} \sample \puncture(\key, x^{*}) \\
            y \sample \yspace
          \end{array}
        \right]\right\rvert
        \leq \negl,
      \end{align*}
      where the probability is taken over the randomness of \(\keygen\),
      \(\puncture\), and \(\adv_1\).
  \end{itemize}
  Denote the above probability as \(\adv^{\prf}(\secpar, \adv)\).
  We further say that \(\prf\) is \(\delta\)-secure, for some concrete negligible
  function \(\delta(\secpar)\), if for all \qpt{} adversaries \(\adv\), the advantage
  \(\adv^{\prf}(\secpar, \adv)\) is smaller than \(\delta(\secpar)^{\Omega(1)}\).
\end{definition}

\begin{definition}[Statistically Injective Pseudorandom Function]
  \label{def:stat-inj-prf}
  A family of statistically injective (puncturable) pseudorandom functions with
  (negligible) failure probability \(\varepsilon(\cdot)\) is a (puncturable) pseudorandom
  functions family \(\prf\) such that with probability \(1 - \varepsilon(\secpar)\) over
  the random choice of key \(\key \gets \keygen(1^\secpar)\), we have that
  \(\prf(\key, \cdot)\) is injective.
\end{definition}

\begin{definition}[Extracting Pseudorandom Function]
  \label{def:extr-prf}
  A family of extracting (puncturable) pseudorandom functions with error
  \(\varepsilon(\cdot)\) for min-entropy \(k(\cdot)\) is a (puncturable) pseudorandom functions
  family \(\prf\) mapping \(n(\secpar)\) bits to \(m(\secpar)\) bits such that
  for all \(\secpar \in \NN\), if \(X\) is any distribution over \(n(\secpar)\)
  bits with min-entropy greater than \(k(\secpar)\), then the statistical
  distance between \((\key, \prf(\key, X))\) and
  \((\key, r \gets \bin^{m(\secpar)})\) is at most \(\varepsilon(\cdot)\), where
  \(\key \gets \keygen(1^\secpar)\).
\end{definition}

\section{A New Monogamy-of-Entanglement Game for Coset States}
\label{sec:new-moe}
In this section, we present a new monogamy-of-entanglement game for coset states and prove an upper-bound on the probability of winning this game.
Along the way, we present a BB84 version of this game with the same upper-bound.

\subsection{The Coset Version}
\begin{definition}[Monogamy-of-Entanglement Game with Identical Basis (Coset Version)]
  \label{def:new-moe-coset}
  This game is between a challenger and a triple of adversaries $(\adv, \bdv, \cdv)$ - where $\bdv$ and $\cdv$ are not communicating, and is parametrized by a security parameter $\secpar$.

  \begin{itemize}
    \item The challenger samples a subspace $A \gets \bin^{\secpar \times \frac{\secpar}{2}}$ and two vectors $(s, s') \gets \FF_2^n \times \FF_2^n$.
    Then the challenger prepares the coset state $\ket{A_{s, s'}}$ and sends $\ket{A_{s, s'}}$ to $\adv$.
    \item $\adv$ prepares a bipartite quantum state $\sigma_{12}$, then sends $\sigma_1$ to $\bdv$ and $\sigma_2$ to $\cdv$.
    \item The challenger samples $b \gets \bin$, then sends $(A, b)$ to both $\bdv$ and $\cdv$.
    \item $\bdv$ returns $u_1$ and $\cdv$ returns $u_2$.
  \end{itemize}
  We say that $\bdv$ makes a correct guess if $(b = 0\ \land\ u_1 \in A + s)$ or if $(b = 1\ \land\ u_1 \in A^\perp + s')$.
  Similarly, we say that $\cdv$ makes a correct guess if $(b = 0\ \land\ u_2 \in A + s)$ or if $(b = 1\ \land\ u_2 \in A^\perp + s')$.
  We say that $(\adv, \bdv, \cdv)$ win the game if both $\bdv$ and $\cdv$ makes a correct guess.
  For any triple of adversaries $(\adv, \bdv, \cdv)$ and any security parameter $\secpar \in \NN$ for this game, we note $\mathsf{MoE}_{coset}(\secparam, \adv, \bdv, \cdv)$ the random variable indicating whether $(\adv, \bdv, \cdv)$ win the game or not.
\end{definition}

We note that there is a trivial way for a triple of adversaries to win this game with probability $1/2$, by applying the following strategy.
$\adv$ samples a random bit $b^*$.
$\adv$ measures $\ket{A_{s, s'}}$ in the computational basis if $b^* = 0$, or in the Hadamard basis if $b^* = 1$.
In both cases, $\adv$ sends the outcome $u$ to both $\bdv$ and $\cdv$.
Regardless of the value of $A$ and $b$, $\bdv$ and $\cdv$ both return $u$.
Because when $b^* = b$ (which happens with probability $1/2$), the outcome of the measurement is a vector of the expected coset space, the adversaries win the game with probability $1/2$.
In the rest of this section we prove that no triple of adversaries can actually win the game with a probability significantly greater than $1/2$.

\begin{theorem}
  \label{th:new_moe_coset}
  There exists a negligible function $\negl[\cdot]$ such that, for any triple of algorithms $(\adv, \bdv, \cdv)$ and any security parameter $\secpar \in \NN$, $\prob{\mathsf{MoE}_{coset}(\secparam, \adv, \bdv, \cdv) = 1} \leq 1/2 + \negl$.
\end{theorem}

The proof of this theorem is given in subsequent sections.

\subsection{The BB84 Version}
\label{subsec:bb84-ver}
We introduce below the BB84 version of this game.
We show in the following that it is sufficient to study the BB84 version (which is simpler) to prove \cref{th:new_moe_coset}, as any triple of adversaries for the BB84 version can be turned into a triple of adversaries for the coset version without changing the probability of winning.

\paragraph{Notations.}
Through all \cref{subsec:bb84-ver} and \cref{subsec:bb84-ver-proof}, we use the following notations.
Let $n \in \NN$, we note $\thetaset_n := \{\theta \in \bin^n\ :\ \lvert \theta \rvert = n / 2\}$ - where $\lvert \cdot \rvert$ denotes the Hamming weight - and $N := {n \choose n/2}$.
Thus, $\thetaset_\secpar$ has exactly $N$ elements.

\begin{definition}[Monogamy-of-Entanglement Game with Identical Basis (BB84 Version)]
  \label{def:new-moe-bb84}
  This game is between a challenger and a triple of adversaries $(\adv, \bdv, \cdv)$ - where $\bdv$ and $\cdv$ are non-communicating, and is parametrized by a security parameter $\secpar$.
  An illustration of this game is depicted in \cref{fig:new-moe-bb84}.

  \begin{itemize}
    \item The challenger samples $x \gets \bin^\secpar$ and $\theta \gets \thetaset_\secpar$.
    Then the challenger prepares the state $\ket{x^\theta} := \bigotimes \limits_{i \in \intval{1, \secpar}} \Hgate^{\theta_i} \ket{x_i}$ and sends $\ket{x^\theta}$ to $\adv$.
    \item $\adv$ prepares a bipartite quantum state $\sigma_{12}$, then sends $\sigma_1$ to $\bdv$ and $\sigma_2$ to $\cdv$.
    \item The challenger samples $b \gets \bin$, then sends $(\theta, b)$ to both $\bdv$ and $\cdv$.
    \item $\bdv$ returns $x_1$ and $\cdv$ returns $x_2$.
  \end{itemize}
  Let $x_{T_b} := \{x_i\ \vert\ \theta_i = b\}$.
  We say that $(\adv, \bdv, \cdv)$ win the game if $x_1 = x_2 = x_{T_b}$.
  For any triple of adversaries $(\adv, \bdv, \cdv)$ and any security parameter $\secpar \in \NN$ for this game, we note $\mathsf{MoE}_{BB84}(\secparam, \adv, \bdv, \cdv)$ the random variable indicating whether $(\adv, \bdv, \cdv)$ win the game or not.
\end{definition}

We note that the trivial strategy for the coset version can be easily adapted for the BB84 one.
Hence, the greatest probability of winning this game is also lower bounded by $1/2$.

\begin{figure}
    \centering
    \begin{tikzpicture}[framed]

        \node[draw, square, align=center] (challenger) {$\mathsf{Challenger}$\\
            \scriptsize $x \gets \bin^\secpar$\\
            \scriptsize $\theta \gets \thetaset_\secpar$\\
            \scriptsize $b \gets \bin$
        };

        \node[draw, minimum width={width("AAAA")}, square, right=25mm of challenger] (alice) {$\adv$};
        \draw[->] (challenger) 
            -- node[above, midway] {$\ket{x^\theta}$}
            node[below, midway] {} (alice);

        \node[draw, minimum width={width("AAAA")}, square, above right=5mm and 4mm of alice] (bob) {$\bdv$};
        \draw[->] (alice) |- node[above]{$\sigma_1$} (bob);
        \node[above left=5mm and 15mm of bob] (xb) {$(\theta, b)$};
        \draw[->, dashed] (challenger) |- (xb) -| (bob);
        \node[right=7mm of bob] (yb) {$x_1$};
        \draw[->] (bob) -- (yb);

        \node[draw, minimum width={width("AAAA")}, square, below right=5mm and 3mm of alice] (charlie) {$\cdv$};
        \draw[->] (alice) |- node[below]{$\sigma_2$} (charlie);
        \node[below left=10mm and 15mm of charlie] (xc) {$(\theta, b)$};
        \draw[->, dashed] (challenger) |- (xc) -| (charlie);
        \node[right=5mm of charlie] (yc) {$x_2$};
        \draw[->] (charlie) -- (yc);

        \node[right=of yb, align=left, rectangle, draw] () {
            Winning Condition:\\
            $x_1 = x_2 = x_{T_b}$};

    \end{tikzpicture}

    \caption{Monogamy-of-Entanglement Game with Identical Basis (BB84 Version)}
    \label{fig:new-moe-bb84}
\end{figure}
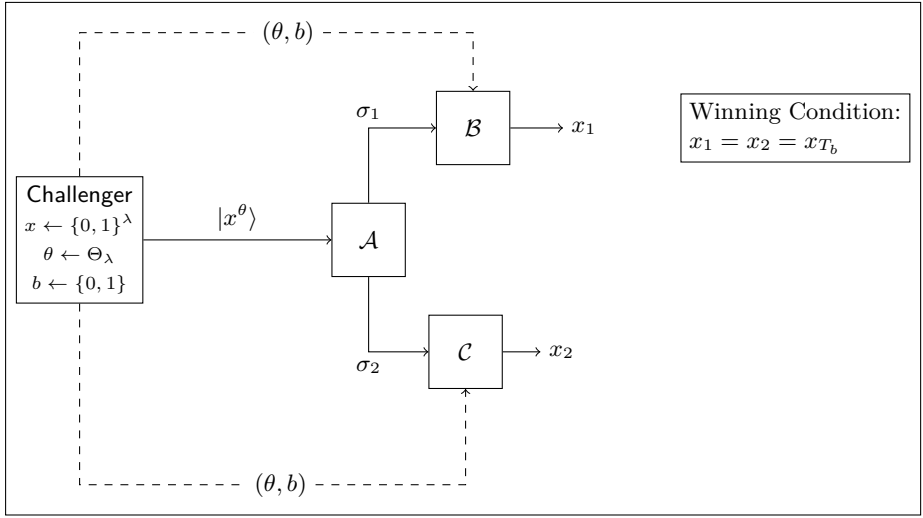
 
\begin{theorem}
  \label{th:new_moe_bb84}
  There exists a negligible function $\negl[\cdot]$ such that, for any triple of algorithms $(\adv, \bdv, \cdv)$ and any security parameter $\secpar \in \NN$, $\prob{\mathsf{MoE}_{BB84}(\secparam, \adv, \bdv, \cdv) = 1} \leq 1/2 + \negl$.
\end{theorem}

Proof of~\cref{th:new_moe_coset} follows similarly as that
of~\cite{Culf2022monogamyof}, in which the winning probability of cloning
adversaries in the monogamy-of-entanglement game of coset states reduces to the
winning probability of the adversaries in the game of BB84 states.
We thus provide the proof of~\cref{th:new_moe_bb84} below.

\subsection{Proof of \cref{th:new_moe_bb84}}
\label{subsec:bb84-ver-proof}
This proof follows the same structure as \cite{Culf2022monogamyof}.
We can separate the proof in four main steps.
\begin{enumerate}
  \item In the first step, we define the \emph{extended non-local game} \cite{JMRW16} associated the monogamy-of-entanglement game (BB84 version), and show that the greatest winning probability of the monogamy game is the same as the one of this extended non-local game.
  This step allows us to use a technique from \cite{Tomamichel_2013} to bound the winning probability.
  \item In the second step, we express any strategy for this extended non-local game with security parameter $n \in \NN$ as a tripartite quantum state $\rho_{012}$ as well as two families of projective measurements, $\{B^{\theta, b}\}$ and $\{C^{\theta, b}\}$, both indexed by $\theta \in \thetaset_n$ and $b \in \bin$.
  We define the projector $\rmPi_{\theta, b} = \sum_{x \in \bin^n} \ketbra{x}{x}^\theta \otimes B^{\theta, b}_{x_{T_b}} \otimes C^{\theta, b}_{x_{T_b}}$ such that the winning probability of this strategy is $p_{win} = \expectop_{\theta, b}\left[ \trace\left(\rmPi_{\theta, b}\  \rho_{012}\right)\right]$.
  Then, we show the following upper-bound:
  \begin{align*}
    p_{win} & \leq \frac{1}{2N}\sum_{\substack{1 \leq k \leq N\\\alpha \in \bin}}\ \max_{\theta, b} \big \lVert \rmPi_{\theta, b} \rmPi_{\pi_{k, \alpha}(\theta \Vert b)} \rVert\\
    ~       & = \frac{1}{2} + \frac{1}{2N}\sum_{1 \leq k \leq N}\ \max_{\theta, b} \big \lVert \rmPi_{\theta, b} \rmPi_{\pi_{k, 1}(\theta \Vert b)} \rVert
  \end{align*}
  where $\{\pi_{k, \alpha}\}_{k \in \intval{1, N}, \alpha \in \bin}$ is a family of permutations to be defined later in the proof.
  \item In the third step, we show that the quantity $\lVert \rmPi_{\theta, b}\ \rmPi_{\theta', b'} \rVert$ is upper-bounded by a small quantity as long as $b' \neq b$.
  \item Finally, in the fourth step, we show that there exists a family of permutations such that, when $\alpha = 0$, $\pi_{k, \alpha}(\theta, b) = (\theta', b')$ for some $\theta'$ and $b' \neq b$, and conclude the proof.
\end{enumerate}

\paragraph{Step 1: extended non-local game.}
We define the following extended non-local game, and show that any triple of adversaries that win the monogamy-of-entanglement game with same basis (BB84 version) with probability $p$ can be turned into another triple of adversaries that win this extended non-local game with the same probability $p$.

\begin{definition}[Extended Non-Local Game]
  \label{def:ext-non-local-game}
  This game is between a challenger and two adversaries $\adv$ and $\bdv$, and is parametrized by a security parameter $\secpar$.
  \begin{itemize}
    \item $\bdv$ and $\cdv$ jointly prepare a quantum state $\rho_{012}$ - where $\rho_0$ is a $\secpar$-qubits quantum state, then send $\rho_0$ to the challenger.
    $\bdv$ and $\cdv$ keep $\rho_1$ and $\rho_2$ respectively.
    From this step $\bdv$ and $\cdv$ cannot communicate.
    \item The challenger samples $\theta \gets \thetaset_n$ and $b \gets \bin$.
    Then, for all $i \in \intval{1, \secpar}$, the challenger measures the $i^{\text{th}}$ qubit of $\rho_0$ in computational basis if $\theta_i = 0$ or in Hadamard basis if $\theta_i = 1$.
    Let $m \in \bin^n$ denote the measurement outcome.
    Finally, the challenger sends $(\theta, b)$ to $\bdv$ and $\cdv$.
    \item $\bdv$ returns $m_1$ and $\cdv$ returns $m_2$.
  \end{itemize}
  Let $m_{T_b} := \{m_i\ \vert\ \theta_i = b\}$.
  We say that $(\bdv, \cdv)$ win the game if $m_1 = m_2 = m_{T_b}$.
\end{definition}

\begin{lemma}
\label{lem:ext-non-local}
  Let $n \in \NN$ and $(\adv, \bdv, \cdv)$ a triple of adversaries for the monogamy-of-entanglement game (\cref{def:new-moe-bb84}) parametrized by $n$, that win with probability $p_n$.
  Then there exists a quantum state $\rho_{012}$ and a pair of adversaries $(\adv'_1, \adv'_2)$ for the extended non-local game (\cref{def:ext-non-local-game}) that win with the same probability $p_n$.
\end{lemma}

\begin{proof}
  Consider a triple of adversaries for the monogamy-of-entanglement game (\cref{def:new-moe-bb84}), parametrized by $n \in \NN$, that win with probability $p_n$.
  We can model these adversaries as a CPTP map $\rmPhi: \mathcal{H}_0 \rightarrow \mathcal{H}_1 \times \mathcal{H}_2$, and POVMs families $\{B^{\theta, b}\}$ and $\{C^{\theta, b}\}$, both indexed by $\theta \in \thetaset_n$ and $b \in \bin$.
  Then we have
  $$
    p_n = \expectop_{\substack{\theta \in \thetaset_n\\b \in \bin}} \expectop_{x \in \bin^n} Tr\left[ (B_{x_{T_b}}^{\theta, b} \otimes C_{x_{T_b}}^{\theta, b}) \rmPhi(\ketbra{x^\theta}{x^\theta}) \right].
  $$
  The strategy for the extended non-local game is as follows.
  $\bdv$ and $\cdv$ prepare the bipartite state $\rho_{00'} = \bigotimes_{1 \leq i \leq n} \ketbra{\phi^+}{\phi^+}$ where $\phi^+$ denotes the EPR state $(\ket{00} + \ket{11})/\sqrt{2}$, and where $\rho_0$ (resp. $\rho_{0'}$) is composed of the first halves (resp. second halves) of these EPR states.
  Then, they apply $\rmPhi$ to $\rho_{0'}$.
  Let $\rho_{012}$ denotes the resulting state.
  They send $\rho_0$ to the challenger, $\bdv$ keeps $\rho_1$ and $\cdv$ keeps $\rho_2$.
  Later, when $\bdv$ receives $(\theta, b)$, from the challenger, $\bdv$ applies the POVM $B^{\theta, b}$ to $\rho_1$ and returns the outcome.
  $\cdv$ does the same with POVM $C^{\theta, b}$ and $\rho_2$.
  The probability of winning of such strategy is then
  \begin{equation}
    \label{eq:pwin-ext-non-local}
    p'_n = \expectop_{\substack{\theta \in \thetaset_n\\b \in \bin}} \sum_{x \in \bin^n} Tr\left[ \left(\ketbra{x^\theta}{x^\theta} \otimes B_{x_{T_b}}^{\theta, b} \otimes C_{x_{T_b}}^{\theta, b}\right) \rho_{012} \right].
  \end{equation}
  We do the following calculation.
  \begin{align*}
    Tr\left[ \left(\ketbra{x^\theta}{x^\theta} \otimes B_{x_{T_b}}^{\theta, b} \otimes C_{x_{T_b}}^{\theta, b}\right) \rho_{012} \right]
    &= \frac{1}{2^n} \sum\limits_{r, r' \in \bin^n} Tr\left[  \left(\ketbra{x^\theta}{x^\theta} \otimes B_{x_{T_b}}^{\theta, b} \otimes C_{x_{T_b}}^{\theta, b}\right) \left(\ketbra{r}{r'} \otimes \rmPhi \left(\ketbra{r}{r'}\right)\right) \right]\\
    ~ & = \frac{1}{2^n} \sum\limits_{r, r' \in \bin^n} \braket{r \vert x^\theta} \braket{x^\theta \vert r'} Tr\left[ \left(B_{x_{T_b}}^{\theta, b} \otimes C_{x_{T_b}}^{\theta, b}\right) \rmPhi \left(\ketbra{r}{r'}\right) \right]\\
    ~ & = \frac{1}{2^n} \sum\limits_{r, r' \in \bin^n} Tr\left[ \left(B_{x_{T_b}}^{\theta, b} \otimes C_{x_{T_b}}^{\theta, b}\right) \rmPhi \left(\ket{r}\braket{r \vert x^\theta} \braket{x^\theta \vert r'}\bra{r'}\right) \right]\\
    ~ & = \frac{1}{2^n} Tr\left[ \left(B_{x_{T_b}}^{\theta, b} \otimes C_{x_{T_b}}^{\theta, b}\right) \rmPhi \left(\frac{1}{2^n} \sum\limits_{r \in \bin^n} \ketbra{r}{r}\ \ketbra{x^\theta}{x^\theta} \frac{1}{2^n} \sum\limits_{r' \in \bin^n} \ketbra{r'}{r'}\right) \right]\\
    ~ & = \frac{1}{2^n} Tr\left[ \left(B_{x_{T_b}}^{\theta, b} \otimes C_{x_{T_b}}^{\theta, b}\right) \rmPhi \left(\ketbra{x^\theta}{x^\theta}\right) \right]
  \end{align*}
By plugging this result into \cref{eq:pwin-ext-non-local}, we get $p'_n = p_n$, which concludes the proof.

\end{proof}

\paragraph{Step 2: first upper-bound of the winning probability.}
We prove an upper-bound for the extended non-local game above.
We need the following lemma.

\begin{lemma}[Lemma 2 of \cite{Tomamichel_2013}]
  \label{lem:lem2tfkw13}
  Let $\rmPi_1, \dots, \rmPi_n$ be projective positive semi-definite operators on a Hilbert space, and $\{\pi_i\}_{i \in \intval{1, n}}$ be a set of orthogonal permutations for some integer $n$.
  Then
  $$
  \Big \lVert \sum\limits_{i = 1}^n \rmPi_i \Big \rVert \leq \sum\limits_{i = 1}^n \max\limits_{j \in \intval{1, n}} \big \lVert \rmPi_j \rmPi_{\pi_i(j)} \big \rVert
  $$
\end{lemma}

Let $\left(\{B^{\theta, b}\}_{\theta \in \thetaset_n, b \in \bin}, \{C^{\theta, b}\}_{\theta \in \thetaset_n, b \in \bin}, \rho_{012}\right)$ be a strategy for the extended non-local game.
Using Naimark's dilation theorem, we can assume without loss of generality that the $B^{\theta, b}$ and $C^{\theta, b}$ are all projective.
Let $\rmPi_{\theta, b}$ be the following projector: $\rmPi_{\theta, b} := \sum_{x \in \bin^n} \ketbra{x}{x}^\theta \otimes B^{\theta, b}_{x_{T_b}} \otimes C^{\theta, b}_{x_{T_b}}$.
Then the winning probability of this strategy is
\begin{align}
  p_{win} & = \expectop\limits_{\theta \in \thetaset_n, b \in \bin} \trace\left(\rmPi_{\theta, b}\  \rho_{012}\right) \nonumber\\
  ~       & \leq \expectop\limits_{\theta \in \thetaset_n, b \in \bin} \Vert \rmPi_{\theta, b} \Vert \nonumber\\
  ~       & \leq \frac{1}{2N}\sum_{\substack{1 \leq k \leq N\\\alpha \in \bin}}\ \max_{\theta, b} \big \lVert \rmPi_{\theta, b} \rmPi_{\pi_{k, \alpha}(\theta, b)} \rVert \label{eq:step2}
\end{align}
where the first inequality follows from the definition of the norm and the second from \cref{lem:lem2tfkw13}; and where $\{\pi_{k, \alpha}\}_{k \in \intval{1, N}, \alpha \in \bin}$ is a family of mutually orthogonal permutations.

\paragraph{Step 3: upper-bound of $\lVert \rmPi_{\theta, b} \rmPi_{\theta', 1-b} \rVert$.}
In this part, we show that for all $(\theta, \theta') \in \thetaset_n$ and all $b \in \bin$, we can upper-bound $\lVert \rmPi_{\theta, b} \rmPi_{\theta', 1-b} \rVert$ by a small quantity.

Let $(\theta, \theta') \in \thetaset_n^2$ and $b \in \bin$.
Note $R := \{i \in \intval{1, N}\ :\ \theta_i \neq \theta'_i\}$, $T := \{i \in \intval{1, N}\ :\ \theta_i = b\}$, $T' := \{i \in \intval{1, N}\ :\ \theta'_i = 1-b\}$ and $S := \{i \in R\ :\ \theta_i = b \text{ and } \theta'_i = 1-b\}$.
We define $\pbar$ and $\qbar$ as follows:
\begin{align*}
  \pbar & := \sum \limits_{x_T \in \bin^T} \Hgate^b \ketbra{x_S}{x_S} \Hgate^b \otimes \id_{\bar{S}} \otimes B^{\theta, b}_{x_T} \otimes \id_C\\
  \qbar & := \sum \limits_{x_{T'} \in \bin^{T'}} \Hgate^{1-b} \ketbra{x_S}{x_S} \Hgate^{1-b} \otimes \id_{\bar{S}} \otimes C^{\theta', 1-b}_{x_{T'}} \otimes \id_B
\end{align*}
where $\ketbra{x_S}{x_S}$ denotes the subsystem of $\ketbra{x_T}{x_T}$ whose indices belong to $S$, and $\id_{\bar{S}}$ denotes the rest of the system.

Remark that we have:
\begin{align*}
  \lVert \rmPi_{\theta, b} \rmPi_{\theta', 1-b} \rVert^2 & = \lVert \rmPi_{\theta', 1-b} \rmPi_{\theta, b} \rmPi_{\theta', 1-b} \rVert\\
                                                        & \leq \lVert \rmPi_{\theta', 1-b} \pbar \rmPi_{\theta', 1-b} \rVert\\
                                                        & = \lVert \pbar \rmPi_{\theta', 1-b} \pbar \rVert\\
                                                        & \leq \pbar\qbar\pbar
\end{align*}
where we have the first line because $\rmPi_{\theta, b}$ is a projection, the second because $\rmPi_{\theta, b} \leq \pbar$, the third because $\rmPi_{\theta, b}$ and $\pbar$ are projections and the last because $\rmPi_{\theta', 1-b} \leq \qbar$.

Consider now the quantity $\pbar \qbar \pbar$.
We compute the following upper-bound for $\pbar \qbar \pbar$:

\begin{align*}
  \pbar\qbar\pbar & = \sum \limits_{\substack{x_T, z_T \in \bin^T\\y_{T'} \in \bin^{T'}}} \Hgate^b \ketbra{x_S}{x_S} \Hgate^b \Hgate^{1-b} \ketbra{y_S}{y_S} \Hgate^{1-b} \Hgate^b \ketbra{z_S}{z_S} \Hgate^b \otimes \id_{\bar{S}} \otimes B^{\theta, b}_{x_T} \  B^{\theta, b}_{z_T} \otimes C^{\theta', 1-b}_{y_{T'}}\\
  ~                     & = \sum \limits_{\substack{x_T \in \bin^T\\y_{T'} \in \bin^{T'}}} \Hgate^b \ketbra{x_S}{x_S} \Hgate^b \Hgate^{1-b} \ketbra{y_S}{y_S} \Hgate^{1-b} \Hgate^b \ketbra{x_S}{x_S} \Hgate^b \otimes \id_{\bar{S}} \otimes B^{\theta, b}_{x_T} \otimes C^{\theta', 1-b}_{y_{T'}}\\
  ~                     & = 2^{-\lvert S \rvert}\sum \limits_{\substack{x_T \in \bin^T\\y_{T'} \in \bin^{T'}}} \Hgate^b \ketbra{x_S}{x_S} \Hgate^b \otimes \id_{\bar{S}} \otimes B^{\theta, b}_{x_T} \otimes C^{\theta', 1-b}_{y_{T'}}\\
  ~                     & = 2^{-\lvert S \rvert}\sum \limits_{x_T \in \bin^T} \Hgate^b \ketbra{x_S}{x_S} \Hgate^b \otimes \id_{\bar{S}} \otimes B^{\theta, b}_{x_T} \otimes \id_C
\end{align*}
where the first equality comes from $B^{\theta, b}_{x_T} \  B^{\theta, b}_{z_T} = B^{\theta, b}_{x_T}$ if $x_T = z_T$ and $0$ otherwise; the second comes from $\bra{x_S}\Hgate^b \Hgate^{1-b}\ketbra{y_S}{y_S}\Hgate^{1-b} \Hgate^{b}\ket{x_S} = \lvert \bra{x_S} \Hgate \ket{y_S} \rvert^2 = 2^{-\lvert S \rvert}$ for all $x_T, y_{T'}$ and the third from $\sum_{y_{T'}} C^{\theta', 1-b}_{y_{T'}} = \id_C$.
Notice that we can assume without loss of generality that $\lvert S \rvert$ is larger than $\lvert R \rvert / 2$: if it is not the case, we just swap the roles of $\theta$ and $\theta'$.
Thus, by linearity and from $\sum_{x_{T}} B^{\theta, b}_{x_{T}} = \id_B$, it comes
$
  \lVert \pbar\qbar\pbar \rVert \leq 2^{-\lvert S \rvert} \leq 2^{-\lvert R \rvert / 2}
$
hence
\begin{equation}
  \label{eq:step3-0}
  \lVert \rmPi_{\theta, b} \rmPi_{\theta', 1-b} \rVert \leq 2^{-\lvert R \rvert / 4}
\end{equation}

\begin{remark}
  Remark that, when considering $\lVert \rmPi_{\theta, b} \rmPi_{\theta', b} \rVert$ instead, we have $S = \emptyset$.
  Thus, the reasoning above yields the trivial upper-bound
  \begin{equation}
    \label{eq:step3-1}
    \lVert \rmPi_{\theta, b} \rmPi_{\theta', b} \rVert \leq 1
  \end{equation}
\end{remark}

\paragraph{Step 4: finding the permutation family.}
In this part, we construct a family of mutually orthogonal permutations $\{\pi_{k, \alpha}\}_{k \in \intval{1, N}, \alpha \in \bin}$ such for all $k \in \intval{1, N}$, $\pi_{k, 0}$ ``flips'' the last input's bit and $\pi_{k, 1}$ leaves it unchanged.

We use the following lemma, proven in \cite{Culf2022monogamyof}.
\begin{lemma}[Lemma 3.4 of \cite{Culf2022monogamyof}]
  \label{lem:cv21-3.4}
  Let $n$ be an even integer, $\thetaset_n := \{\theta \in \bin^n\ :\ \lvert \theta \rvert = n / 2\}$ and $N = {n \choose{n/2}}$.
  Then there is a family of $N$ mutually orthogonal permutations $\{\tilde{\pi}_k\}_{k \in \intval{1, N}}$ of $\thetaset_n$ such that the following holds.
  For each $i \in \intval{1, n/2}$, there are exactly ${n/2 \choose i}^2$ permutations $\tilde{\pi}_{k}$ such that the number of positions at which $\theta$ and $\tilde{\pi}_k(\theta)$ are both $1$ is $n/2 - i$.
\end{lemma}

We prove the following corollary.
\begin{corollary}
  \label{coro:permutations}
  Let $n$ be an even integer, $\thetaset_n := \{\theta \in \bin^n\ :\ \lvert \theta \rvert = n / 2\}$ and $N = {n \choose{n/2}}$.
  Then there is a family of $2N$ mutually orthogonal permutations $\{\pi_{k, \alpha}\}_{k \in \intval{1, N}, \alpha \in \bin}$ of $\thetaset_n \times \bin$ such that the two following properties hold.
  \begin{itemize}
    \item For each $i \in \intval{1, n/2}$, there are exactly ${n/2 \choose i}^2$ permutations $\pi_{k, 0}$ such that the number of positions at which $\theta$ and $\theta'$ are both $1$ is $n/2 - i$ (\ie{} $\theta$ and $\theta'$ differ in $2i$ positions).
    \item If $\alpha = 0$, then $b' = 1 - b$. Otherwise, $b' = b$.
  \end{itemize}
  where we use the notation $(\theta' \Vert b') := \pi_{k, \alpha}(\theta \Vert b)$.
\end{corollary}

\begin{proof}
  Let $\{\tilde{\pi}_k\}_{k \in \intval{1, N}}$ be a family of orthogonal permutations promised in \cref{lem:cv21-3.4}.
  Define the family $\{\pi_{k, \alpha}\}_{k \in \intval{1, N}, \alpha \in \bin}$ as follows.
  For all $k \in \intval{1, N}$:
  \begin{align*}
    \pi_{k, 0}(\theta \vert\vert b) & = \tilde{\pi}_k(\theta) \vert\vert (1 - b) \\
    \pi_{k, 1}(\theta \vert\vert b) & = \tilde{\pi}_k(\theta) \vert\vert b
  \end{align*}

  The two properties follow directly by construction.
  It remains to prove that these $2N$ permutations are mutually orthogonal.
  Assume $\pi_{k, \alpha}(\theta) = \pi_{k', \alpha'}(\theta)$.
  Then we have $\alpha = \alpha'$, and $\tilde{\pi}_k(\theta) = \tilde{\pi}_{k'}(\theta)$, hence $k = k'$ because $\{\tilde{\pi}_k\}_k$ is a family of orthogonal permutations.
\end{proof}

\paragraph{Concluding the proof.}
We make use of the following lemma from \cite{Culf2022monogamyof}.

\begin{lemma}[Lemma 3.6 of \cite{Culf2022monogamyof}]
  Let $n \geq 2$ an integer, and note $N = {n \choose n/2}$.
  Then we have
  $$
    \frac{1}{N}\sum_{i=0}^{n/2} {n/2 \choose i}^2 2^{-i/2} \leq \sqrt{e} \left(\cos\frac{\pi}{8}\right)^n
  $$
\end{lemma}
The rest of the proof follows easily.
We first rewrite \cref{eq:step2} as
$$
  p_{win} \leq \frac{1}{2N}\sum_{k=1}^{N}\ \max_{\theta, b} \big \lVert \rmPi_{\theta, b} \rmPi_{\pi_{k, 1}(\theta, b)} \rVert + \frac{1}{2N}\sum_{k=1}^ {N}\ \max_{\theta, b} \big \lVert \rmPi_{\theta, b} \rmPi_{\pi_{k, 0}(\theta, b)} \rVert
$$
Then, by plugging the permutation's family of \cref{coro:permutations}, and using the upper-bounds proved in \cref{eq:step3-0} and \cref{eq:step3-1}, it comes
\begin{align*}
  p_{win} & \leq \frac{1}{2} + \frac{1}{2N}\sum_{i = 1}^{n/2} 2^{-i/2}\\
  ~       & \leq \frac{1}{2} + \frac{\sqrt{e}}{2} \left(\cos\frac{\pi}{8}\right)^n.
\end{align*}

\subsection{Computational Version}

We provide below a computational version of the monogamy-of-entanglement with identical basis.
The only difference is that the adversaries are given access to obfuscated membership programs for the coset space and its dual.
This game is still hard to win with probability significantly greater than $1/2$ if we make the assumption that the adversaries are polynomially bounded.
The proof of this statement follows directly from the proof of hardness of the computational version of the regular monogamy-of-entanglement game \cite{C:CLLZ21}.

\begin{definition}[Computational Monogamy-of-Entanglement Game with Identical Basis (Coset Version)]
  \label{def:new-cpt-moe-coset}
  This game is between a challenger and a triple of adversaries $(\adv, \bdv, \cdv)$ - where $\bdv$ and $\cdv$ are not communicating, and is parametrized by a security parameter $\secpar$.

  \begin{itemize}
    \item The challenger samples a subspace $A \gets \bin^{\secpar \times \frac{\secpar}{2}}$ and two vectors $(s, s') \gets \FF_2^n \times \FF_2^n$.
    Then the challenger prepares the coset state $\ket{A_{s, s'}}$ as well as two obfuscated membership programs $\obfmbr_{A+s} := \iO(A + s)$ and $\obfmbr_{A^\perp + s'} := \iO(A^\perp + s')$ and sends $\left(\ket{A_{s, s'}}, \obfmbr_{A+s}, \obfmbr_{A^\perp+s'}\right)$ to $\adv$.
    \item $\adv$ prepares a bipartite quantum state $\sigma_{12}$, then sends $\sigma_1$ to $\bdv$ and $\sigma_2$ to $\cdv$.
    \item The challenger samples $b \gets \bin$, then sends $(A, b)$ to both $\bdv$ and $\cdv$.
    \item $\bdv$ returns $u_1$ and $\cdv$ returns $u_2$.
  \end{itemize}
  For $i \in \{1, 2\}$, we say that $\adv_i$ makes a correct guess if $(b = 0\ \land\ u'_i \in A + s)$ or if $(b = 1\ \land\ u'_i \in A^\perp + s')$.
  We say that $(\adv, \bdv, \cdv)$ win the game if both $\bdv$ and $\cdv$ makes a correct guess.
  For any triple of adversaries $(\adv, \bdv, \cdv)$ and any security parameter $\secpar \in \NN$ for this game, we note $\mathsf{MoE}_{coset (comp)}(\secparam, \adv, \bdv, \cdv)$ the random variable indicating whether $(\adv, \bdv, \cdv)$ win the game or not.
\end{definition}

\begin{theorem}
  \label{th:new_moe_coset_comp}
  There exists a negligible function $\negl[\cdot]$ such that, for any triple of $\qpt$ algorithms $(\adv, \bdv, \cdv)$ and any security parameter $\secpar \in \NN$, $\prob{\mathsf{MoE}_{coset (comp)}(\secparam, \adv, \bdv, \cdv) = 1} \leq 1/2 + \negl$.
\end{theorem}

\subsection{Parallel Repetition of the Game}
\label{sec:moe-parallel}
For our proof of anti-piracy of copy-protection, we actually need a parallel version of this game, where the challenger samples $\kappa \in \NN$ independent cosets and an independent basis choice for each coset; and the adversaries are supposed to return a vector in the correct space for all the cosets to win the game.
We show that the winning probability of this game is negligible.

\begin{definition}[$\kappa$-Parallel Computational Monogamy-of-Entanglement Game with Identical Basis (Coset Version)]
  \label{def:new-moe-parallel}
  This game is between a challenger and a triple of adversaries $(\adv, \bdv, \cdv)$ - where $\bdv$ and $\cdv$ are not communicating, and is parametrized by a security parameter $\secpar$.

  \begin{itemize}
    \item The challenger samples $\kappa$ subspaces $\{A_i\}_{i \in \intval{1, \kappa}}$ and $\kappa$ pairs of vectors $\{(s_i, s'_i)\}_{i \in \intval{1, \kappa}}$ where $A_i \gets \bin^{\secpar \times \frac{\secpar}{2}}$ and $(s_i, s'_i) \gets \FF_2^n \times \FF_2^n$ for all $i \in \intval{1, \kappa}$.
    Then the challenger prepares the coset states $\{\ket{A_{i, s_i, s'_i}}\}_{i \in \intval{1, \kappa}}$ as well as the associated obfuscated membership programs $\obfmbr_{A_i+s_i} := \iO(A_i + s_i)$ and $\obfmbr_{A^\perp_i + s'_i} := \iO(A^\perp_i + s'_i)$ for $i \in \intval{1, \kappa}$; and sends $\left(\{\ket{A_{i, s_i, s'_i}}\}_{i \in \intval{1, \kappa}}, \{\obfmbr_{A_i+s_i}, \obfmbr_{A^\perp_i+s'_i}\}_{i \in \intval{1, \kappa}}\right)$ to $\adv$.
    \item $\adv$ prepares a bipartite quantum state $\sigma_{12}$, then sends $\sigma_1$ to $\bdv$ and $\sigma_2$ to $\cdv$.
    \item The challenger samples $r \gets \bin^\kappa$, then sends $\{A_i\}_{i \in \intval{1, \kappa}}$ and $r$ to both $\bdv$ and $\cdv$.
    \item $\bdv$ returns $\kappa$ vectors $\{u_i\}_{i \in \intval{1, \kappa}}$ and $\cdv$ returns $\kappa$ vectors $\{u'_i\}_{i \in \intval{1, \kappa}}$.
  \end{itemize}
  We say that $\bdv$ makes a correct guess if $(r_i = 0\ \land\ u_i \in A_i + s_i)$ or if $(r_i = 1\ \land\ u_i \in A^\perp_i + s'_i)$ for all $i \in \intval{1, \kappa}$.
  Similarly, we say that $\cdv$ makes a correct guess if $(r_i = 0\ \land\ u'_i \in A_i + s_i)$ or if $(r_i = 1\ \land\ u'_i \in A^\perp_i + s'_i)$ for all $i \in \intval{1, \kappa}$.
  We say that $(\adv, \bdv, \cdv)$ win the game if both $\bdv$ and $\cdv$ makes a correct guess.
  For any triple of adversaries $(\adv, \bdv, \cdv)$ and any security parameter $\secpar \in \NN$ for this game, we note $\kappa-\mathsf{MoE}_{coset (comp)}(\secparam, \adv, \bdv, \cdv)$ the random variable indicating whether $(\adv, \bdv, \cdv)$ win the game or not.
\end{definition}

\begin{theorem}
  \label{th:new-moe-parallel}
  There exists a negligible function $\negl[\cdot]$ such that, for any triple of $\qpt$ algorithms $(\adv, \bdv, \cdv)$ and any security parameter $\secpar \in \NN$, $\prob{\kappa-\mathsf{MoE}_{coset (comp)}(\secparam, \adv, \bdv, \cdv) = 1} \leq \negl$.
\end{theorem}

\paragraph{Comparison with \cite{cryptoeprint:2023/410}.}
In \cite{cryptoeprint:2023/410}, the authors also present a new monogamy-of-entanglement game for coset states.
Their game is similar too our parallel version except that, instead of receiving the same challenge bitstring $r$, $\bdv$ and $\cdv$ receive respectively $r_1$ and $r_2$, two independently sampled challenge bitstrings, and must answer accordingly.
Note that the hardness of the parallel version of our game can be proven using lemma 18 of \cite{C:AnaKalLiu23} on their game\footnotemark.
We still provide a direct proof for this theorem in \cref{subsec:moe-parallel} for completeness.
We emphasize that for the single-instance version, however, the same lemma cannot be applied.
\footnotetext{We thank Alper Çakan and Vipul Goyal for pointing out this shorter proof.}

\subsection{Proof of Parallel Version of the Monogamy Game}
\label{subsec:moe-parallel}

In this subsection, we prove \cref{th:new-moe-parallel}.
We do it by proving that a parallel version of the BB84 version of the monogamy game has negligible security, as the coset version follows as for the single instance.
As the proof follows the same structure as the one of \cref{th:new_moe_bb84}, we only describe here the important steps of the proof.

\paragraph{Step 1: extended non-local game.}
We first describe the extended non-local game for this parallel version of the game.
This game is between a challenger and two adversaries $\adv$ and $\bdv$, and is parametrized by a security parameter $\secpar$ and a number of repetitions $\kappa := \poly$.
  \begin{itemize}
    \item $\bdv$ and $\cdv$ jointly prepare a quantum state $\rho_{012}$ - where $\rho_0$ is composed of $\kappa$ $\secpar$-qubits registers, denoted as $\rho^1_0, \dots, \rho^\kappa_0$ - then send $\rho_0$ to the challenger.
    $\bdv$ and $\cdv$ keep $\rho_1$ and $\rho_2$ respectively.
    From this step $\bdv$ and $\cdv$ cannot communicate.
    \item For $j \in \intval{1, \kappa}$, the challenger samples $\theta^j \gets \thetaset_n$, then the challenger samples $r \gets \bin^\kappa$.
    Then, for all $i \in \intval{1, \secpar}$ and $j \in \intval{1, \kappa}$, the challenger measures the $i^{\text{th}}$ qubit of $\rho^j_0$ in computational basis if $\theta^j_i = 0$ or in Hadamard basis if $\theta^j_i = 1$.
    Let $m^j \in \bin^n$ denote the measurement outcome for every $j$.
    Finally, the challenger sends $\theta := (\theta^1, \dots, \theta^\kappa)$ and $r$ to $\bdv$ and $\cdv$.
    \item $\bdv$ returns $\{m^j_1\}_{j \in \intval{1, \kappa}}$ and $\cdv$ returns $\{m^j_2\}_{j \in \intval{1, \kappa}}$.
  \end{itemize}
  Let $m^j_{T_{r_j}} := \{m^j_i\ \vert\ \theta^j_i = r_j\}$.
  We say that $(\bdv, \cdv)$ win the game if $m^j_1 = m^j_2 = m_{T_{r_j}}$ for all $j \in \intval{1, \kappa}$.

\paragraph{Step 2: first upper-bound.}
Let $\theta = (\theta^1, ..., \theta^\kappa)$, we define $\rmPi_{\theta, r} := \bigotimes \limits_{j = 1}^\kappa \sum_{x \in \bin^n} \ketbra{x}{x}^{\theta^j} \otimes B^{\theta, r}_{x_{T_{r}}} \otimes C^{\theta, r}_{x_{T_r}}$.
We then prove in the same way as in \cref{th:new_moe_bb84} that
$$
  p_{win} \leq \frac{1}{(2N)^\kappa}\sum_{\substack{k = k_1 \Vert \dots \Vert k_\kappa\\1 \leq k_j \leq N\ \forall j\\\alpha \in \bin^\kappa}}\ \max_{\theta, r} \big \lVert \rmPi_{\theta, r} \rmPi_{\pi_{k, \alpha}(\theta, r)} \rVert
$$
where $\{\pi_{k, \alpha}\}$ is a family of mutually orthogonal permutations indexed by $k = k_1 \Vert \dots \Vert k_\kappa$ - where each $k_j \in \intval{1, N}$ - and $r \in \bin^\kappa$.

\paragraph{Step 3: upper-bound of $\lVert \rmPi_{\theta, r} \rmPi_{\theta', \bar{r}} \rVert$.}
Let $\theta = (\theta^1, ..., \theta^\kappa)$ and $\theta' = (\theta^{'1}, ..., \theta^{'\kappa})$ where each $\theta^j$ and $\theta^{'j}$ belongs to $\thetaset_n$.
Let $r \in \bin^\kappa$.
For every $j \in \intval{1, \kappa}$,
note $R^j := \{i \in \intval{1, N}\ :\ \theta^j_i \neq \theta^{'j}_i\}$, $T^j := \{i \in \intval{1, N}\ :\ \theta^j_i = r_j\}$, $T^{'j} := \{i \in \intval{1, N}\ :\ \theta^{'j}_i = 1-r_j\}$ and $S^j := \{i \in R\ :\ \theta^j_i = r_j \text{ and } \theta^{'j}_i = 1-r_j\}$.
We define $\pbar$ and $\qbar$ as follows:
\begin{align*}
  \pbar & = \sum_{\substack{j \in \intval{1, \kappa}\\x_{T^j} \in \bin^{T^j}}} \bigotimes \limits_{j = 1}^\kappa \Hgate^{r_j} \ketbra{x_{S^j}}{x_{S^j}} \Hgate^{r_j} \otimes \id_{\bar{S^j}} \otimes B^{\theta, r}_{x_{T}} \otimes \id_C\\
  \qbar & = \sum_{\substack{j \in \intval{1, \kappa}\\x_{T^{'j}} \in \bin^{T^{'j}}}} \bigotimes \limits_{j = 1}^\kappa \Hgate^{1-r_j} \ketbra{x_{S^j}}{x_{S^j}} \Hgate^{1-r_j} \otimes \id_{\bar{S^j}} \otimes \id_B \otimes C^{\theta', 1-\bar{r}}_{x_{T'}}
\end{align*}
where $T := T^1 \Vert \dots \Vert T^\kappa$, $\ketbra{x_{S^j}}{x_{S^j}}$ denotes the subsystem of $\ketbra{x_{T^j}}{x_{T^j}}$ whose indices belong to $S^j$, and $\id_{\bar{S^j}}$ denotes the rest of the system.

Following the same reasoning as in \cref{th:new_moe_bb84} (step 3), it comes
$$
  \lVert \rmPi_{\theta, r} \rmPi_{\theta', \bar{r}} \rVert \leq 2^{-\frac{\sum_j \lvert R^j \rvert}{4}}
$$

\paragraph{Step 4: finding the permutation family.}
Let $\{\pi^\star_{k, \alpha}\}_{k \in \intval{1, N}, \alpha \in \bin}$ denotes the permutation family defined in step 4 of \cref{th:new_moe_bb84}.
We define the permutation family $\{\pi_{k, \beta}\}$ - indexed by $k = k_1 \Vert \dots \Vert k_\kappa$ where each $k_j \in \intval{1, N}$ and $\beta \in \bin^\kappa$ - as $\pi_{k, r}(\theta_1 \Vert \dots \Vert \theta_\kappa, r) = \pi^\star_{k_1, \beta_1}(\theta_1, r_1) \Vert \dots \Vert \pi^\star_{k_\kappa, \beta_\kappa}(\theta_\kappa, r_\kappa)$.
It is easy to see that this family is orthogonal and has the same required properties as in the single instance proof, that is that for every $j \in \intval{1, \kappa}$ and $i \in \intval{1, n/2}$, there are exactly ${n/2 \choose i}^2$ permutations $\pi_{k, 0}$ such that the number of positions at which $\theta^j$ and $\theta^{'j}$ are both $1$ is $n/2 - i$ (\ie{} $\vert R^j \vert = 2^i$).
Using this set of permutations we have:
\begin{align*}
  p_{win} & \leq \frac{1}{(2N)^\kappa} \sum_{\substack{k = k_1 \Vert \dots \Vert k_\kappa\\\beta \in \bin^\kappa}} \max_{\substack{\theta = \theta_1 \Vert \dots, \Vert \theta_\kappa\\r \in \bin^\kappa}} \Vert \rmPi_{\theta, r} \rmPi_{\theta' ,r'} \Vert\\
  ~       & = \frac{1}{(2N)^\kappa} \sum_{w = 0}^\kappa \sum_{\substack{k = k_1 \Vert \dots \Vert k_\kappa\\\beta \in \bin^\kappa, \vert \beta \vert = w}} \max_{\substack{\theta = \theta_1 \Vert \dots, \Vert \theta_\kappa\\r \in \bin^\kappa}} \Vert \rmPi_{\theta, r} \rmPi_{\theta' ,r'} \Vert\\
  ~       & \leq \frac{1}{(2N)^\kappa} \sum_{w = 0}^\kappa \binom{\kappa}{w} \left(\sum_{\ell = 0}^{n/2} \binom{n/2}{\ell}^2 2^{-\ell / 2}\right)^w\\
  ~       & = \frac{1}{(2N)^\kappa} \left(1 + \sum_{\ell = 0}^{n/2} \binom{n/2}{\ell}^2 2^{-\ell / 2}\right)^{\kappa}\\
  ~       & \leq \frac{1}{(2N)^\kappa} \left(1 + \binom{n/2}{n/4}^2 \sum_{\ell = 0}^{n/2} 2^{-\ell / 2}\right)^{\kappa}\\
  ~       & = \frac{1}{(2N)^\kappa} \left(1 + \binom{n/2}{n/4}^2 \frac{1 - 2^{-n/4 - 1/2}}{1 - 2^{-1/2}}\right)^{\kappa}
\end{align*}

Where in the first equality, we split the sum over the possible weights of $\beta$; the first inequality comes from \cref{coro:permutations}; we obtain the second equality by applying the binomial theorem; the second inequality comes from $\binom{n}{k} \leq \binom{n}{n/2}$ for all $n ,k$; and the last inequality comes from the fact that the sum is the sum of a geometric series.

Using both Stirling approximation and asymptotic development of logarithm, we get that the logarithm of this last inequality decreases linearly in $k$, meaning that the upper bound is negligible in $n$ which concludes the proof.
\section{Conjectures on Simultaneous Compute-and-Compare Obfuscation}
\label{sec:conjecture}
In this section, we present our conjectures.
We first give an overview of the conjectures, then we define them formally, and finally we discuss their relation to similar conjectures in a recent work \cite{EPRINT:AnaBeh23}.

\subsection{Overview}
Recall that for (sub-exponentially) unpredictable distribution $\distrib$, there
exists a compute-and-compare obfuscator (\cref{sec:preli-cc}).
We are interested in whether this result still holds in a non-local context.
More precisely, consider the two following tasks, which we call \emph{simultaneous distinguishing}
and \emph{simultaneous predicting}.
Simultaneous predicting asks two players, Bob and Charlie, given a function
associated to a compute-and-compare program, and a quantum state as auxiliary
information on the program, to output the associated lock value.
Note that the function given to Bob and the one given to Charlie are not
necessarily the same, and that the same goes for the quantum states they are
given.
Also, crucially, Bob's and Charlie's quantum states can be entangled.
In simultaneous distinguishing, Bob and Charlie are given either an obfuscated
compute-and-compare program, or the outcome of a simulator on this program's
parameters\footnotemark.
As in simultaneous predicting, they are also given a quantum state each, but
here, they are asked to tell whether they received the obfuscated program, or
the simulated one.
\footnotetext{By parameters, we mean input size, output size, and depth of a
  given circuit.}

These two tasks are parameterized by a distribution over triple of the form $(\ccprog_1, \ccprog_2, \sigma_{12})$ - where the two first elements are compute-and-compare programs used to create the challenges in the challenge phase and the last one is the bipartite quantum state shared by Bob and Charlie
We say that such a distribution is simultaneously unpredictable if no adversaries can succeed in the associated simultaneous predicting task; and that simultaneous compute-and-compare obfuscation exists for this distribution if there is a compute-and-compare obfuscator with respect to which no adversaries can succeed in the associated simultaneous distinguishing task.
The question we ask now is:
\begin{center}
  \emph{\textbf{Question.}
  Is there simultaneous compute-and-compare obfuscation for any simultaneous unpredictable distribution ?}
\end{center}
As discussed in \cite{C:CLLZ21}, this question is far from trivial.
Indeed, consider its contraposition: \emph{if all candidate algorithms for simultaneous compute-and-compare obfuscation fail to obfuscate the programs as desired, does it mean that the distribution is simultaneously predictable for a certain pair of algorithms ?}
Intuitively, the difficulty here stems from whether the challenges are independent or not: if they are,
then one can analyze the two adversaries in the distinguishing game independently, and thus say that if the first adversary succeeds in their part of the task, then they can predict their lock value, and that same goes for the second adversary.
If the challenges are not independent in the other hand, it is not clear what happens when the first adversary predicts the lock value: as, concretely, the prediction is a measurement, perhaps this measurement perturbs the other register in a way that prevents the other adversary to predict their lock value.

In this work, we break down this question in the following way: we parameterize the distinguishing task by a distribution over pairs of coins used as random tape by the compute-and-compare obfuscator, and by a distribution over bits used to determine whether Bob and Charlie receive the obfuscated program or the simulated one.
We consider two types of distributions for the coins' distribution and the bits' distribution:
\begin{itemize}
  \item the uniform distribution, where the pairs $(r_1, r_2)$ (resp. $(b_1, b_2)$) are such that $r_1$ and $r_2$ (resp. $b_1$ and $b_2$) are uniformly and independently sampled;
  \item the identical distribution where the pairs $(r_1, r_2)$ (resp. $(b_1, b_2)$) are such that $r_1$ (resp. $b_1$) is uniformly  sampled and $r_2 = r_1$ (resp. $b_2 = b_1$).
  We then simply write these pairs as $(r, r)$ and $(b, b)$.
\end{itemize}

In \cite{C:CLLZ21}, the authors show that the answer of the question above is yes when both the coins' and the bits' distributions are uniform.
In particular, they use a technique called ``threshold projective implementation'' to show that, with these parameters, one can analyze the two adversaries independently, hence Bob prediction does not perturb Charlie's one.
We conjecture that this is still the case when the coins' distribution is identical and the bits' one is either uniform or also identical.

\paragraph{Relation with \cite{EPRINT:AnaBeh23}.}
In a recent work of Ananth and Behera \cite{EPRINT:AnaBeh23}, the authors make a similar conjecture, this time on simultaneous Goldreich-Levin prediction.
Roughly, the usual Goldreich-Levin theorem states that if $F$ is a one way function (meaning that a random $x$ is not predictable given $F(x)$), then no adversary can distinguish the dot product $x \cdot r$ from a random bit, given $F(x)$ - where $x$ is a random input and $r$ a random bitstring of same length.
The authors of \cite{EPRINT:AnaBeh23} consider the simultaneous version of this task, that is, assuming that $(x_1, x_2)$ are simultaneously unpredictable given $(F(x_1), F(x_2))$ (in the same sense as our definitons above), then $x_1 \cdot r_1$ and $x_2 \cdot r_2$ are simultaneously indistinguishable from two random strings - where the pairs $(x_1, x_2)$, $(r_1, r_2)$, and $(b_1, b_2)$, are sampled from different types of distributions (uniform or identical), similarly as in our case - and they finally describe two conjectures.
Note that, as there is a construction of compute-and-compare obfuscation \cite{C:CLLZ21} (based on $\iO$ and hardness of LWE assumptions) that ultimately relies on the Goldreich-Levin theorem, then we expect that the conjectures of \cite{EPRINT:AnaBeh23}, combined with $\iO$ and LWE assumptions, imply our conjectures.

\subsection{Definitions}
We present formally the notions of simultaneous distinguishing and simultaneous predicting games.
For ease of reading, we first introduce what we call simultaneous compute-and-compare distributions.

\begin{definition}[Simultaneous Compute-and-Compare Distribution]
  We call \emph{simultaneous compute-and-compare distribution} a family of distributions $\distrib_\ccprog = \{\distrib_\secpar\}_{\secpar \in \NN}$ over triple of the form $(\ccprog_1[f_1, y_1, m_1],\allowbreak \ccprog_2[f_2, y_2, m_2], \sigma_{12})$ where $\ccprog_1[f_1, y_1, m_1]$ and $\ccprog_2[f_2, y_2, m_2]$ are both compute-and-compare programs, and $\sigma_{12}$ is a bipartite quantum state representing some auxiliary information.
  In the following, we denote the first and second registers of $\sigma_{12}$ as $\sigma_1$ and $\sigma_2$.
\end{definition}

We are now ready to define simultaneous distinguishing and predicting games.

\begin{definition}[Simultaneous Distinguishing Game]
  We define below a \emph{simultaneous distinguishing game}, parameterized by a pair of efficient algorithms $(\ccobf, \simul)$ - where $\ccobf$ uses a bitstring $r \in \bin^{\ell}$ as random coins for some $\ell \in \NN$, a simultaneous compute-and-compare distribution $\distrib_\ccprog$, a ``coins' distribution'' $\distrib_R$ over $\bin^{2 \ell}$, a ``bits' distribution'' $\distrib_B$ over $\bin^2$, and a security parameter $\secpar$.
  This game is between a challenger and a pair of adversaries $\bdv$ and $\cdv$.
  \label{def:simul-dist}
  \begin{itemize}
    \item The challenger samples $(\ccprog_1, \ccprog_2, \sigma_{12}) \gets \distrib_\ccprog$, $(b_1, b_2) \gets \distrib_B$ and $(r_1, r_2) \gets \distrib_R$.
    \item The challenger sends $\sigma_1$ to $\bdv$.
    The challenger also sends $\ccobf(\ccprog_1; r_1)$ if $b_1 = 0$, or $\simul(\secparam, \ccprog_1.\params)$ if $b_1 = 1$ to $\bdv$.
    \item Similarly, the challenger sends $\sigma_2$ to $\cdv$.
    Then they also send $\ccobf(\ccprog_2; r_2)$ if $b_2 = 0$, or $\simul(\secparam, \ccprog_2.\params)$ if $b_2 = 2$ to $\cdv$.
  \end{itemize}
  $\bdv$, and $\cdv$ win the game if $\bdv$ returns $b'_1 = b_1$ and $\cdv$ returns $b'_2 = b_2$.

  We denote the random variable that indicates whether a pair of adversaries $(\bdv, \cdv)$ wins the game or not as $\funcfont{Simul-Dist}^{(\ccobf, \simul)}_{\distrib_\ccprog, \distrib_R, \distrib_B}(\secparam, \bdv, \cdv)$.
\end{definition}

\begin{definition}[Simultaneous Predicting Game]
  We define below a \emph{simultaneous predicting game}, parametrized by a simultaneous compute-and-compare distribution $\distrib_\ccprog$, and a security parameter $\secpar$.
  This game is between a challenger and a pair of adversaries $\bdv$ and $\cdv$.
  \label{def:simul-predict}
  \begin{itemize}
      \item The challenger samples $(\ccprog_1[f_1, y_1, m_1], \ccprog_2[f_2, y_2, m_2], \sigma_{12}) \gets \distrib_\ccprog$.
      \item Then, the challenger sends $(f_1, \sigma_1)$ to $\bdv$ and $(f_2, \sigma_2)$ to $\cdv$.
  \end{itemize}
  $\bdv$, and $\cdv$ win the game if $\bdv$ returns $y'_1 = y_1$ and $\cdv$ returns $y'_2 = y_2$.

  We denote the random variable that indicates whether a pair of adversaries $(\bdv, \cdv)$ wins the game or not as $\funcfont{Simul-Predict}_{\distrib_\ccprog}(\secparam, \bdv, \cdv)$.
\end{definition}

\subsection{Conjectures}
We now state our two conjectures.
An informal description of the conjectures is illustrated in \cref{fig:cc-dist}.

\begin{conjecture}
  \label{conj:uniformB}
  Let $\distrib_\ccprog$ a simultaneous compute-and-compare distribution, $\distrib_R$ the identical distribution over $\bin^{2 \ell}$, and $\distrib_B$ the uniform distribution over $\bin^2$.
  Assume that, for all pair of $\qpt$ adversaries $(\bdv, \cdv)$,
  $$
    \prob{\funcfont{Simul-Predict}_{\distrib_\ccprog}(\secparam, \bdv, \cdv) = 1} \leq \negl
  $$
  Then, there exists a compute-and-compare obfuscator $\ccobf$ and its associated simulator $\simul$ such that, for all pair of $\qpt$ adversaries $(\bdv, \cdv)$,
  $$
    \prob{\funcfont{Simul-Dist}^{(\ccobf, \simul)}_{\distrib_\ccprog, \distrib_R, \distrib_B}(\secparam, \bdv, \cdv) = 1} \leq \frac{1}{2} + \negl
  $$
\end{conjecture}

\begin{conjecture}
  \label{conj:identicalB}
  Let $\distrib_\ccprog$ a simultaneous compute-and-compare distribution, $\distrib_R$ the identical distribution over $\bin^{2 \ell}$, and $\distrib_B$ the identical distribution over $\bin^2$.
  Assume that, for all pair of $\qpt$ adversaries $(\bdv, \cdv)$,
  $$
    \prob{\funcfont{Simul-Predict}_{\distrib_\ccprog}(\secparam, \bdv, \cdv) = 1} \leq \negl
  $$
  Then, there exists a compute-and-compare obfuscator $\ccobf$ and its associated simulator $\simul$ such that, for all pair of $\qpt$ adversaries $(\bdv, \cdv)$,
  $$
    \prob{\funcfont{Simul-Dist}^{(\ccobf, \simul)}_{\distrib_\ccprog, \distrib_R, \distrib_B}(\secparam, \bdv, \cdv) = 1} \leq \frac{1}{2} + \negl
  $$
\end{conjecture}

As we in fact use the contrapositions of these conjectures in the following of the paper, we present these contrapositions as the following corollaries.

\begin{corollary}
  \label{coro:conj-uniformB}
  Let $\distrib_\ccprog$ a simultaneous compute-and-compare distribution, $\distrib_R$ the identical distribution over $\bin^{2 \ell}$, and $\distrib_B$ the uniform distribution over $\bin^2$.
  Assume that, for all efficient and functionality preserving algorithm $\ccobf$, and for all efficient simulator $\simul$, there exists a pair of $\qpt$ adversaries $(\bdv, \cdv)$ winning the distinguishing game - parametrized by $\distrib_\ccprog, \distrib_R$, and $\distrib_B$ - with non-negligible advantage over $1/2$.
  Then there exists a pair of $\qpt$ adversaries $(\bdv, \cdv)$ winning the associated predicting game with non-negligible probability.
\end{corollary}

\begin{corollary}
  \label{coro:conj-identicalB}
  Let $\distrib_\ccprog$ a simultaneous compute-and-compare distribution, $\distrib_R$ the identical distribution over $\bin^{2 \ell}$, and $\distrib_B$ the identical distribution over $\bin^2$.
  Assume that, for all efficient and functionality preserving algorithm $\ccobf$, and for all efficient simulator $\simul$, there exists a pair of $\qpt$ adversaries $(\bdv, \cdv)$ winning the distinguishing game - parametrized by $\distrib_\ccprog, \distrib_R$, and $\distrib_B$ - with non-negligible advantage over $1/2$.
  Then there exists a pair of $\qpt$ adversaries $(\bdv, \cdv)$ winning the associated predicting game with non-negligible probability.
\end{corollary}

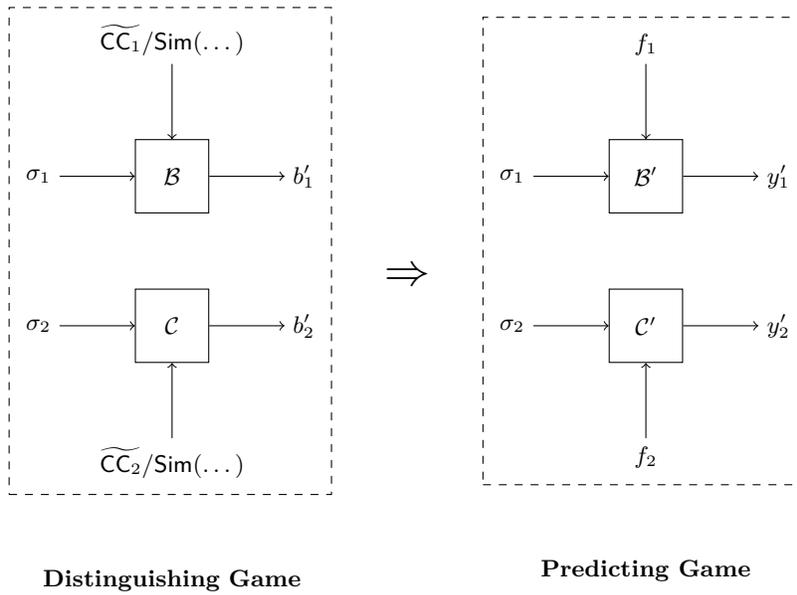
\begin{figure}
  \centering
\begin{tikzpicture}
  \node[draw, square, minimum width={width("AAAA")}] (Dbob) {$\bdv$};
  \node[draw, square, minimum width={width("AAAA")}, below=of Dbob] (Dcharlie) {$\cdv$};
  \node[above=of Dbob] (c1) {$\widetilde{\ccprog_1} / \simul(\dots)$};
  \node[below=of Dcharlie] (c2) {$\widetilde{\ccprog_2} / \simul(\dots)$};
  \draw[->] (c1) -- (Dbob);
  \draw[->] (c2) -- (Dcharlie);
  \node[left=of Dbob] (Ds1) {$\sigma_1$};
  \node[left=of Dcharlie] (Ds2) {$\sigma_2$};
  \draw[->] (Ds1) -- (Dbob);
  \draw[->] (Ds2) -- (Dcharlie);
  \node[right=of Dbob] (b1) {$b'_1$};
  \node[right=of Dcharlie] (b2) {$b'_2$};
  \draw[->] (Dbob) -- (b1);
  \draw[->] (Dcharlie) -- (b2);
  \node[draw, dashed, fit=(Ds1)(c1)(Ds2)(c2)(b1)] {};
  \node[below=of c2] {\textbf{Distinguishing Game}};

  \node[right=of b1] (rightofb1) {};
  \node[below=of rightofb1] {\LARGE$\Rightarrow$};

  \node[right=of rightofb1] (Ps1) {$\sigma_1$};
  \node[draw, square, minimum width={width("AAAA")}, right=of Ps1] (Pbob) {$\bdv'$};
  \node[draw, square, minimum width={width("AAAA")}, below=of Pbob] (Pcharlie) {$\cdv'$};
  \node[above=of Pbob] (f1) {$f_1$};
  \node[below=of Pcharlie] (f2) {$f_2$};
  \draw[->] (f1) -- (Pbob);
  \draw[->] (f2) -- (Pcharlie);
  \node[left=of Pcharlie] (Ps2) {$\sigma_2$};
  \draw[->] (Ps1) -- (Pbob);
  \draw[->] (Ps2) -- (Pcharlie);
  \node[right=of Pbob] (y1) {$y'_1$};
  \node[right=of Pcharlie] (y2) {$y'_2$};
  \draw[->] (Pbob) -- (y1);
  \draw[->] (Pcharlie) -- (y2);
  \node[draw, dashed, fit=(Ps1)(f1)(Ps2)(f2)(y1)] {};
  \node[below=of f2] {\textbf{Predicting Game}};
\end{tikzpicture}   \caption{Contraposition of the conjectures: if $\bdv$ and $\cdv$ win the distinguishing game on the left with significant advantage over $1/2$, then there exist $\bdv'$ and $\cdv'$ winning the predicting game on the right with non-negligible probability.
  $\widetilde{\ccprog_1}$ and $\widetilde{\ccprog_2}$ represent the compute-and-compare obfuscation of $\ccprog_1$ and $\ccprog_2$ with the same random coins.}
  \label{fig:cc-dist}
\end{figure}
\section{Single-Decryptor and Copy-Protection of Pseudorandom Functions}
\label{sec:sd}
In this section, we recall the notions of single-decryptor \cite{EPRINT:GeoZha20} and copy-protection of pseudorandom functions \cite{C:CLLZ21}.
These primitives are used later to prove the security of our constructions of copy-protection of point functions and unclonable encryption.
In \cite{C:CLLZ21}, the authors give a definition of anti-piracy security and provide a secure construction for these two primitives.
We give two variants of anti-piracy security of single-decryptor and of anti-piracy security of copy-protection of pseudorandom functions and show that their constructions are secure with respect to these two variants.

\subsection{Definition of a Single-Decryptor}
\begin{definition}[Single-Decryptor Encryption Scheme]
  A single-decryptor encryption scheme is a tuple of algorithms
  \(\scheme{\setup, \qkeygen, \allowbreak \enc, \dec}\) with the
  following properties:
  \begin{itemize}
    \item \((\sk, \pk) \leftarrow \setup(1^{\secpar})\).
      On input a security parameter \(\secpar\), the classical setup algorithm \(\setup\) outputs a classical secret key \(\sk\) and a public key \(\pk\).
    \item \(\qkey_{\sk} \leftarrow \qkeygen(\sk)\).
      On input a classical secret key \(\sk\), the quantum key generation
      algorithm \(\qkeygen\) outputs a quantum secret key \(\qkey_{\sk}\).
    \item \(c \leftarrow \enc(\pk, m)\).
      On input a public key \(\pk\) and a message~\(m\) in the message space
      \(\mspace\), the classical randomized encryption
      algorithm \(\enc\) outputs a classical ciphertext \(c\).
      We sometimes write \(\enc(\pk, m; r)\) to precise that we use the random bitstring $r$ as the randomness in the algorithm.
    \item \(m / \bot \leftarrow \dec(\qkey_{\sk}, c)\).
      On input a quantum secret key \(\qkey_{\sk}\), a classical ciphertext \(y\),
      the quantum decryption algorithm \(\dec\) outputs a classical message
      \(m\) or a decryption failure symbol \(\bot\).
  \end{itemize}
\end{definition}

\paragraph{Correctness.}
  We say that a single-decryptor scheme $\scheme{\setup, \qkeygen, \allowbreak \enc, \dec}$ has correctness if there exists a negligible function \(\negl[\cdot]\), such that for all
  \(\secpar \in \NN\), for all \(m \in \mspace\), the following holds:
  \begin{align*}
    \pr\left[
      \begin{array}{l}
        \dec(\qkey_{\sk}, c) = m
      \end{array}
      \ \middle\vert\
      \begin{array}{r}
        (\sk, \pk) \gets \setup(1^{\secpar}) \\
        \qkey_{\sk} \gets \qkeygen(\sk) \\
        c \leftarrow \enc(\pk, m)
      \end{array}
      \right] \geq 1 - \negl.
    \end{align*}

Note that correctness implies that an honestly generated quantum decryption key
can be used to decrypt correctly polynomially many times, from the gentle
measurement lemma~\cite{wilde2011classical}.

\paragraph{Anti-piracy security.}
We now define indistinguishable anti-piracy security of a single-decryptor scheme.
This notion is defined through a piracy game, in which a first adversary Alice is given a quantum decryption key and must split it and share it between two other adversaries, Bob and Charlie.
Bob and Charlie then receive an encryption of either the first, or the second message of a pair $(m_0, m_1)$ - known by the three adversaries - as a challenge, and must guess the encryption of which message they were given.
The game (and hence the anti-piracy security) is defined with respect to two distributions: $\distrib_B$ that yields two bits deciding which message will be encrypted for each test, and $\distrib_R$ that yields two strings to be used as the randomness for the encryption of each challenge.
In order to prove the security of our unclonable encryption and copy-protection schemes, we need to consider the security when $\distrib_R$ is the identical distribution\footnotemark and $\distrib_B$ is either the uniform or the identical distribution.
We denote the case where $\distrib_B$ is the uniform distribution as anti-piracy with respect to \emph{product distribution}, and the case where it is the identical distribution as anti-piracy with respect to \emph{identical distribution}.
\footnotetext{Recall (\cref{subsec:preli-distribs}) that a distribution is identical if it yields a pair of identical elements $(x, x)$ (where $x$ is sampled uniformly at random).}

\begin{remark}
  Note that the original anti-piracy security proposed in \cite{C:CLLZ21} is simply our definition where $\distrib_B$ and $\distrib_R$ are both uniform distributions.
\end{remark}

\begin{definition}[Piracy Game for Single-Decryptor]
  \label{def:sd-piracy-game}
  We define below a \emph{piracy game for single-decryptor}, parametrized by a single-decryptor scheme \(\mathcal{E} = \scheme{\setup, \qkeygen, \enc, \dec}\), and a security parameter \(\secpar\).
  This game is between a challenger and a triple of adversaries \((\adv, \bdv, \cdv)\).
  As the two variant of the games (with respect to product or identical distribution) differ only in the challenge phase, we describe below a different challenge phase for each variant.

  \begin{itemize}
    \item \textbf{Setup phase:}
    \begin{itemize}
      \item The challenger samples \((\sk, \pk) \gets \setup(\secparam)\).
      \item The challenger samples \(\qkey_\sk \gets \qkeygen(\sk)\).
      \item The challenger sends \((\pk, \qkey_\sk)\) to \(\adv\).
    \end{itemize}
    \item \textbf{Splitting phase:}
      $\adv$ prepares a bipartite quantum state $\sigma_{12}$, then sends $\sigma_1$ to $\bdv$, $\sigma_2$ to $\cdv$, and two pairs of messages $(m^1_0, m^1_1)$ and $(m^2_0, m^2_1)$ to the challenger.
      \item \textbf{Challenge phase (product distribution):}
      \begin{itemize}
        \item The challenger samples $(b_1, b_2) \sample \bin$, and $(r_1, r_2) \sample \bin^{\poly}$.
        \item The challenger sends $\enc(\pk, m^1_{b_1}; r_1)$ to $\bdv$, and $\enc(\pk, m^2_{b_2}; r_2)$ to $\cdv$.
      \end{itemize}
      \item \textbf{Challenge phase (identical distribution):}
      \begin{itemize}
        \item The challenger samples $b \sample \bin$, and $r \sample \bin^{\poly}$.
        \item The challenger sends $\enc(\pk, m^1_{b}; r)$ to $\bdv$, and $\enc(\pk, m^2_{b}; r)$ to $\cdv$.
      \end{itemize}
  \end{itemize}
  $\adv$, $\bdv$, and $\cdv$ win the game in the product distribution if $\bdv$ returns $b'_1 = b_1$ and $\cdv$ returns $b'_2 = b_2$; and in the identical distribution if both $\bdv$ and $\cdv$ return $b$.

  We denote the random variable that indicates whether a triple of adversaries
  \((\adv, \bdv, \cdv)\) win the game or not as \(\funcfont{SD-AP}^{\mathcal{E}}_{PD}(\secparam, \adv, \bdv, \cdv)\) or \(\funcfont{SD-AP}^{\mathcal{E}}_{ID}(\secparam, \adv, \bdv, \cdv)\) depending on which variant of the game we consider ($\text{PD}$ and $\text{ID}$ respectively denote the product and identical distributions).
\end{definition}

\begin{definition}[Indistinguishable Anti-Piracy Security]
  \label{def:sd-ap-cpa}
  A single-decryptor scheme $\mathcal{E}$ has \emph{indistinguishable anti-piracy security} with respect to the product distribution if no $\qpt$ adversary can win the piracy game above (with the product distribution challenge phase) with probability significantly greater than $1/2$.
  More precisely, for any triple of \(\qpt\) adversaries \((\adv, \bdv, \cdv)\):
  \[
    \prob{\funcfont{SD-AP}^{\mathcal{E}}_{PD}(\secparam, \adv, \bdv, \cdv) = 1} \leq 1/2 + \negl.
  \]

  Similarly, a single-decryptor scheme $\mathcal{E}$ has \emph{indistinguishable anti-piracy security} with respect to the identical distribution if no $\qpt$ adversary can win the piracy game above (with the identical distribution challenge phase) with probability significantly greater than $1/2$.
  More precisely, for any triple of \(\qpt\) adversaries \((\adv, \bdv, \cdv)\):
  \[
    \prob{\funcfont{SD-AP}^{\mathcal{E}}_{ID}(\secparam, \adv, \bdv, \cdv) = 1} \leq 1/2 + \negl.
  \]
\end{definition}

\subsection{Construction of Single-Decryptor}
In this section, we present the single-decryptor construction of \cite{C:CLLZ21}.

\begin{construction}{\cite{C:CLLZ21} Single-Decryptor Scheme}{sd}
  Given a security parameter \(\secpar\), let \(n := \secpar\) and \(\kappa\) be
  polynomial in \(\secpar\).
  \begin{itemize}
    \item \((\sk, \pk) \gets \setup(\secparam)\): ~
    \begin{itemize}
      \item Sample coset spaces \(\{A_i, s_i, s'_i\}_{i \in \intval{1, \kappa}}\) where
        each \(A_i\) is of dimension \(n / 2\);
      \item Construct the membership programs for each coset
        \(\{\obfmbr_{A_i + si}, \obfmbr_{A^\perp_i + s'i}\}_{i \in \intval{1, \kappa}}\);
      \item Return
        \(\left(\sk \coloneqq \{A_i, s_i, s'_i\}_{i \in \intval{1, \kappa}}, \pk \coloneqq \{\obfmbr_{A_i + si}, \obfmbr_{A^\perp_i + s'i}\}_{i \in \intval{1, \kappa}}\right)\).
    \end{itemize}
    \item \(\qkey_\sk \gets \qkeygen(\sk)\): ~
    \begin{itemize}
      \item Parse $\sk$ as \(\{A_i, s_i, s'_i\}_{i \in \intval{1, \kappa}}\);
      \item Return \(\bigotimes\limits_{i = 1}^\kappa \ket{A_{i, s_i, s'_i}}\).
    \end{itemize}
    \item \(c \gets \enc(\pk, m)\): ~
    \begin{itemize}
      \item Parse $\pk$ as \(\{\obfmbr_{A_i + si}, \obfmbr_{A^\perp_i + s'i}\}_{i \in \intval{1, \kappa}}\);
      \item Sample \(r \sample \bin^\kappa\);
      \item Generate an obfuscated program \(\iO(\Qmr)\) of program \(\Qmr\)
        described in~\cref{fig:sd-prog_p}.
      \item Return \(c \coloneqq \left(r, \iO(\Qmr)\right)\).
    \end{itemize}
    \item \(m/\bot \gets \dec(\qkey_{\sk}, c)\): ~
    \begin{itemize}
      \item Parse $\qkey_\sk$ as \(\bigotimes\limits_{i = 1}^\kappa \ket{A_{i, s_i, s'_i}}\) and
        \(c \gets \left(r, \iO(\Qmr)\right)\);
      \item For all \(i \in \intval{1, \kappa}\): if \(r_i = 1\), apply
        \(\Hgate^{\otimes n}\) to \(\ket{A_{i, s_i, s'_i}}\);
      \item Let \(\rho'\) be the resulting state, run \(\iO(\Qmr)\) coherently
        on \(\rho'\) and measure the final register to get \(m\);
      \item Return \(m\).
    \end{itemize}
  \end{itemize}
\end{construction}

\begin{figure}[!htbp]
  \label{fig:sd-prog_p}
  \centering
  \begin{tcolorbox}[arc=0pt,outer arc=0pt,colback=white,
    boxrule=0.2mm,left=0pt,right=0pt,top=0pt,bottom=0pt,hbox]
    \begin{varwidth}[t]{0.8\textwidth}
      \textbf{Hardcoded:} Programs $\{\progfont{P}_i\}_{i \in \intval{1, \kappa}}$ such that for all \(i \in \intval{1, \kappa}\):
      $\progfont{P}_i := \left\{\begin{array}{ll}
        \obfmbr_{A_i + si}        & \text{ if } r_i = 0\\
        \obfmbr_{A^\perp_i + s'i} & \text{ if } r_i = 1
      \end{array}\right.$.

      On input vectors \(u_1, u_2, \dots, u_\kappa\), do the following:
      \begin{enumerate}
        \item If for all \(i \in \intval{1, \kappa}\): \(\progfont{P}_i(u_i) = 1\), then
          output \(m\).
        \item Otherwise: output \(\bot\).
      \end{enumerate}
    \end{varwidth}
  \end{tcolorbox}
  \caption{Program \(\Qmr\).}
\end{figure}

\begin{remark}
  Note that the underlying $\iO$ algorithm used in the encryption algorithm of \cref{constr:sd} might use a random tape.
  In the following, we denote by $\enc(\pk, m; (r_{\iO}, r))$ the encryption of a message $m$ with the key $\pk$ and with random coins $r_{\iO}$ and $r$ respectively used for the $\iO$ algorithm and for the program $\Qmr$.
\end{remark}

\begin{theorem}
  \label{th:sd-ap-ir}
  Assuming the existence of post-quantum indistinguishability obfuscation,
  one-way functions, compute-and-compare obfuscation for the class of
  unpredictable distributions, and \cref{conj:uniformB}, \cref{constr:sd} has \emph{indistinguishable anti-piracy security with respect to the product distribution}.
\end{theorem}

\begin{theorem}
  \label{th:sd-ap-id}
  Assuming the existence of post-quantum indistinguishability obfuscation,
  one-way functions, compute-and-compare obfuscation for the class of
  unpredictable distributions, and \cref{conj:identicalB}, \cref{constr:sd} has \emph{indistinguishable anti-piracy security with respect to the identical distribution}.
\end{theorem}

\subsection{Proof of \cref{th:sd-ap-ir}}
\label{sec:sd-ap-proof}

In this section, we prove \cref{th:sd-ap-ir}.
Our proof follows the structure of \cite{C:CLLZ21}.
We proceed in the proof through a sequence of hybrids.
For any pair of hybrids \((\pcgamename_i, \pcgamename_j)\), we say that \(\pcgamename_i\) is \(\emph{negligibly close to}\) \(\pcgamename_j\) if for triple of \(\qpt\) adversaries \((\adv, \bdv, \cdv)\), the probability that \((\adv, \bdv, \cdv)\) wins \(\pcgamename_i\) is negligibly close to the probability that they win \(\pcgamename_j\).

\begin{gamedescription}[nr=-1,arg=]
  \describegame This is the piracy game for the single-decryptor of \cref{constr:sd}, with respect to the product distribution.

  \begin{itemize}
    \item \textbf{Setup phase:}
    \begin{itemize}
      \item The challenger samples coset spaces \(\{A_i, s_i, s'_i\}_{i \in \intval{1, \kappa}}\) where
      each \(A_i\) is of dimension \(n / 2\).
      \item Then the challenger constructs the membership programs for each coset
      \(\{\obfmbr_{A_i + si}, \obfmbr_{A^\perp_i + s'i}\}_{i \in \intval{1, \kappa}}\).
      \item Finally, the challenger sends $\qkey_\sk := \{\ket{A_{i, s_i, s'_i}}\}_{i \in \intval{1, \kappa}}$ and $\pk := \{\obfmbr_{A_i + si}, \obfmbr_{A^\perp_i + s'i}\}_{i \in \intval{1, \kappa}}$ to $\adv$.
    \end{itemize}
    \item \textbf{Splitting phase:}
      $\adv$ prepares a bipartite quantum state $\sigma_{12}$, then sends $\sigma_1$ to $\bdv$, $\sigma_2$ to $\cdv$, and two pairs of messages $(m^1_0, m^1_1)$ and $(m^2_0, m^2_1)$ to the challenger.
    \item \textbf{Challenge phase:}
    \begin{itemize}
      \item The challenger samples random coins $r_{\iO} \sample \bin^{\poly}$ - to be used in the $\iO$ algorithm, and $r \sample \bin^{\poly}$ - to be used in the encryption algorithm, and two bits uniformly at random $b_1, b_2 \sample \bin$.
      \item The challenger computes $c_1 := (r, \progfont{Q}_1) \gets \enc(\pk, m^1_{b_1}; (r, r_\iO))$, and $c_2 := (r, \progfont{Q}_2) \gets \enc(\pk, m^2_{b_2}; (r, r_\iO))$ (note that the programs $\progfont{Q}_1$ and $\progfont{Q}_2$ have been obfuscated using $r_\iO$ as the randomness).
      \item The challenger sends $c_1$ to $\bdv$, and $c_2$ to $\cdv$.
    \end{itemize}
  \end{itemize}
  $\adv$, $\bdv$, and $\cdv$ win the game if $\bdv$ returns $b'_1 = b_1$ and $\cdv$ returns $b'_2 = b_2$.

  \describegame In this second hybrid, we replace the obfuscated programs $\progfont{Q}_1$ and $\progfont{Q}_2$ by obfuscated compute-and-compare programs.
  More formally, for $i \in \intval{1, \kappa}$, we define\footnotemark
  $$
  \can_{i, b}(\cdot) := \left\{\begin{array}{ll}
    \can_{A_i}(\cdot) &\text{ if } b = 0\\
    \can_{A^\perp_i}(\cdot) &\text{ if } b = 1
  \end{array}\right.
  \text{ and }
  c_{i, b} := \left\{\begin{array}{ll}
    \can_{A_i}(s_i) &\text{ if } b = 0\\
    \can_{A^\perp_i}(s'_i) &\text{ if } b = 1
  \end{array}\right.
  $$
  We similarly define $\can_r(u_1, ..., u_\kappa) = (\can_{1, r_1}(u_1), \dots, \can_{\kappa, r_\kappa}(u_\kappa))$ and $c_r = (c_{1, r_1}, \dots, c_{\kappa, r_\kappa})$ for any $r \in \bin^\kappa$.
  Finally, we write $\ccprog_1$ and $\ccprog_2$ to denote $\ccprog[\can_r, c_r, m^1_{b_1}]$ and $\ccprog[\can_r, c_r, m^2_{b_2}]$.

  Then, we replace $\progfont{Q}_1$ by $\iO(\ccprog_1)$ and $\progfont{Q}_2$ by $\iO(\ccprog_2)$.
  Because the programs $\progfont{Q}_1$ and $\progfont{Q}_2$ are respectively functionally equivalent to $\ccprog_1$ and $\ccprog_2$, then from $\iO$ security, $\pcgamename_0$ and $\pcgamename_1$ are negligibly close.
  \footnotetext{Recall that for a subspace $A$ and a vector $u$, $\can_A(u)$ - defined in \cref{par:can-coset} - is the coset representative of $A + u$.
  Recall also that $\can_A$ can be efficiently implemented given a description of $A$.}

  \describegame In this last hybrid, we replace $\iO(\ccprog_1)$ by $\iO(\ccobf(\secparam, \ccprog_1))$ and $\iO(\ccprog_2)$ by $\iO(\ccobf(\secparam, \ccprog_2))$.
  Because the programs $\ccprog_1$ and $\ccprog_2$ are respectively functionally equivalent to $\ccobf(\secparam, \ccprog_1)$ and $\ccobf(\secparam,\allowbreak \ccprog_2)$, then from $\iO$ security, $\pcgamename_1$ and $\pcgamename_2$ are negligibly close.
\end{gamedescription}

\paragraph{Leveraging compute-and-compare obfuscation.}
Before proceeding to the reduction, we introduce the two following lemmas.

\begin{lemma}
  \label{lem:cc-extract}
  Define the simultaneous compute-and-compare distribution $\distrib^\adv_\ccprog$, parametrized with a $\qpt$ algorithm $\adv$ for the hybrid $\pcgamename_2$ as follows:
  \begin{itemize}
    \item sample $\kappa$ cosets descriptions $(A_i, s_i, s'_i)_{i \in \intval{1, \kappa}}$;
    \item run $\adv$ on $\otimes_{i=1}^{\kappa}\ket{A_i, s_i, s'_i}$ to get $\sigma_{12}$ and $(m^1_0, m^1_1), (m^2_0, m^2_1)$;
    \item sample $r \sample \intval{1, \kappa}$ and $b_1, b_2 \sample \bin$;
    \item define the bipartite quantum state $\sigma'_{12}$ with $\sigma_1 \otimes \ketbra{b_1}{b_1}$ as first register and $\sigma_2 \otimes \ketbra{b_2}{b_2}$ as second register;
    \item return $\left(\ccprog[\can_r, c_r, m^1_{b_1}], \ccprog[\can_r, c_r, m^2_{b_2}], \sigma'_{12}\right)$.
  \end{itemize}

  Assume in addition that a triple of $\qpt$ algorithms $(\adv, \bdv, \cdv)$ win hybrid $\pcgamename_2$ with non-negligible advantage over $1/2$.
  Then there exists a pair of $\qpt$ algorithms $(\bdv', \cdv')$ that win the simultaneous predicting game, parametrized with $\distrib^\adv_\ccprog$, with non-negligible probability.
\end{lemma}

\begin{proof}
  The proof follows from the contraposition of \cref{conj:uniformB}.
  We construct a pair of $\qpt$ adversaries $(\bdv', \cdv')$ for the simultaneous distinguishing game parametrized with any efficient and functionality preserving $\ccobf$, any efficient simulator $\simul$, the simultaneous compute-and-compare distribution $\distrib^\adv_\ccprog$, the identical coins' distribution $\distrib_R$, and the uniform bits' distribution $\distrib_B$.
  \begin{itemize}
    \item $\bdv'$ receives the program $\progfont{C_1}$ from the challenger: $\progfont{C_1}$ is either a compute-and-compare obfuscation $\ccobf(\ccprog_1; r)$ - where $\ccprog_1 := \ccprog[\can_r, c_r, m^1_{b_1}]$ - or a simulated program $\simul(\ccprog_1.\params)$.
    $\bdv'$ also receives $\sigma_1 \otimes \ketbra{b_1}{b_1}$.
    \item $\bdv'$ runs $\bdv$ on $(\sigma_1, \progfont{C_1})$ to get the outcome $b'_1$.
    \item If $b'_1 = b_1$, $\bdv'$ returns $0$, otherwise $\bdv'$ return $1$.
  \end{itemize}
  $\cdv'$ is defined similarly by replacing the ``$1$'' indices by ``$2$'' indices.

  Because, $(\adv, \bdv, \cdv)$ win $\pcgamename_2$ with non-negligible advantage over $1/2$, and, when $\progfont{C_1}$ (resp.\@ $\progfont{C_2}$) is the obfuscated program, the challenge given to $\bdv$ (resp.\@ $\cdv$) comes from the same distribution as in $\pcgamename_2$, then $\bdv$ (resp.\@ $\cdv$) guesses $b_1$ (resp.\@ $b_2$) correctly with non-negligible advantage over $1/2$.
  On the other hand, when $\progfont{C_1}$ (resp.\@ $\progfont{C_2}$) is simulated, then it does not hold any information on $b_1$ (resp.\@ $b_2$), hence $\bdv$ (resp.\@ $\cdv$) guesses correctly only with probability $1/2$.
  Thus, $\bdv'$ and $\cdv'$ succeed in simultaneously distinguishing the simulated programs from the obfuscated ones with non-negligible advantage over $1/2$.
  Because this reasoning holds for all efficient functionality preserving $\ccobf$ and simulator $\simul$, then there is no compute-and-compare obfuscator for the distribution $\distrib_\ccprog$.
  Using the contraposition of \cref{conj:uniformB}, completes the proof.
\end{proof}

\paragraph{Reduction to monogamy-of-entanglement.}
We are now ready to proceed to the reduction.
Assume that there exists a triple of $\qpt$ algorithms $(\adv, \bdv, \cdv)$ that win the last hybrid $\game_2$ with non-negligible advantage over $1/2$.
Then, by \cref{lem:cc-extract}, there exists two $\qpt$ algorithms $\bdv'$ and $\cdv'$ that win the simultaneous predicting game defined above.
We construct a triple of $\qpt$ algorithms $(\adv'', \bdv'', \cdv'')$ for the $\kappa$-parallel computational version of monogamy-of-entanglement game (\cref{def:new-moe-parallel}).
\begin{itemize}
  \item $\adv''$, on input the coset states $\{\ket{A_i, s_i, s'_i}\}_{i \in \intval{1, \kappa}}$ and the obfuscated membership programs $\{\mbr_{A_i + s_i}, \allowbreak\mbr_{A^\perp_i + s'_i}\}_{i \in \intval{1, \kappa}}$:
  \begin{itemize}
    \item runs $\adv$ on these coset states and programs to get $\sigma_{12}$ and $(m^1_0, m^1_1), (m^2_0, m^2_1)$;
    \item sample $b_1, b_2 \sample \bin$;
    \item then prepares the bipartite quantum state $\sigma'_{12}$ with $\sigma_1 \otimes \ketbra{b_1}{b_1}$ as first register and $\sigma_2 \otimes \ketbra{b_2}{b_2}$ as second register;
    \item and finally sends $\sigma'_1$ to $\bdv''$ and $\sigma'_2$ to $\cdv''$.
  \end{itemize}
  \item $\bdv''$, on input $\sigma'_1$, the subspace descriptions $\{A_i\}_{i \in \intval{1, \kappa}}$ and the random basis $r$:
  \begin{itemize}
    \item construct a description of $\can_r$ (note that such a description can be computed efficiently given $\{A_i\}_{i \in \intval{1, \kappa}}$ and $r$);
    \item runs $\bdv'$ on $(\sigma_1, \can_r)$ to get the outcome $y'_1$;
    \item and finally returns $y'_1$.
  \end{itemize}
  \item $\cdv''$ is defined similarly as $\bdv''$ by replacing the ``$1$'' indices by ``$2$'' indices.
\end{itemize}
From \cref{lem:cc-extract}, we know that, with non-negligible probability, both $y'_1$ and $y'_2$ are the lock values of the compute-and-compare programs.
Then $\adv'', \bdv'', \cdv''$ win the game with non-negligible probability, contradicting \cref{th:new-moe-parallel} and concluding the proof.

\subsection{Proof of \cref{th:sd-ap-id}}
The proof of \cref{th:sd-ap-id} is almost the same as the one of \cref{th:sd-ap-ir}: we proceed with the same sequence of hybrids and use \cref{conj:identicalB} instead of \cref{conj:uniformB} to finish the proof.

\subsection{Copy-Protection of Pseudorandom Functions}
\label{sec:cp-prf}

In this subsection, we formally define copy-protection of pseudorandom function \cite{C:CLLZ21} and its correctness and anti-piracy notions.

\begin{definition}[Pseudorandom Function Copy-Protection Scheme]
  \label{def:cp-prf}
    A pseudorandom function copy-protection scheme for the pseudorandom function $\prf: \kspace \times \xspace \rightarrow \yspace$ (where $\yspace \subseteq \bin^m$) associated with the key generation procedure $\keygen$ is a tuple of algorithms $\scheme{\keygen, \cpprotect, \eval}$ with the following properties:
    \begin{itemize}
        \item $\key \gets \keygen(\secparam)$.
        This is the key generation procedure of the underlying pseudorandom function: on input a security parameter, the $\keygen$ algorithm outputs a key $\key$.
        \item $\qkey_\key \leftarrow \cpprotect(\secparam, \key)$.
          On input a pseudorandom function key $\key \in \kspace$, the quantum protection algorithm outputs a quantum state $\qkey_\key$.
        \item $y \gets \eval(\secparam, \rho, x)$.
          On input a quantum state $\rho$ and an input $x \in \xspace$, the quantum evaluation algorithm outputs $y \in \yspace$.
    \end{itemize}
\end{definition}

\paragraph{Correctness.}
A pseudorandom function copy-protection scheme has \emph{correctness} if the quantum protection of
any key $\key$ computes $\prf(\key, \cdot)$ on every $x$ with overwhelming probability.
$$
    \forall \key \in \kspace,\ \forall x \in \xspace,\ \pr\left[\eval(\secparam, \qkey_\key, x) = \prf(\key, x)\ :\ \qkey_\key \gets \cpprotect(\secparam, \key)\right] = 1 - \negl
$$

\paragraph{Anti-piracy security.}
We now define anti-piracy security of a pseudorandom function copy-protection scheme similarly as the anti-piracy of single-decryptor.

Anti-piracy security is defined through the following piracy game, in which the adversary is provided a quantum key and a pseudorandom function image, and must ``split'' the quantum key such that both shares can be used to distinguish between the input of this image or another ``fake'' input sampled uniformly at random.
More precisely, the game is played by a triple of adversaries (Alice, Bob and Charlie): Alice splits the quantum state and, in order to test both shares, they are sent to Bob and Charlie as well as the challenge (image's input of fake input) who are asked to guess which type of input they received.
We define two different variants for this security: security with respect to the \emph{product distribution} and security with respect to the \emph{identical distribution}.
When considering security with respect to the product distribution, for each share, a challenge is sampled independently to be either the image's input, or a freshly sampled fake input (both with probability $1/2$).
When considering security with respect to the identical distribution on the other hand, the challenge is still either image's input or a fake input, but it is the same for both Bob and Charlie.

\begin{definition}[Piracy Game for Pseudorandom Function Copy-Protection]
  \label{def:cp-prf-piracy-game}
  We define below a \emph{piracy game for pseudorandom function copy-protection}, parametrized by a pseudorandom function copy-protection scheme $\scheme{\keygen, \cpprotect, \eval}$ and a security parameter $\secpar$.
  As the two variants of the game (with respect to the product distribution or to the identical distribution) differ only in the challenge phase, we define a different challenge phase for each variant.
  This game is between a challenger and a triple of adversaries $(\adv, \bdv, \cdv)$.

  \begin{itemize}
    \item \textbf{Setup phase:}
    \begin{itemize}
      \item The challenger samples $\key \in \keygen(\secparam)$ and computes $\qkey_\key \gets \cpprotect(\secparam, \key)$.
      \item The challenger samples $x \sample \xspace$ and computes $y := \prf(\key, x)$.
      \item The challenger sends $\qkey_\key$ and $y$ to $\adv$.
    \end{itemize}
    \item \textbf{Splitting phase:}
      $\adv$ prepares a bipartite quantum state $\sigma_{12}$, then sends $\sigma_1$ to $\bdv$ and $\sigma_2$ to $\cdv$.
    \item \textbf{Challenge phase (product distribution):}
    \begin{itemize}
      \item The challenger samples two bits $b_1, b_2 \sample \bin$, and two inputs $x_1, x_2 \sample \xspace$.
      \item If $b_1 = 0$, the challenger sends $x$ to $\bdv$; otherwise, the challenger sends $x_1$.
      \item Similarly, if $b_2 = 0$, the challenger sends $x$ to $\cdv$; otherwise, the challenger sends $x_2$.
    \end{itemize}
    \item \textbf{Challenge phase (identical distribution):}
    \begin{itemize}
      \item The challenger samples a bit $b \sample \bin$, and an input $x_0 \sample \xspace$.
      \item If $b = 0$, the challenger sends $x$ to both $\bdv$ and $\cdv$; otherwise, the challenger send them $x_0$.
    \end{itemize}
  \end{itemize}
  $\adv$, $\bdv$, and $\cdv$ win the game with respect to the product distribution if $\bdv$ returns $b_1$ and $\cdv$ returns $b_2$; and with respect to the identical distribution if both $\bdv$ and $\cdv$ return $b$.

  We denote the random variable that indicates whether a triple of adversaries
  \((\adv, \bdv, \cdv)\) win the game or not as \(\funcfont{CP-PRF-AP}^{\scheme{\keygen, \cpprotect, \eval}}_{\text{PD}}(\secparam, \adv, \bdv, \cdv)\) or \(\funcfont{CP-PRF-AP}^{\scheme{\keygen, \cpprotect, \eval}}_{\text{ID}}(\secparam, \adv, \bdv, \cdv)\) depending on if we consider the security with respect to the product distribution (PD) or to the identical distribution (ID).
\end{definition}

\begin{definition}[Indistinguishable Anti-Piracy Security]
  \label{def:cp-prf-ap}
  A pseudorandom function copy-protection scheme $\scheme{\keygen, \cpprotect, \eval}$ has \emph{indistinguishable anti-piracy security} with respect to the product distribution if no $\qpt$ adversary can win the piracy game above (with identical distribution challenge phase) with probability significantly greater than $1/2$.
  More precisely, for any triple of \(\qpt\) adversaries \((\adv, \bdv, \cdv)\):
  \[
    \prob{\funcfont{CP-PRF-AP}^{\scheme{\keygen, \cpprotect, \eval}}_{\text{PD}}(\secparam, \adv, \bdv, \cdv) = 1} \leq 1/2 + \negl.
  \]

  Furthermore, we say that such a scheme has \emph{indistinguishable anti-piracy security} with respect to the product distribution if no $\qpt$ adversary can win the piracy game above (with identical distribution challenge phase) with probability significantly greater than $1/2$.
  More precisely, for any triple of \(\qpt\) adversaries \((\adv, \bdv, \cdv)\):
  \[
    \prob{\funcfont{CP-PRF-AP}^{\scheme{\keygen, \cpprotect, \eval}}_{\text{ID}}(\secparam, \adv, \bdv, \cdv) = 1} \leq 1/2 + \negl.
  \]
\end{definition}

\begin{theorem}
  \label{th:cp-prf-ap}
  Assuming the existence of post-quantum indistinguishability obfuscation,
  one-way functions, compute-and-compare obfuscation for the class of
  unpredictable distributions, and \cref{conj:uniformB} (resp.\@ \cref{conj:identicalB}), there exists a pseudorandom function copy-protection scheme with indistinguishable anti-piracy security with respect to the product distribution (resp.\@ with respect to the identical distribution).
\end{theorem}

We present the construction that achieves this security in the two variants and the corresponding proof in \cref{app:cp-prf}.
\section{Copy-Protection of Point Functions in the Plain Model}
\label{sec:copy-prot}
In this section, we present the definition of copy-protection of point
functions \cite{aaronson2009quantum}.
Then we present a construction of this primitive from \cite{chevalier2023semi}.
This construction was proven secure for a non-colliding anti-piracy game's
challenge distribution.
We prove that the same construction is actually secure for the product challenge
distribution as well as the identical challenge distribution.\footnotemark
Through all this section, $\secpar$ denotes a security parameter and $n = \poly$.
\footnotetext{We actually present a more general version of the construction of \cite{chevalier2023semi}.}

\subsection{Definitions}
We consider copy-protection of point functions for a family of point functions $\{\pf_y\}_{y \in \bin^n}$, and denote $\pf_y$ the point function with point $y$, that is the function such that
$$
  \pf_y(x) = \left\{\begin{array}{ll}
    1 & \text{ if } x = y\\
    0 & \text{ otherwise }
  \end{array}\right.
$$

\begin{definition}[Point Functions Copy-Protection Scheme]
  \label{def:cp-general}
  A copy-protection scheme of a family of point functions $\{\pf\}_{y \in \bin^n}$ is a tuple of algorithms $\scheme{\cpprotect, \eval}$ with the following properties:
  \begin{itemize}
    \item $\qkey_y \leftarrow \cpprotect(\secparam, y)$.
      On input a point $y \in \bin^n$, the quantum protection algorithm outputs a quantum state $\qkey_y$.
    \item $b \gets \eval(\secparam, \rho, x)$.
      On input a quantum state $\rho$ and an input $x \in \bin^n$, the quantum evaluation algorithm outputs a bit $b \in \bin$.
  \end{itemize}
\end{definition}

\paragraph{Correctness.}
A point functions copy-protection scheme has \emph{correctness} if the quantum protection of
any point function $\pf_y$ computes $\pf_y$ on every $x$ with overwhelming probability.
$$
    \forall y \in \bin^n,\ \forall x \in \bin^n,\ \pr\left[\eval(\secparam, \qkey_y, x) = \pf_y(x)\ :\ \qkey_y \gets \cpprotect(\secparam, y)\right] = 1 - \negl
$$

\paragraph{Anti-piracy security.}
We now define anti-piracy security of a point functions copy-protection scheme.
This notion is defined through a piracy game, in which the adversary is given a quantum copy-protection of a point function $\prf_y$ and must split it such that both shares can be used to evaluate the function correctly.
More precisely, the game is played by a triple of adversaries (Alice, Bob and Charlie): Alice splits the quantum state and, in order to test both shares, they are sent to Bob and Charlie as well as the challenge (a point) who are asked to return the evaluation of the function on this point.
We consider two variants of this security notion, namely \emph{anti-piracy security with respect to the product distribution} and \emph{anti-piracy security with respect to the identical distribution}.
In the first variant, for each share, a challenge is sampled independently to be either the point $y$ or another freshly sampled random point; in the second variant, the challenges are either both $y$, or both $x$.

\begin{definition}[Piracy Game for Copy-Protection of Point Functions]
  \label{def:cp-piracy-game}
  We define below a \emph{piracy game for copy-protection of point functions}, parametrized by a copy-protection scheme $\schemefont{CP} = \scheme{\cpprotect, \eval}$ and a security parameter $\secpar$.
  This game is between a challenger and a triple of adversaries $(\adv, \bdv, \cdv)$.
  As the two variants of the game differ only in the challenge phase, we describe below a different challenge phase for each variant.
  \begin{itemize}
    \item \textbf{Setup phase:}
    \begin{itemize}
      \item The challenger samples $y \in \bin^n$ and computes $\rho_y \gets \cpprotect(\secparam, y)$.
      \item The challenger then sends $\rho_y$ to $\adv$.
    \end{itemize}
    \item \textbf{Splitting phase:}
      $\adv$ prepares a bipartite quantum state $\sigma_{12}$, then sends $\sigma_1$ to $\bdv$ and $\sigma_2$ to $\cdv$.
    \item \textbf{Challenge phase (product distribution):}
    \begin{itemize}
      \item The challenger samples $b_1, b_2 \sample \bin$, and $x_1, x_2 \sample \bin^n$.
      \item If $b_1 = 0$, the challenger sends $y$ to $\bdv$; otherwise the challenger sends $x_1$.
      \item Similarly, if $b_2 = 0$, the challenger sends $y$ to $\cdv$; otherwise the challenger sends $x_2$.
    \end{itemize}
    \item \textbf{Challenge phase (identical distribution):}
    \begin{itemize}
      \item The challenger samples $b \sample \bin$, and $x \sample \bin^n$.
      \item If $b = 0$, the challenger sends $y$ to both $\bdv$ and $\cdv$; otherwise, the challenger send them $x$.
    \end{itemize}
  \end{itemize}
  $\adv$, $\bdv$, and $\cdv$ win the game with respect to the product distribution if $\bdv$ returns $b'_1 = b_1$ and $\cdv$ returns $b'_2 = b_2$; and with respect to the identical distribution if both $\bdv$ and $\cdv$ return $b$.

  We denote the random variable that indicates whether a triple of adversaries
  \((\adv, \bdv, \cdv)\) win the game or not as $\funcfont{CP-AP}^{\scheme{\cpprotect, \eval}}_{\text{PD}}(\secparam, \adv, \bdv, \cdv)$ or as $\funcfont{CP-AP}^{\scheme{\cpprotect, \eval}}_{\text{ID}}(\secparam, \adv, \bdv, \cdv)$ depending on which variant of the game we consider ($\text{PD}$ and $\text{ID}$ respectively denote the product and identical distributions).
\end{definition}

\begin{definition}[Anti-Piracy Security]
  A point functions copy-protection scheme $\scheme{\cpprotect, \eval}$ has \emph{anti-piracy security} with respect to the product distribution if no triple of $\qpt$ adversaries can win the piracy game above (with the product distribution challenge phase) with probability significantly greater than $1/2$.
  More precisely, for any triple of \(\qpt\) adversaries \((\adv, \bdv, \cdv)\):
  \[
    \prob{\funcfont{CP-AP}^{\scheme{\cpprotect, \eval}}_{\text{PD}}(\secparam, \adv, \bdv, \cdv) = 1} \leq 1/2 + \negl.
  \]

  Similarly, we say that a point functions copy-protection scheme $\scheme{\cpprotect, \eval}$ has \emph{anti-piracy security} with respect to the identical distribution if no $\qpt$ adversary can win the piracy game above (with the identical distribution challenge phase) with probability significantly greater than $1/2$.
  More precisely, for any triple of \(\qpt\) adversaries \((\adv, \bdv, \cdv)\):
  \[
    \prob{\funcfont{CP-AP}^{\scheme{\cpprotect, \eval}}_{\text{ID}}(\secparam, \adv, \bdv, \cdv) = 1} \leq 1/2 + \negl.
  \]
\end{definition}

\subsection{Construction}

In this subsection, we present a construction for copy-protection of point functions.
This construction uses a pseudorandom functions copy-protection scheme $\prf.\scheme{\keygen, \cpprotect, \eval}$.

\begin{construction}{Copy-Protection of Point Functions}{cp-cc-from-cp-prf}
  \begin{itemize}
    \item $\cpprotect(\secparam, y)$:
    \begin{itemize}
      \item Sample $\key \gets \prf.\keygen(\secparam)$.
      \item Compute $\rho_\key \gets \prf.\cpprotect(\key)$.
      \item Compute $z := \prf(\key, y)$.
      \item Return $(\rho_\key, z)$.
    \end{itemize}
    \item $\eval(\secparam, (\rho, z), x)$:
    \begin{itemize}
      \item Compute $z' \gets \prf.\eval(\rho, x)$.
      \item If $z' = z$: return 1.
      \item Otherwise: return 0.
    \end{itemize}
  \end{itemize}
\end{construction}

\begin{theorem}
  \label{th:cp-ap-pd}
  Assuming the underlying pseudorandom functions copy-protection scheme has anti-piracy security with respect to the product distribution, \cref{constr:cp-cc-from-cp-prf} has correctness and anti-piracy security with respect to the product distribution.
\end{theorem}

\begin{theorem}
  \label{th:cp-ap-id}
  Assuming the underlying pseudorandom functions copy-protection scheme has anti-piracy security with respect to the identical distribution, \cref{constr:cp-cc-from-cp-prf} has correctness and anti-piracy security with respect to the identical distribution.
\end{theorem}

\begin{proof}[Proof of \cref{th:cp-ap-pd,th:cp-ap-id}]
  \textit{(Correctness)} for any $y$, running the evaluation algorithm on the point $y$ yields $1$ with probability close to $1$ from the correctness of the underlying copy-protection of pseudorandom functions.
  Running the evaluation algorithm on a point $x \neq y$ yields $1$ only if $\prf.\eval(\rho_\key, x) = \prf(\key, y)$, which happens with negligible probability over $\key$ from the security of the underlying pseudorandom function.

  \textit{(Anti-piracy security)} the anti-piracy security with respect to the product distribution (resp.\@ identical distribution) comes directly from the anti-piracy security with respect to the product distribution (resp.\@ identical distribution) of the underlying copy-protection of pseudorandom function scheme.
  In both cases, the reduction is simply the identity.
\end{proof}

\begin{corollary}
  Assuming the existence of post-quantum indistinguishability obfuscation,
  one-way functions, compute-and-compare obfuscation for the class of
  unpredictable distributions, and \cref{conj:uniformB} (resp.\@ \cref{conj:identicalB}), there exists a point functions copy-protection scheme with correctness and anti-piracy security with respect to the product distribution (resp.\@ the identical distribution).
\end{corollary}

\begin{proof}
  This result follows directly from \cref{th:cp-prf-ap}.
\end{proof}

\section{Unclonable Encryption in the Plain Model}
\label{sec:unc-ind}
In this section, we introduce the notion of unclonable encryption \cite{broadbent2019uncloneable} and present a construction in the plain model.
Our construction uses a pseudorandom function copy-protection scheme with anti-piracy security with respect to the product distribution (\cref{def:cp-prf-ap}) as a black box.
Our construction is a symmetric one-time unclonable encryption scheme, which implies the existence of a reusable public key encryption scheme using the transformation of \cite{TCC:AnaKal21}.
Through all this section, we refer to symmetric unclonable encryption simply as unclonable encryption, and use public key unclonable encryption to denote the public key version.

\subsection{Definitions}
\label{sec:unc-ind_def}
In this section, we define unclonable encryption as well as its correctness and indistinguishable anti-piracy security.

\begin{definition}[One-Time Unclonable Encryption Scheme]
  \label{def:unc-enc}
  A one-time unclonable encryption scheme with message space $\mspace$ is a tuple of algorithms $\scheme{\keygen, \enc, \dec}$ with the following properties:
  \begin{itemize}
    \item $\key \leftarrow \keygen(\secparam)$.
    On input a security parameter, the key generation algorithm outputs a key $\key$.
    \item $\rho \leftarrow \enc(\key, m)$.
    On input a key $\key$ and a message $m \in \mspace$, the encryption algorithm outputs quantum ciphertext $\rho$.
    \item $m \leftarrow \dec(\key, \rho)$.
    On input a key $\key$ and a quantum ciphertext $\rho$, the decryption algorithm outputs a message $m$.
  \end{itemize}
\end{definition}

\paragraph{Correctness.}
An unclonable encryption scheme has \emph{correctness} if decrypting a quantum encryption of any message $m$ yields $m$ with overwhelming probability.
More precisely:
$$
    \forall m \in \mspace,\ \pr\left[\dec(\key, \rho) = m\ :\ \begin{array}{l}
      \key \gets \keygen(\secparam)\\
      \rho \gets \enc(\key, m)
    \end{array}\right] = 1 - \negl
$$

\paragraph{Indistinguishable anti-piracy security.}
We now define indistinguishable anti-piracy security of a one-time unclonable encryption scheme.
This notion is defined through a game in which an adversary is given a quantum encryption of either $m_0$ or $m_1$ - two messages chosen by the adversary at the beginning of the game - and is asked to split it such that both shares can be used to guess which message has been encrypted.
Note that although our definition holds for \emph{one-time} unclonable encryption schemes, we can similarly define this notion for \emph{reusable} unclonable encryption schemes by giving the adversary access to an encryption oracle, before asking them to choose the pair of messages.

\begin{definition}[Piracy Game for a One-Time Unclonable Encryption Scheme]
  We define below a \emph{piracy game for one-time unclonable encryption}, parametrized by a one-time unclonable encryption scheme $\scheme{\keygen, \enc, \dec}$, and a security parameter $\secpar$.
  This game is between a challenger and a triple of adversaries $(\adv, \bdv, \cdv)$.
  \label{def:unc-enc-piracy-game}
  \begin{itemize}
    \item \textbf{Setup phase:}
    \begin{itemize}
      \item $\adv$ sends a message pair $(m_0, m_1) \in \mspace^2$ to the challenger.
      \item The challenger samples $\key \in \keygen(\secparam)$ and $b \sample \bin$, and computes $\rho \gets \enc(\key, m_b)$.
      \item The challenger sends $\rho$ to $\adv$.
    \end{itemize}
    \item \textbf{Splitting phase:}
      $\adv$ prepares a bipartite quantum state $\sigma_{12}$, then sends $\sigma_1$ to $\bdv$, and $\sigma_2$ to $\cdv$.
    \item \textbf{Challenge phase:}
      The challenger sends $\key$ to both $\bdv$ and $\cdv$.
  \end{itemize}
  $\adv$, $\bdv$, and $\cdv$ win the game if $\bdv$ returns $b'_1 = b$ and $\cdv$ returns $b'_2 = b$.

  We denote the random variable that indicates whether a triple of adversaries
  \((\adv, \bdv, \cdv)\) win the game or not as $\funcfont{UncEnc-AP}^{\scheme{\keygen, \enc, \dec}}(\secparam, \adv, \bdv, \cdv)$.
\end{definition}

\begin{definition}[Indistinguishable Anti-Piracy Security of an Unclonable Encryption Scheme]
  A one-time unclonable encryption scheme $\scheme{\keygen, \enc, \dec}$ has \emph{indistinguishable anti-piracy security} if no triple of $\qpt$ adversaries can win the piracy game above with probability significantly greater than $1/2$.

  More precisely, for any triple of $\qpt$ adversaries $(\adv, \bdv, \cdv)$
  $$
  \prob{\funcfont{UncEnc-AP}^{\scheme{\keygen, \enc, \dec}}(\secparam, \adv, \bdv, \cdv) = 1} \leq 1/2 + \negl.
  $$
\end{definition}

\subsection{Construction}
\label{sec:unc_ind_constr}

In this subsection, we present a construction of a one-time unclonable encryption scheme for single-bit messages.
Through all the subsection, $\secpar$ denotes a security parameter and $n(\cdot), m(\cdot)$ are polynomials; whenever it is clear from the context, we note $n$ and $m$ instead of $n(\secpar)$ and $m(\secpar)$.
Let $\prf.\scheme{\keygen, \cpprotect, \eval}$ be a pseudorandom function copy-protection scheme with input space $\bin^n$ and output space $\bin^m$.
In addition, we ask the copy-protected pseudorandom function to be extracting with error $2^{-\lambda - 1}$ for min-entropy $n$.
Note that the copy-protected pseudorandom function presented in \cref{app:cp-prf} has this property.

\begin{construction}{Unclonable Encryption}{unc-enc}
  \begin{itemize}[label=\(\),rightmargin=\leftmargin]
    \item \(\keygen(\secparam)\):
    \begin{itemize}
      \item Sample a key $\key_S \sample \bin^n$.
      \item Return $\key_S$.
    \end{itemize}
    \item $\enc(\key_S, b)$:
    \begin{itemize}
      \item Sample $\key_P \gets \prf.\keygen(\secparam)$ and compute $\rho_{\key_P} \gets \prf.\cpprotect(\key_P)$.
      \item Sample $r \sample \bin^n$; let $c_0 \gets \prf(\key_P, k_{S} \oplus r)$ and $c_1 \sample \bin^m$.
      \item Return $(r, c_b, \rho_{\key_P})$.
    \end{itemize}
    \item $\dec(\key_S, (r, c, \rho_{\key_P}))$:
    \begin{itemize}
      \item Compute $c^* \gets \prf(\key_P, \key_S \oplus r)$.
      \item Return $0$ if $c^* = c$ and $1$ otherwise.
    \end{itemize}
  \end{itemize}
\end{construction}

\begin{theorem}
  Assume $\prf.\scheme{\keygen, \cpprotect, \eval}$ has indistinguishable anti-piracy security with respect to the identical distribution (\cref{def:cp-prf-ap}).
  Then \cref{constr:unc-enc} has correctness and indistinguishable anti-piracy security.
\end{theorem}

\paragraph{Proof of correctness.}
The correctness comes directly from the correctness and security of the underlying $\prf$ copy-protection scheme.
More precisely, $\dec(\key_S, \enc(\key_S, 0)) = 1$ means that $\prf.\eval(\rho_{\key_P}, \key_S \oplus r) \neq \prf(\key_P, \key_S \oplus r)$ which happens with negligible probability from the correctness of the $\prf$ copy-protection scheme.
And $\dec(\key_S, \enc(\key_S, 1)) = 0$ means that $\prf.\eval(\rho_{\key_P}, \key_S \oplus r) = y$ for a uniformly random $y$ happens with non-negligible probability, which contradicts $\prf$ security.

\paragraph{Proof of indistinguishable anti-piracy security.}
\label{sec:unc_ind_proof}

We proceed in the proof through a sequence of hybrids.
For any pair of hybrids \((\pcgamename_i, \pcgamename_j)\), we say that \(\pcgamename_i\) is \(\emph{negligibly close to}\) \(\pcgamename_j\) if for triple of \(\qpt\) adversaries \((\adv, \bdv, \cdv)\), the probability that \((\adv, \bdv, \cdv)\) wins \(\pcgamename_i\) is negligibly close to the probability that they win \(\pcgamename_j\).

\begin{gamedescription}[nr=-1,arg=]
  \describegame The first hybrid is the piracy game for our construction.

  \begin{itemize}
    \item \textbf{Setup phase:}
    \begin{itemize}
      \item The challenger samples $\key_S \sample \bin^n$.
      \item The challenger samples $\key_P \gets \prf.\keygen(\secparam)$ and computes $\rho_{\key_P} \gets \prf.\cpprotect(\key_P)$.
      \item The challenger samples $r \sample \bin^n$, then sets $c_0 \gets \prf(\key_P, k_{S} \oplus r)$ and samples $c_1 \sample \bin^m$.
      \item The challenger samples $b \sample \bin$ and sends $(r, c_b, \rho_{\key_P})$ to $\adv$.
    \end{itemize}
    \item \textbf{Splitting phase:}
      $\adv$ prepares a bipartite quantum state $\sigma_{12}$, then sends $\sigma_1$ to $\bdv$, and $\sigma_2$ to $\cdv$.
    \item \textbf{Challenge phase:}
    The challenger sends $\key_S$ to both $\bdv$ and $\cdv$.
  \end{itemize}
  $\adv$, $\bdv$, and $\cdv$ win the game if $\bdv$ returns $b'_1 = b$ and $\cdv$ returns $b'_2 = b$.

  \describegame In the second hybrid, we replace $c_1$ by the pseudorandom function evaluation of a random input.
  More formally, in the setup phase, we replace $c_1 \sample \bin^m$ by $c_1 := \prf(\key_P, x)$ where $x \sample \bin^m$.
  As $x$ is sampled uniformly at random, from the extracting property of the underlying pseudorandom function, $\pcgamename_0$ is negligibly close to $\pcgamename_1$.

  \describegame In this third hybrid, we replace the random $x$ by $\key'_S \oplus r$ where $\key'_S$ is sampled uniformly at random from $\bin^n$.
  As $\key'_S \oplus r$ is still uniformly random, this does not change the overall distribution of the game.
  Thus, $\pcgamename_2$ is negligibly close to $\pcgamename_1$ (more precisely, it is exactly the same game).

  \describegame For the third hybrid, instead of sending either $\prf(\key_P, \key_S \oplus r)$ or $\prf(\key_P, \key'_S \oplus r)$ - depending on $b$ - to $\adv$ in the setup phase, and sending $\key_S$ to $\bdv$ and $\cdv$, we send only $\prf(\key_P, \key_S \oplus r)$ to $\adv$ in the setup phase, and send either $\key_S$ or $\key'_S$ to $\bdv$ and $\cdv$- still depending on $b$.
  Note that this is actually only relabelling, hence the distribution of the game is not changed either.
  Thus, the $\pcgamename_3$ has exactly the same distribution as $\pcgamename_2$.
  We describe $\pcgamename_3$ more precisely below:

  \begin{itemize}
    \item \textbf{Setup phase:}
    \begin{itemize}
      \item The challenger samples $\key_S, \key'_S \sample \bin^n$.
      \item The challenger samples $\key_P \gets \prf.\keygen(\secparam)$ and computes $\rho_{\key_P} \gets \prf.\cpprotect(\key_P)$.
      \item The challenger samples $r \sample \bin^n$, and sends $(r, \prf(\key_P, k_{S} \oplus r), \rho_{\key_P})$ to $\adv$.
    \end{itemize}
    \item \textbf{Splitting phase:}
      $\adv$ prepares a bipartite quantum state $\sigma_{12}$, then sends $\sigma_1$ to $\bdv$, and $\sigma_2$ to $\cdv$.
    \item \textbf{Challenge phase:}
    The challenger sends $\key_S$ to both $\bdv$ and $\cdv$ if $b = 0$, otherwise the challenger sends $\key'_S$.
  \end{itemize}
  $\adv$, $\bdv$, and $\cdv$ win the game if $\bdv$ returns $b'_1 = b$ and $\cdv$ returns $b'_2 = b$.
\end{gamedescription}

We now reduce the game $\pcgamename_3$ from the piracy game of the underlying pseudorandom function copy-protection scheme with respect to the product distribution.
Assume that there exists a triple of \(\qpt\) adversaries \((\adv, \bdv, \cdv)\) who wins \(\game_{3}\) with advantage \(\delta\).
We construct a triple of \(\qpt\) adversaries \((\adv', \bdv', \cdv')\) who wins the piracy game of the underlying pseudorandom function copy-protection scheme with respect to the product distribution with the same advantage \(\delta\).
\((\adv', \bdv', \cdv')\) behave in the following way.
\begin{itemize}
  \item $\adv'$, on input a quantum protected pseudorandom function key $\rho_\key$ and a pseudorandom function image $y := \prf(\key, x)$:
  \begin{itemize}
    \item samples $r \sample \bin^n$ (note that, by defining $\key_S := x \oplus r$, we can write $x$ as $\key_S \oplus r$);
    \item runs $\adv$ on $(r, y, \rho_\key)$ to get $\sigma_{12}$;
    \item prepares the bipartite quantum state $\sigma'_{12}$ where the first register is $\sigma_1 \otimes \ketbra{r}{r}$ and the second one is $\sigma_2 \otimes \ketbra{r}{r}$;
    \item sends $\sigma'_1$ to $\bdv$ and $\sigma'_2$ to $\cdv$.
  \end{itemize}
  \item $\bdv'$, on input $\sigma'_1$ and $x$:
  \begin{itemize}
    \item computes $\key := r \oplus x$;
    \item runs $\bdv$ on $(\sigma_1, \key)$;
    \item returns the outcome.
  \end{itemize}
  \item $\cdv'$ is defined similarly by replacing the ``$1$'' indices by ``$2$'' indices.
\end{itemize}

The inputs given to $\bdv$ and $\cdv$ follow the same distribution as their inputs in $\pcgamename_3$.
Thus, \((\adv', \bdv', \cdv')\) win the game with the same advantage as \((\adv, \bdv, \cdv)\), which concludes the proof.

\begin{corollary}
  Assuming the existence of post-quantum indistinguishability obfuscation,
  one-way functions, compute-and-compare obfuscation for the class of
  unpredictable distributions, and \cref{conj:identicalB}, there exists a one-time unclonable encryption scheme with correctness and indistinguishable anti-piracy security for $1$-bit long messages.
\end{corollary}

\subsection{Extension to Multi-Bits Messages}
\label{sec:unc-enc-multibits}

We describe below a way to extend our scheme to any message space of the form $\bin^\ell$ where $\ell(\cdot)$ is a polynomial in $\secpar$.
Our construction encrypts the message bit by bit, but not in an independent way.
Indeed, we use the same pseudorandom function key for encrypting all the bits (and hence the same copy-protected pseudorandom function key); and show that this is enough to achieve indistinguishable anti-piracy security.

\begin{construction}{Unclonable Encryption with Message Space $\bin^\ell$}{unc-enc-multibits}
  \begin{itemize}[label=\(\),rightmargin=\leftmargin]
    \item \(\keygen(\secparam)\):
    \begin{itemize}
      \item For $i \in \intval{1, \ell}$: sample a key $\key_{S, i} \sample \bin^n$.
      \item Return $\key_S := \left(\key_{S, i}\right)_{i \in \intval{1, \ell}}$.
    \end{itemize}
    \item $\enc(\key_S, m)$:
    \begin{itemize}
      \item Sample $\key_P \gets \prf.\keygen(\secparam)$ and compute $\rho_{\key_P} \gets \prf.\cpprotect(\key_P)$.
      \item For $i \in \intval{1, \ell}$: sample $r_i \sample \bin^n$ and compute $y_i := \key_{S, i} \oplus r_i$.
      \item Let $c_{0, i} \gets \prf(\key_P, y_i)$ and $c_{1, i} \sample \bin^m$.
      \item Let $r := \left(r_i\right)_{i \in \intval{1, \ell}}$ and $c_m := \left(c_{m_i, i}\right)_{i \in \intval{1, \ell}}$.
      \item Return $(r, c_m, \rho_{\key_P})$.
    \end{itemize}
    \item $\dec(\key_S, (r, c, \rho_{\key_P}))$:
    \begin{itemize}
      \item For $i \in \intval{1, \ell}$: compute $y_i := \key_{S, i} \oplus r_i$ and $c^*_i \gets \prf(\key_P, y_i)$.
      \item Set $m \in \bin^\ell$ such that $m_i := 0$ if $c^*_i = c_i$ and $m_i := 1$ otherwise.
      \item Return $m$
    \end{itemize}
  \end{itemize}
\end{construction}

\begin{theorem}
  \label{th:unc-enc-ap}
  Assume $\prf.\scheme{\keygen, \cpprotect, \eval}$ has indistinguishable anti-piracy security with respect to the product distribution.
  Then \cref{constr:unc-enc-multibits} has correctness and indistinguishable anti-piracy security.
\end{theorem}

\begin{proof}
  The correctness proof is the same as for the single-bit version: it relies on correctness and security of the underlying pseudorandom function copy-protection scheme.

  We give a short summary of the proof of anti-piracy security, as it uses a usual hybrid argument.
  By doing small hops, we show that if no adversary can distinguish between the encryption of two messages differing on only one index, then no adversary can distinguish between the encryption of two messages differing on only two indices, and so on and so forth until finally showing that no adversary can distinguish between the encryption of two messages differing on all indices.
  It then remains to show that no adversary can distinguish between the encryption of two messages differing only on one index, which follows from the anti-piracy security of the pseudorandom function copy-protection scheme.
\end{proof}

\begin{corollary}
  Assuming the existence of post-quantum indistinguishability obfuscation,
  one-way functions, compute-and-compare obfuscation for the class of
  unpredictable distributions, and \cref{conj:identicalB}, there exists a public-key reusable unclonable encryption scheme with correctness and indistinguishable anti-piracy security for $1$-bit long messages.
\end{corollary}

\begin{proof}
  In \cite[Section 5]{TCC:AnaKal21}, the authors present a way to construct a public-key reusable one-time unclonable encryption scheme from any symmetric one-time unclonable encryption scheme with indistinguishable anti-piracy security, using a (post-quantum) symmetric encryption scheme with pseudorandom ciphertexts and a (post-quantum) single-key public-key functional encryption scheme.
  We refer the interested reader to this paper for a description of the construction.
\end{proof}

\section{Unclonable Unforgeability for Tokenized Signature}
In this section, we present another application to our new monogamy-of-entanglement game.
We define a new security property for tokenized signature schemes and prove that the protocol presented in \cite{C:CLLZ21} achieves this property.
In the following, we first present tokenized signature schemes, with the previous security definitions given in the literature, together with our new security definition; then we describe the \cite{C:CLLZ21} tokenized signature protocol; and finally we prove that this protocol achieves our new security definition.

\subsection{Definitions}

\begin{definition}
  \label{def:tok-sig}
  A tokenized signature scheme is a tuple of algorithms $\scheme{\keygen, \tokengen, \sign, \verify}$ with the following properties:
  \begin{itemize}
    \item $(\sk, \vk) \leftarrow \keygen(\secparam)$.
      On input a security parameter $\secparam$, the key generation algorithm outputs a secret key $\sk$ and a public verification key $\vk$.
    \item $\rho \gets \tokengen(\sk)$.
      On input a secret key $\sk$, the quantum token generation algorithm outputs a quantum signing token $\rho$.
    \item $s \gets \sign(\rho, m)$.
    On input a signing token $\rho$ and a message $m$ to be signed, the signing algorithm outputs a classical signature $s$ of $m$.
    \item $\top / \bot \gets \verify(\vk, m, s)$.
    On input a verification key $\vk$, a message $m$, and a signature $s$, the verification algorithm either accepts or rejects - that is, outputs $\top$ or $\bot$ respectively.
  \end{itemize}
\end{definition}

\paragraph{Correctness.}
A tokenized signature scheme has \emph{correctness} if a signature of any message $m$ produced by a valid token is accepted with overwhelming probability.
$$
  \forall m \in \bin, \pr\left[\verify(\vk, m, s) = \top\ :\ \begin{array}{l}
    s \gets \sign(\rho, m)\\
    \rho \gets \tokengen(\sk)\\
    (\sk, \vk) \gets \keygen(\secparam)
  \end{array}
  \right] = 1 - \negl
$$

\paragraph{(Strong-)unforgeability.}
A tokenized signature has \emph{unforgeability} if no $\qpt$ adversary can produce two different messages, together with a valid signature for each message.
More precisely, for all $\qpt$ adversary $\adv$,
$$
  \pr\left[\begin{array}{c}
    \verify(\vk, m_1, s_1) = \top\\
    \land\\
    \verify(\vk, m_2, s_2) = \top\\
    \land\\
    m_1 \neq m_2
  \end{array}\ :\ 
  \begin{array}{l}
    (m_1, s_1, m_2, s_2) \gets \adv(\vk, \rho)\\
    \rho \gets \tokengen(\sk)\\
    (\sk, \vk) \gets \keygen(\secparam)
  \end{array}
  \right] = \negl
$$
Note that, as we consider public verification, the adversary is also given the verification key.

Similarly, a tokenized signature has \emph{strong-unforgeability} if no $\qpt$ adversary can produce two different (message, signature) pairs.
Note that the messages in the two pairs can be equal, but in this case, the signatures must be different.
More precisely, for all $\qpt$ adversary $\adv$,
$$
  \pr\left[\begin{array}{c}
    \verify(\vk, m_1, s_1) = \top\\
    \land\\
    \verify(\vk, m_2, s_2) = \top\\
    \land\\
    (m_1, s_1) \neq (m_2, s_2)
  \end{array}\ :\ 
  \begin{array}{l}
    (m_1, s_1, m_2, s_2) \gets \adv(\vk, \rho)\\
    \rho \gets \tokengen(\sk)\\
    (\sk, \vk) \gets \keygen(\secparam)
  \end{array}
  \right] = \negl
$$

\subsection{Unclonable Unforgeability}
To motivate our new definition, consider the recent copy-protection of digital signatures primitive \cite{TCC:LLQZ22}.
Similarly, as our quantum tokens, this primitive allows producing quantum signing keys that can be used to sign messages.
However, while the tokens considered in this work are one-time, the keys in a copy-protection of digital signature scheme can be reuse a polynomial number of times.
This makes the (weak and strong) unforgeability definitions above not-applicable to such a scheme, as, if one is given a copy-protected key, they will be able to sign as many messages as they want.
The security that is considered for this scheme is rather defined through a game, in which $\adv$ is given a quantum key; she is asked to split and share it with $\bdv$ and $\cdv$; and the latter each have to produce a valid signature of a random message, sampled by the challenger.

In the following, we consider applying this unclonability definition to tokenized signature schemes.
That is, we wonder what can a triple of adversaries $(\adv, \bdv, \cdv)$ do if $\adv$ receives a token and split it to share it between $\bdv$ and $\cdv$.
The latter are then asked to sign a random challenge message, and must both produce a valid signature.
We then say that a tokenized signature scheme has unclonable unforgeability if no such triple of adversaries can win this game with probability greater than some trivial winning probability \footnotemark.

\paragraph{Challenge distribution.}
In the case where the challenge messages are different - typically if they are sampled from a product distribution - then it is easy to see that the (even weak) unforgeability property implies this unclonable security.
Indeed, if such adversaries win the game, then another adversary $\adv'$, given a signature token, could simply run $\adv$ on this token to get a bipartite state $\sigma_{12}$, then run $\bdv$ on $(\sigma_1, 0)$ and $\cdv$ on $(\sigma_2, 1)$, to obtain a signature of $0$ and a signature of $1$, and therefore break unforgeability security.
However, when we consider the case where the challenges given to $\bdv$ and $\cdv$ are the same message - typically if they are sampled from an identical distribution - then the aforementioned strategy no longer works.
\footnotetext{Similarly, as for copy-protection and single-decryptor, this trivial probability depends on the challenge distribution.
In the following, we consider tokenized signature scheme for single-bit messages, and identical challenge distribution.
Therefore, the trivial probability is $1/2$.}
Based on these observations, we present our new definitions in the following.

\begin{definition}[Piracy Game for Tokenized Signature]
  We define below a \emph{piracy game for tokenized signature}, parametrized by a tokenized signature scheme for single-bit message space $\schemefont{TS} = \scheme{\keygen, \tokengen, \sign, \verify}$ and a security parameter $\secpar$.
  This game is between a challenger and a triple of adversaries $(\adv, \bdv, \cdv)$.

  \begin{itemize}
    \item \textbf{Setup phase:}
    \begin{itemize}
      \item The challenger samples a random pair of keys $(\sk, \vk) \gets \keygen(\secparam)$.
      \item The challenger prepares a signing token $\rho \gets \tokengen(\sk)$.
      \item The challenger sends $(\rho, \vk)$ to $\adv$.
    \end{itemize}
    \item \textbf{Splitting phase:}
      $\adv$ prepares a bipartite quantum state $\sigma_{12}$, then sends $\sigma_1$ to $\bdv$, and $\sigma_2$ to $\cdv$.
    \item \textbf{Challenge phase:}
    \begin{itemize}
      \item The challenger samples a bit $b \sample \bin$.
      \item The challenger sends $b$ to $\bdv$, and $b$ to $\cdv$.
    \end{itemize}
  \end{itemize}
  Let $s_1$ denotes the output of $\bdv$ and $s_2$ denotes the output of $\cdv$.
  $\adv$, $\bdv$, and $\cdv$ win the game if $\verify(\vk, b, s_1) = 1$ and $\verify(\vk, b, s_2) = 1$.

  We denote the random variable that indicates whether a triple of adversaries $\left(\adv, \bdv, \cdv\right)$ wins the game or not as $\funcfont{TS-UU-1}^{\schemefont{TS}}(\secparam, \adv, \bdv, \cdv)$.
\end{definition}

\begin{definition}[Unclonable Unforgeability]
  A tokenized signature scheme $\schemefont{TS}$ for single-bit messages has unclonable unforgeability if no $\qpt$ adversary can win the game above with probability significantly greater than $1/2$.
  More precisely, for any triple of $\qpt$ adversaries $\left(\adv, \bdv, \cdv\right)$:
  $$
    \prob{\funcfont{TS-UU-1}^{\schemefont{TS}}(\secparam, \adv, \bdv, \cdv) = 1} \leq 1/2 + \negl
  $$
\end{definition}

\begin{remark}
  Although this definition is made for tokenized signature schemes with single-bit messages, we informally propose two ways of extending it for general tokenized signature schemes.
  The first way simply consists in sampling a challenge message from the identical distribution over the message space $\mspace$, and then to say that the scheme is secure if no adversaries can win the game with probability greater than $1/\lvert\mspace\rvert$.

  The second way is defined in an indistinguishable fashion: after sending the quantum token to $\adv$, the challenger asks $\adv$ to return them a pair of messages $(m_0, m_1)$.
  The game then proceeds as before, except that the challenge distribution is the identical distribution over $\{m_0, m_1\}$.
  We say that the scheme is secure if no adversaries can win the game with probability greater than $1/2$.
\end{remark}

\subsection{The \cite{C:CLLZ21} Tokenized Signature Scheme}
We give a construction of single-bit tokenized signature scheme from hidden coset states in~\cref{constr:token-sig}.
This construction is identical to the one for weak unforgeability in~\cite{C:CLLZ21}.
Note that \cite{chevalier2023semi} showed that the same construction also achieves strong unforgeability.

\begin{construction}{A Single-Bit Tokenized Signature Scheme}{token-sig}
  \begin{itemize}[label=\(\),rightmargin=\leftmargin]
    \item \(\keygen(1^{\secpar}):\)
    \begin{itemize}[rightmargin=\leftmargin]
        \item Set \(n = \poly\).
        \item Sample uniformly \(A \subseteq \FF_{2}^{n}\) of dimension
          \(\frac{n}{2}\).
        \item Sample \(s, s' \sample \FF_{2}^{n}\).
        \item Output \(\sk \coloneqq (A, s, s')\)
          \(\vk \coloneqq (\iO(P_{A+s}), \iO(P_{A^{\perp}+s'}))\) (where, by \(A\), we mean a
          description of the subspace \(A\), and $P_{A+s}$ and $P_{A^\perp + s'}$ denote the membership programs for $A + s$ and $A^\perp + s'$ respectively.)
      \end{itemize}
    \item \(\tokengen(\sk):\) 
    \begin{itemize}
      \item Parse \(\sk\) as \((A, s, s')\).
      \item Output \(\rho \coloneqq \ket{A_{s, s'}}\).
    \end{itemize}
    \item \(\sign(m, \rho):\)
      \begin{itemize}
        \item Compute \(\Hgate^{\otimes n}\rho\) if \(m = 1\), otherwise
          do nothing to the quantum state.
        \item Measure the state in the computational basis.
          Let \(\sigma\) be the outcome.
        \item Output \((m, \sigma)\).
      \end{itemize}
    \item \(\verify(\vk, m, \sigma):\)
    \begin{itemize}
      \item Parse \(\vk\) as \((C_{0}, C_{1})\) where \(C_{0}\) and \(C_{1}\) are circuits.
      \item Output \(C_{m}(\sigma)\).
    \end{itemize}
  \end{itemize}
\end{construction}

\begin{theorem}
  \label{thm:token-sig}
  Assuming the existence of quantum-secure indistinguishability obfuscation and quantum-secure injective one-way functions, the scheme given in~\cref{constr:token-sig} has unclonable unforgeability.
\end{theorem}

\begin{proof}
  The proof of this theorem follows immediately from~\cref{th:new_moe_coset_comp}.
\end{proof}
 
\bibliographystyle{alpha}
\renewcommand{\doi}[1]{\url{#1}}
\bibliography{../cryptobib/abbrev3,../cryptobib/crypto,ref.bib}

\newcommand{\etalchar}[1]{$^{#1}$}
\begin{thebibliography}{CGLZR23}

\bibitem[Aar09]{aaronson2009quantum}
Scott Aaronson.
\newblock Quantum copy-protection and quantum money.
\newblock In {\em 2009 24th Annual IEEE Conference on Computational
  Complexity}, pages 229--242. IEEE, 2009.

\bibitem[AB23]{EPRINT:AnaBeh23}
Prabhanjan Ananth and Amit Behera.
\newblock A modular approach to unclonable cryptography, 2023.
\newblock \url{https://arxiv.org/abs/2311.11890}.

\bibitem[AK21]{TCC:AnaKal21}
Prabhanjan Ananth and Fatih Kaleoglu.
\newblock Unclonable encryption, revisited.
\newblock In Kobbi Nissim and Brent Waters, editors, {\em TCC~2021, Part~I},
  volume 13042 of {\em {LNCS}}, pages 299--329. Springer, Heidelberg, November
  2021.

\bibitem[AKL{\etalchar{+}}22]{C:AKLLZ22}
Prabhanjan Ananth, Fatih Kaleoglu, Xingjian Li, Qipeng Liu, and Mark Zhandry.
\newblock On the feasibility of unclonable encryption, and more.
\newblock In Yevgeniy Dodis and Thomas Shrimpton, editors, {\em CRYPTO~2022,
  Part~II}, volume 13508 of {\em {LNCS}}, pages 212--241. Springer, Heidelberg,
  August 2022.

\bibitem[AKL23]{C:AnaKalLiu23}
Prabhanjan Ananth, Fatih Kaleoglu, and Qipeng Liu.
\newblock Cloning games: {A} general framework for unclonable primitives.
\newblock In Helena Handschuh and Anna Lysyanskaya, editors, {\em CRYPTO~2023,
  Part~V}, volume 14085 of {\em {LNCS}}, pages 66--98. Springer, Heidelberg,
  August 2023.

\bibitem[AL21]{EC:AnaLaP21}
Prabhanjan Ananth and Rolando~L. {La Placa}.
\newblock Secure software leasing.
\newblock In Anne Canteaut and Fran\c{c}ois-Xavier Standaert, editors, {\em
  EUROCRYPT~2021, Part~II}, volume 12697 of {\em {LNCS}}, pages 501--530.
  Springer, Heidelberg, October 2021.

\bibitem[BB20]{bennett2020quantum}
Charles~H Bennett and Gilles Brassard.
\newblock Quantum cryptography: Public key distribution and coin tossing.
\newblock {\em arXiv preprint arXiv:2003.06557}, 2020.

\bibitem[BDS23]{ben2023quantum}
Shalev Ben-David and Or~Sattath.
\newblock Quantum tokens for digital signatures.
\newblock {\em Quantum}, 7:901, 2023.

\bibitem[BGI{\etalchar{+}}01]{C:BGIRSVY01}
Boaz Barak, Oded Goldreich, Russell Impagliazzo, Steven Rudich, Amit Sahai,
  Salil~P. Vadhan, and Ke~Yang.
\newblock On the (im)possibility of obfuscating programs.
\newblock In Joe Kilian, editor, {\em CRYPTO~2001}, volume 2139 of {\em
  {LNCS}}, pages 1--18. Springer, Heidelberg, August 2001.

\bibitem[BGI14]{PKC:BoyGolIva14}
Elette Boyle, Shafi Goldwasser, and Ioana Ivan.
\newblock Functional signatures and pseudorandom functions.
\newblock In Hugo Krawczyk, editor, {\em PKC~2014}, volume 8383 of {\em
  {LNCS}}, pages 501--519. Springer, Heidelberg, March 2014.

\bibitem[BJL{\etalchar{+}}21]{TCC:BJLPS21}
Anne Broadbent, Stacey Jeffery, S{\'e}bastien Lord, Supartha Podder, and Aarthi
  Sundaram.
\newblock Secure software leasing without assumptions.
\newblock In Kobbi Nissim and Brent Waters, editors, {\em TCC~2021, Part~I},
  volume 13042 of {\em {LNCS}}, pages 90--120. Springer, Heidelberg, November
  2021.

\bibitem[BL20]{broadbent2019uncloneable}
Anne Broadbent and S{\'e}bastien Lord.
\newblock {Uncloneable Quantum Encryption via Oracles}.
\newblock 158:4:1--4:22, 2020.

\bibitem[BW13]{AC:BonWat13}
Dan Boneh and Brent Waters.
\newblock Constrained pseudorandom functions and their applications.
\newblock In Kazue Sako and Palash Sarkar, editors, {\em ASIACRYPT~2013,
  Part~II}, volume 8270 of {\em {LNCS}}, pages 280--300. Springer, Heidelberg,
  December 2013.

\bibitem[CG23]{EPRINT:ColGun23}
Andrea Coladangelo and Sam Gunn.
\newblock How to use quantum indistinguishability obfuscation.
\newblock Cryptology ePrint Archive, Paper 2023/1756, 2023.
\newblock \url{https://eprint.iacr.org/2023/1756}.

\bibitem[CGLZR23]{cryptoeprint:2023/410}
Alper Cakan, Vipul Goyal, Chen-Da Liu-Zhang, and João Ribeiro.
\newblock Unbounded leakage-resilience and intrusion-detection in a quantum
  world.
\newblock Cryptology ePrint Archive, Paper 2023/410, 2023.
\newblock \url{https://eprint.iacr.org/2023/410}.

\bibitem[CHV23]{chevalier2023semi}
C{\'e}line Chevalier, Paul Hermouet, and Quoc-Huy Vu.
\newblock Semi-quantum copy-protection and more.
\newblock In {\em Theory of Cryptography Conference}, pages 155--182. Springer,
  2023.

\bibitem[CLLZ21]{C:CLLZ21}
Andrea Coladangelo, Jiahui Liu, Qipeng Liu, and Mark Zhandry.
\newblock Hidden cosets and applications to unclonable cryptography.
\newblock In Tal Malkin and Chris Peikert, editors, {\em CRYPTO~2021, Part~I},
  volume 12825 of {\em {LNCS}}, pages 556--584, Virtual Event, August 2021.
  Springer, Heidelberg.

\bibitem[CMP20]{EPRINT:ColMajPor20}
Andrea Coladangelo, Christian Majenz, and Alexander Poremba.
\newblock Quantum copy-protection of compute-and-compare programs in the
  quantum random oracle model.
\newblock Cryptology ePrint Archive, Report 2020/1194, 2020.
\newblock \url{https://eprint.iacr.org/2020/1194}.

\bibitem[CV22]{Culf2022monogamyof}
Eric Culf and Thomas Vidick.
\newblock A monogamy-of-entanglement game for subspace coset states.
\newblock {\em {Quantum}}, 6:791, September 2022.

\bibitem[GGM84]{FOCS:GolGolMic84}
Oded Goldreich, Shafi Goldwasser, and Silvio Micali.
\newblock How to construct random functions (extended abstract).
\newblock In {\em 25th FOCS}, pages 464--479. {IEEE} Computer Society Press,
  October 1984.

\bibitem[Got02]{gottesman2002uncloneable}
Daniel Gottesman.
\newblock Uncloneable encryption.
\newblock {\em arXiv preprint quant-ph/0210062}, 2002.

\bibitem[GZ20]{EPRINT:GeoZha20}
Marios Georgiou and Mark Zhandry.
\newblock Unclonable decryption keys.
\newblock Cryptology ePrint Archive, Report 2020/877, 2020.
\newblock \url{https://eprint.iacr.org/2020/877}.

\bibitem[HILL99]{HILL99}
Johan H{\aa}stad, Russell Impagliazzo, Leonid~A. Levin, and Michael Luby.
\newblock A pseudorandom generator from any one-way function.
\newblock {\em {SIAM} Journal on Computing}, 28(4):1364--1396, 1999.

\bibitem[JMRW16]{JMRW16}
Nathaniel Johnston, Rajat Mittal, Vincent Russo, and John Watrous.
\newblock Extended non-local games and monogamy-of-entanglement games.
\newblock {\em Proceedings of the Royal Society A: Mathematical, Physical and
  Engineering Sciences}, 472(2189), 2016.

\bibitem[KPTZ13]{CCS:KPTZ13}
Aggelos Kiayias, Stavros Papadopoulos, Nikos Triandopoulos, and Thomas
  Zacharias.
\newblock Delegatable pseudorandom functions and applications.
\newblock In Ahmad-Reza Sadeghi, Virgil~D. Gligor, and Moti Yung, editors, {\em
  ACM CCS 2013}, pages 669--684. {ACM} Press, November 2013.

\bibitem[LLQZ22]{TCC:LLQZ22}
Jiahui Liu, Qipeng Liu, Luowen Qian, and Mark Zhandry.
\newblock Collusion resistant copy-protection for watermarkable
  functionalities.
\newblock In Eike Kiltz and Vinod Vaikuntanathan, editors, {\em TCC~2022,
  Part~I}, volume 13747 of {\em {LNCS}}, pages 294--323. Springer, Heidelberg,
  November 2022.

\bibitem[TFKW13]{Tomamichel_2013}
Marco Tomamichel, Serge Fehr, Jedrzej Kaniewski, and Stephanie Wehner.
\newblock A monogamy-of-entanglement game with applications to
  device-independent quantum cryptography.
\newblock {\em New Journal of Physics}, 15(10):103002, oct 2013.

\bibitem[Wie83]{wiesner1983conjugate}
Stephen Wiesner.
\newblock Conjugate coding.
\newblock {\em ACM Sigact News}, 15(1):78--88, 1983.

\bibitem[Wil11]{wilde2011classical}
Mark~M Wilde.
\newblock From classical to quantum shannon theory.
\newblock {\em arXiv preprint arXiv:1106.1445}, 2011.

\bibitem[WZ82]{wootters1982single}
William~K Wootters and Wojciech~H Zurek.
\newblock A single quantum cannot be cloned.
\newblock {\em Nature}, 299(5886):802--803, 1982.

\end{thebibliography}

\appendix
\crefalias{section}{appendix}

\section{Construction of Pseudorandom Function Copy-Protection}
\label{app:cp-prf}
We present below the construction of pseudorandom function copy-protection scheme of \cite{C:CLLZ21}, and show that it has anti-piracy security with respect to both the product and the identical distributions.

\paragraph{Construction.}
Let $n$ be a polynomial in $\secpar$; we define \(\ell_0, \ell_1, \ell_2\) such that \(n = \ell_0 + \ell_1 + \ell_2\) and \(\ell_2 - \ell_0\)
is large enough.
For this construction, we need three pseudorandom functions:
\begin{itemize}
  \item A puncturable extracting pseudorandom function \(\prf_1: \kspace_1 \times \bin^{n} \rightarrow \bin^m\)
    with error \(2^{-\secpar - 1}\) for min-entropy $n$, where \(m\) is a polynomial in \(\secpar\)
    and \(n \geq m +2\secpar + 4\).
  \item A puncturable injective pseudorandom function
    \(\prf_2: \kspace_2 \times \bin^{\ell_2} \rightarrow \bin^{\ell_1}\) with failure probability
    \(2^{-\secpar}\), with \(\ell_{1} \geq 2 \ell_{2} + \secpar\).
  \item A puncturable pseudorandom function \(\prf_3: \kspace_3 \times \bin^{\ell_1} \rightarrow \bin^{\ell_2}\).
\end{itemize}

\begin{construction}{Pseudorandom Function Copy-Protection}{cp-prf}
  \begin{itemize}[label=\(\),rightmargin=\leftmargin]
    \item \(\cpprotect(\secparam, \key)\):
    \begin{itemize}
      \item Sample \(\ell_0\) random coset states
      \(\{\ket{A_{i, s_i, s'_i}}\}_{i \in \intval{1, \ell_0}}\), where each
      subspace \(A_{i} \subseteq \FF_{2}^{n}\) if of dimension \(\frac{n}{2}\).
      \item For each coset state \(\ket{A_{i, s_i, s'_i}}\), prepare the obfuscated
      membership programs \(\mbr^0_i = \iO(A_i + s_i)\) and
      \(\mbr^1_i = \iO(A^\perp_i + s'_i)\).
      \item Sample \(\key_i \gets \prf_i .\keygen(1^\lambda)\) for \(i \in \{1, 2, 3\}\).
      \item Prepare the program \(\obfmbr \gets \iO(\progfont{P})\),
      where \(\progfont{P}\) is described in~\cref{fig:prog-p-cp-prf}.
      \item Return
      \(\qkey_\key \coloneqq \left(\{\ket{A_{i, s_i, s'_i}}\}_{i \in \intval{1, \ell_0}}, \obfmbr\right)\).
    \end{itemize}

    \item \(\eval(\secparam, \qkey_\key, x)\):
    \begin{itemize}
      \item Parse
      \(\qkey_\key = \left(\{\ket{A_{i, s_i, s'_i}}\}_{i \in \intval{1, \ell_0}}, \obfmbr\right)\).
      \item Parse \(x\) as \(x \coloneqq x^{(0)} \Vert x^{(1)} \Vert x^{(2)}\).
      \item For each \(i \in \intval{1, \ell_{0}}\), if \(x^{(0)}_{i} = 1\), apply
      \(\Hgate^{\otimes n}\) to \(\ket{A_{i, s_{i}, s'_{i}}}\); if
      \(x^{(0)}_{i} = 0\), leave the state unchanged.
      \item Let \(\sigma\) be the resulting state (which can be interpreted as a
      superposition over tuples of \(\ell_{0}\) vectors).
      Run \(\obfmbr\) coherently on input \(x\) and \(\sigma\), and
      measure the final output register to obtain \(y\).
      \item Return $y$.
    \end{itemize}
  \end{itemize}
\end{construction}

\begin{figure}[!htbp]
\centering
\begin{tcolorbox}[arc=0pt,outer arc=0pt,colframe=black,colback=white,
  boxrule=0.4pt,left=0pt,right=0pt,top=0pt,bottom=0pt,hbox]
  \begin{varwidth}[t]{0.8\textwidth}
    \textbf{Hardcoded:} Keys
    \((\key_{1}, \key_{2}, \key_3) \in \kspace_1 \times \kspace_2 \times \kspace_3\),
    programs \(\mbr_{i}^{0}, \mbr_{i}^{1}\) for all \(i \in \intval{1, \ell_{0}}\).

    On input \(x = x^{(0)} \Vert x^{(1)} \Vert x^{(2)}\) and vectors
    \(v_{0}, v_{1}, \cdots, v_{\ell_{0}}\) where each \(v_{i} \in \FF_{2}^{n}\), do the
    following:
    \begin{enumerate}
      \item (\textbf{Hidden Trigger Mode}) If
        \(\prf_{3}(\key_{3}, x^{(1)}) \oplus x^{(2)} = x^{(0)} \Vert \progfont{Q'}\) and
        \(x^{(1)} = \prf_{2}(\key_{2}, x^{(0)}\Vert \progfont{Q'})\): treat
        \(\progfont{Q'}\) as a classical circuit and output
        \(\progfont{Q'}(v_{1}, \cdots, v_{\ell_{0}})\).
      \item (\textbf{Normal Mode}) If for all \(i \in \intval{1, \ell_{0}}\),
        \(\mbr_{i}^{x_{i}}(v_{i}) = 1\), then output \(\prf_1(\key_1, x)\).
        Otherwise, output \(\bot\).
    \end{enumerate}
  \end{varwidth}
\end{tcolorbox}
\caption{Program \(\progfont{P}\).}
\label{fig:prog-p-cp-prf}
\end{figure}

\subsection{Proof of Indistinguishable Anti-Piracy Security With Respect to the Product Distribution}
In this subsection, we prove that the construction above has indistinguishable anti-piracy security with respect to the product distribution.
This proof and the next one (for the identical distribution) both adapt the proof of \cite[Theorem 7.12]{C:CLLZ21} to our settings; some parts are taken verbatim.
We first introduce some notations, a procedure and a lemma that we use in the two proofs.

\paragraph{Notations.}
In the proof, we sometimes parse \(x \in \bin^n\) as \((x^{(0)}, x^{(1)}, x^{(2)})\) such
that \(x = x^{(0)} \concat x^{(1)} \concat x^{(2)}\) (where \(\cdot \concat \cdot\) is the concatenation
operator) and the length of \(x^{(i)}\) is \(\ell_i\) for \(i \in \{0, 1, 2\}\).

We proceed with both proofs through a sequence of hybrids.
For any pair of hybrids \((\pcgamename_i, \pcgamename_j)\), we say that \(\pcgamename_i\) is \(\emph{negligibly close to}\) \(\pcgamename_j\) if for every triple of \(\qpt\) adversaries \((\adv, \bdv, \cdv)\), the probability that \((\adv, \bdv, \cdv)\) wins \(\pcgamename_i\) is negligibly close to the probability that they win \(\pcgamename_j\).

\paragraph{Procedure.}
We define the \(\genTrigger\) procedure (\cref{fig:gen-trigger}) which, given an
input's prefix \(x^{(0)}\) and a pseudorandom function image \(y\) returns a so-called \emph{trigger
  input} \(x'\) that: passes the ``Hidden Trigger'' condition of the program
  \(\progfont{P}\).
  Although this procedure also takes as input pseudorandom function keys $\key_2, \key_3$ and coset states descriptions, we will abuse notation and only write $\genTrigger(x^{(0)}, y)$ when it is clear from the context.
  We will also write $\genTrigger(x^{(0)}; \progfont{Q})$ - where $\progfont{Q}$ is a program - to denote the same procedure using $\progfont{Q}$ instead of the program normally defined in step 1.

  \begin{figure}
    \centering
    \begin{tcolorbox}[arc=0pt,outer arc=0pt,colframe=black,colback=white,
      boxrule=0.4pt,left=0pt,right=0pt,top=0pt,bottom=0pt,hbox]
      \begin{varwidth}[t]{0.8\textwidth}
        Given as input \(x^{(0)} \in \bin^{\ell_0}\), \(y \in \bin^m\),
        \(\key_2, \key_3 \in \kspace_2 \times \kspace_3\) and cosets
        \(\{A_{i, s_i, s'_i}\}_{i \in \intval{1, \ell_0}}\):
        \begin{enumerate}
          \item Let \(\progfont{Q}\) be the program which, given
            \(v_0, \dots, v_{\ell_0}\), returns \(y\) if
            \(R^{x_{0, i}}_i(v_i) = 1\) for all \(i\) or \(\bot\) otherwise.
            \item \(x'^{(1)} \gets \prf_2(\key_2, x^{(0)} \concat \progfont{Q})\);
            \item \(x'^{(2)} \gets \prf_3(\key_3, x'^{(1)}) \oplus (x^{(0)} \concat \progfont{Q})\);
            \item Return \(x^{(0)} \concat x'^{(1)} \concat x'^{(2)}\).
        \end{enumerate}
      \end{varwidth}
    \end{tcolorbox}
  \caption{\(\genTrigger\) procedure.}
  \label{fig:gen-trigger}
\end{figure}

\paragraph{Trigger's inputs lemma.}

The following lemma is taken from \cite[Lemma 7.17]{C:CLLZ21}.
\begin{lemma}
  \label{lem:7.17}
  Assuming post-quantum \(\iO\) and one-way functions, any efficient \(\qpt\)
  algorithm \(\adv\) cannot win the following game with non-negligible
  advantage:
  \begin{itemize}
    \item A challenger samples \(\key_1 \gets \setup(1^\secpar)\) and prepares a
    quantum key
    \(\rho_\key \coloneqq (\{\ket{A_{i,s_i,s'_i}}\}_{i \in \intval{1, l_0}} , \allowbreak \iO{(\progfont{P})})\)
    (recall that \(\progfont{P}\) has keys \(\key_1 , \key_2 , \key_3\)
      hardcoded).
    \item The challenger then samples a random input \(x_1 \gets \bin^n\); let \(y_1 \gets \prf_1 (\key_1 , x_1)\) and computes \(x'_1 \gets \genTrigger(x^{(0)}_{1}, y_1)\).
    \item Similarly, the challenger samples a random input \(x_2 \gets \bin^n\); let \(y_2 \gets \prf_1 (\key_1 , x_2)\) and computes \(x'_2 \gets \genTrigger(x^{(0)}_{2}, y_1)\).
    \item The challenger flips a coin \(b\), and sends either $\left(\rho_\key, x_1, x_2\right)$ or $\left(\rho_\key, x'_1, x'_2\right)$ to $\adv$, depending on the value of the coin.
  \end{itemize}
  \(\adv\) wins if it guesses \(b\) correctly.
\end{lemma}

\begin{gamedescription}[nr=-1,arg=]
  \describegame This is the piracy game with respect to the product distribution of the pseudorandom function copy-protection protocol.

  \begin{itemize}
    \item \textbf{Setup phase:}
    \begin{itemize}
      \item The challenger samples \(\ell_0\) random cosets \(\{A_{i}, s_i, s'_i\}_{i \in \intval{1, \ell_0}}\), and prepares the associated coset states \(\{\ket{A_{i, s_i, s'_i}}\}_{i \in \intval{1, \ell_0}}\) and the obfuscated membership programs \(\{(\mbr^0_i, \mbr^1_i)\}_{i \in \intval{1, \ell_0}}\).
      \item The challenger samples \(\key_i \gets \prf_i .\keygen(1^\lambda)\) for \(i \in \{1, 2, 3\}\) and generates the obfuscated program \(\obfmbr \gets \iO{(\progfont{P})}\).
      \item The challenger samples $x \sample \bin^n$ and computes $y := \prf_1(\key_1, x)$.
      \item Finally, the challenger sends \(\rho_\key := \left(\{\ket{A_{i, s_i, s'_i}}\}_{i \in \intval{1, \ell_0}}, \obfmbr\right)\) and $y$ to \(\adv\).
    \end{itemize}
    \item \textbf{Splitting phase:}
      $\adv$ prepares a bipartite quantum state $\sigma_{12}$, then sends $\sigma_1$ to $\bdv$ and $\sigma_2$ to $\cdv$.
    \item \textbf{Challenge phase (product distribution):}
    \begin{itemize}
      \item The challenger samples two bits $b_1, b_2 \sample \bin$, and two inputs $x_1, x_2 \sample \bin^n$.
      \item If $b_1 = 0$, the challenger sends $x$ to $\bdv$; otherwise, the challenger sends $x_1$.
      \item Similarly, if $b_2 = 0$, the challenger sends $x$ to $\cdv$; otherwise, the challenger sends $x_2$.
    \end{itemize}
    \item \textbf{Challenge phase (identical distribution):}
    \begin{itemize}
      \item The challenger samples a bit $b \sample \bin$, and an input $x_0 \sample \bin^n$.
      \item If $b = 0$, the challenger sends $x$ to both $\bdv$ and $\cdv$; otherwise, the challenger sends them $x_0$.
    \end{itemize}
    $\adv$, $\bdv$, and $\cdv$ win the game with respect to the product distribution if $\bdv$ returns $b_1$ and $\cdv$ returns $b_2$; and with respect to the identical distribution if both $\bdv$ and $\cdv$ return $b$.
  \end{itemize}

  \describegame In this game, we replace the challenges $x$, $x_0$, $x_1$, and $x_2$ by their trigger inputs for both $\bdv$ and $\cdv$.
  More precisely, the challenge phases become the following.

  \begin{itemize}
    \item \textbf{Challenge phase (product distribution):}
    \begin{itemize}
      \item The challenger samples two bits $b_1, b_2 \sample \bin$, and two inputs $x_1, x_2 \sample \bin^n$.
      \item The challenger computes the two images $y_1 := \prf(\key_1, x_1)$ and $y_2 := \prf(\key_1, x_2)$.
      \item If $b_1 = 0$, the challenger sends $\genTrigger(x^{(0)}, y)$ to $\bdv$; otherwise, the challenger sends $\genTrigger(x^{(0)}_1, y_1)$.
      \item Similarly, if $b_2 = 0$, the challenger sends $\genTrigger(x^{(0)}, y)$ to $\cdv$; otherwise, the challenger sends $\genTrigger(x^{(0)}_2, y_2)$.
    \end{itemize}
    \item \textbf{Challenge phase (identical distribution):}
    \begin{itemize}
      \item The challenger samples a bit $b \sample \bin$, and an input $x_0 \sample \bin^n$.
      \item The challenger computes the image $y_0 := \prf(\key_1, x_0)$.
      \item If $b = 0$, the challenger sends $\genTrigger(x^{(0)}, y)$ to both $\bdv$ and $\cdv$; otherwise, the challenger sends them $\genTrigger(x^{(0)}_0, y_0)$.
    \end{itemize}
  \end{itemize}

  The trigger's inputs lemma (\cref{lem:7.17}) implies that \(\currentgame\)
  is negligibly close to \(\previousgame\).

  \describegame In this game, we replace $y$ (in the setup phase) and $y_0, y_1, y_2$ (in the challenge phases) by uniformly random strings.
  Since all the inputs have enough min-entropy
  \(\ell_{1} + \ell_{2} \geq m + 2\secpar + 4\) and \(\prf_{1}\) is extracting, the images are statistically close to independently random bitstrings.
  Thus, \(\pcgamename_2\) is negligibly close to \(\pcgamename_1\).

  \describegame This game has exactly the same distribution as that of
  \(\previousgame\).
  We only change the order in which some values are sampled, and recognize
  that certain procedures become identical to encryption in the
  single-decryptor encryption scheme
  \(\scheme{\sd.\setup, \sd.\qkeygen, \sd.\enc, \sd.\dec}\)
  from~\cref{constr:sd}.
  Thus, the probability of winning in \(\currentgame\) is the same as in
  \(\previousgame\).

  \begin{itemize}
    \item \textbf{Setup phase:}
    \begin{itemize}
      \item The challenger runs \( \sd.\setup(1^\lambda)\) to obtain \(\ell_{0}\)
      random cosets \(\{A_{i}, s_i, s'_i\}_{i \in \intval{1, \ell_0}}\), the
      associated coset states \(\{\ket{A_{i, s_i, s'_i}}\}_{i \in \intval{1, \ell_0}}\) and
      the obfuscated membership programs
      \(\{(\mbr^0_i, \mbr^1_i)\}_{i \in \intval{1, \ell_0}}\).
      Let \(\rho_\sk := \{\ket{A_{i, s_i, s'_i}}\}_{i \in \intval{1, \ell_0}}\).
      \item The challenger samples \(\key_i \gets \prf_i .\keygen(1^\lambda)\) for \(i \in \{1, 2, 3\}\)
      and generates the obfuscated program
      \(\obfmbr \gets \iO{(\progfont{P})}\).
      \item The challenger samples \(y \sample \bin^{m}\)
      and sends \(\rho_\key := \left(\{\ket{A_{i, s_i, s'_i}}\}_{i \in \intval{1, \ell_0}}, \obfmbr\right)\) and $y$
      to \(\adv\).
    \end{itemize}
    \item \textbf{Splitting phase:}
      $\adv$ prepares a bipartite quantum state $\sigma_{12}$, then sends $\sigma_1$ to $\bdv$ and $\sigma_2$ to $\cdv$.

    \item \textbf{Challenge phase (product distribution):}
    \begin{itemize}
      \item The challenger samples two bits $b_1, b_2 \sample \bin$, and two inputs $x_1, x_2 \sample \bin^n$.
      \item The challenger also samples a random set of coins $r \sample \bin^{\poly}$ for the encryption, and two bitstrings $y_1, y_2 \sample \bin^m$.
      \item If $b_1 = 0$, the challenger computes $(x, \progfont{Q}) \gets \sd.\enc(\pk, y; r)$ and sends $\genTrigger(x^{(0)}; y)$ to $\bdv$; otherwise the challenger computes $(x_1, \progfont{Q}) \gets \sd.\enc(\pk, y_1; r)$ and sends $\genTrigger(x^{(0)}_1, y_1)$.
      \item Similarly, if $b_2 = 0$, the challenger computes $(x, \progfont{Q}) \gets \sd.\enc(\pk, y; r)$ and sends $\genTrigger(x^{(0)}, y)$ to $\cdv$; otherwise the challenger computes $(x_2, \progfont{Q}) \gets \sd.\enc(\pk, y_2; r)$ and sends $\genTrigger(x^{(0)}_2, y_2)$.
    \end{itemize}
    $\adv$, $\bdv$, and $\cdv$ win the game with respect to the product distribution if $\bdv$ returns $b_1$ and $\cdv$ returns $b_2$; and with respect to the identical distribution if both $\bdv$ and $\cdv$ return $b$.
  \end{itemize}
\end{gamedescription}

\paragraph{Reduction from single-decryptor’s piracy game for the product distribution.}
We reduce the game $\pcgamename_3$ with respect to the product distribution to the piracy game of the underlying single-decryptor with respect to the product distribution.
Assume that there exists a triple of \(\qpt\) adversaries \((\adv, \bdv, \cdv)\) who wins the last hybrid \(\game_{3}\) with respect to the product distribution with advantage \(\delta\).
We construct a $\qpt$ adversary \((\adv', \bdv', \cdv')\) who wins the piracy game of the single-decryptor scheme of \cref{constr:sd} with respect to the product distribution with the same advantage \(\delta\).

\begin{itemize}
  \item $\adv'$, on input a quantum key $\rho_\sk$ and the associated public key $\pk$:
  \begin{itemize}
    \item samples \(\key_i \gets \prf_i.\keygen(1^\lambda)\) for \(i \in \{1, 2, 3\}\) and use these keys and \(\pk\) to prepare the obfuscated program \(\obfmbr \gets \iO{(\progfont{P})}\);
    \item samples \(y, y_1, y_2 \sample \bin^m\);
    \item runs \(\adv\) on \((\rho_{\sk}, \obfmbr, y)\) to get $\sigma_{12}$;
    \item then prepares $\sigma'_1 := \sigma_1 \otimes \ketbra{\key_2, \key_3}{\key_2, \key_3}$ and $\sigma'_2 := \sigma_2 \otimes \ketbra{\key_2, \key_3}{\key_2, \key_3}$;
    \item and finally sends $\sigma'_1$ to $\bdv$, $\sigma'_2$ to $\cdv$, and the pairs of messages $(y, y_1)$, $(y, y_2)$ to the challenger.
  \end{itemize}
  \item $\bdv'$, on input $\sigma'_1$ and a ciphertext $(r, \progfont{Q})$:
  \begin{itemize}
    \item computes $x' \gets \genTrigger(r; \progfont{Q})$;
    \item runs $\bdv$ on $(\sigma_1, x')$ and returns the outcome.
  \end{itemize}
  \item $\cdv'$ is defined similarly as $\bdv'$ by replacing $\sigma'_1$ by $\sigma'_2$.
\end{itemize}

The adversary \((\adv', \bdv', \cdv')\) perfectly simulates \((\adv, \bdv, \cdv)\), and thus \((\adv', \bdv', \cdv')\) breaks the anti-piracy security of the single-decryptor scheme with the same probability \(\delta\), which
completes the proof.

\paragraph{Reduction from single-decryptor’s piracy game for the identical distribution.}
We reduce the game $\pcgamename_3$ with respect to the identical distribution to the piracy game of the underlying single-decryptor with respect to the identical distribution.
Assume that there exists a triple of \(\qpt\) adversaries \((\adv, \bdv, \cdv)\) who wins the last hybrid \(\game_{3}\) with respect to the identical distribution with advantage \(\delta\).
We construct a $\qpt$ adversary \((\adv', \bdv', \cdv')\) who wins the piracy game of the single-decryptor scheme of \cref{constr:sd} with respect to the identical distribution with the same advantage \(\delta\).

\begin{itemize}
  \item $\adv'$, on input a quantum key $\rho_\sk$ and the associated public key $\pk$:
  \begin{itemize}
    \item samples \(\key_i \gets \prf_i.\keygen(1^\lambda)\) for \(i \in \{1, 2, 3\}\) and use these keys and \(\pk\) to prepare the obfuscated program \(\obfmbr \gets \iO{(\progfont{P})}\);
    \item samples \(y, y_0 \sample \bin^m\);
    \item runs \(\adv\) on \((\rho_{\sk}, \obfmbr, y)\) to get $\sigma_{12}$;
    \item then prepares $\sigma'_1 := \sigma_1 \otimes \ketbra{\key_2, \key_3}{\key_2, \key_3}$ and $\sigma'_2 := \sigma_2 \otimes \ketbra{\key_2, \key_3}{\key_2, \key_3}$;
    \item and finally sends $\sigma'_1$ to $\bdv$, $\sigma'_2$ to $\cdv$, and the pairs of messages $(y, y_0)$, $(y, y_0)$ to the challenger.
  \end{itemize}
  \item $\bdv'$, on input $\sigma'_1$ and a ciphertext $(r, \progfont{Q})$:
  \begin{itemize}
    \item computes $x' \gets \genTrigger(r; \progfont{Q})$;
    \item runs $\bdv$ on $(\sigma_1, x')$ and returns the outcome.
  \end{itemize}
  \item $\cdv'$ is defined similarly as $\bdv'$ by replacing $\sigma'_1$ by $\sigma'_2$.
\end{itemize}

The adversary \((\adv', \bdv', \cdv')\) perfectly simulates \((\adv, \bdv, \cdv)\), and thus \((\adv', \bdv', \cdv')\) breaks the anti-piracy security of the single-decryptor scheme with the same probability \(\delta\), which
completes the proof.
 
\end{document}